\def\notes{0}
\def\pagenumtitle{1}
\def\mystatus{1} %

\def\focsversion{0}

\documentclass{article}

\usepackage{header}

\begin{document}

\pagestyle{empty}
\title{Local Node Differential Privacy\condelse{}{\thanks{To appear at the 67th IEEE Symposium on Foundations of Computer Science (FOCS) 2026.}}}

\ifnum\focsversion=1
    \author{\IEEEauthorblockN{
    Sofya Raskhodnikova\IEEEauthorrefmark{1}, 
    Adam Smith\IEEEauthorrefmark{1}, 
    Connor Wagaman\IEEEauthorrefmark{1}, 
    Anatoly Zavyalov\IEEEauthorrefmark{1}}
    \IEEEauthorblockA{\IEEEauthorrefmark{1}Department of Computer Science\\
    Boston University\\
    Boston, USA\\
    \{sofya,ads22,wagaman,zavyalov\}@bu.edu}
    \thanks{Full version of this paper is available at \href{https://arxiv.org/abs/2602.15802}{https://arxiv.org/abs/2602.15802}.}
    }
    
\else
    \author{Sofya Raskhodnikova$^\dagger$ \and Adam Smith$^\dagger$ \and Connor Wagaman$^\dagger$ \and Anatoly Zavyalov\thanks{Boston University, \texttt{\{sofya,ads22,wagaman,zavyalov\}@bu.edu.}}}
\fi

\date{\today}

\maketitle
\thispagestyle{empty}

\begin{abstract}
    We initiate an investigation of \emph{node differential privacy} for graphs in the \emph{local model} of private data analysis. In our model, dubbed \LNDP, each node sees its own edge list and releases the output of a local randomizer on this input. These outputs are aggregated by an untrusted server to obtain a final output.
    
    We develop a novel algorithmic framework for this setting that allows us to accurately answer arbitrary linear queries about the degree distribution 
    \ifnum\focsversion=0
        $\ddist_G$
    \fi
    of the input graph $G$. Our framework is based on a new object, called the {\em blurry degree distribution}, which closely approximates 
    \ifnum\focsversion=1
        the degree distribution
    \else
        $\ddist_G$
    \fi
    and has lower sensitivity. Instead of answering queries about  
    \ifnum\focsversion=1
        the degree distribution
    \else
        $\ddist_G$
    \fi
    directly, our algorithms answer related queries about the blurry degree distribution.
    This framework yields accurate \LNDP algorithms for the edge count, PMF and CDF of the degree distribution, and other graph statistics.
    For some natural problems, our algorithms match the accuracy achievable with node privacy in the \emph{central} model, where data are held and processed by a trusted server.
    
    We also prove lower bounds on the error required by \LNDP algorithms that imply the optimality of our framework for edge counting in sparse graphs and \ergraph parameter estimation.
    Our lower bounds apply even to interactive protocols with a constant number of rounds of interaction between the nodes and the server. Existing lower-bound techniques for related models either yield loose bounds or do not apply in our setting, 
    because graph data result %
    in inherently overlapping inputs to  local randomizers.
    To prove our bounds, we develop a \emph{splicing} argument that  stitches together views from locally similar but globally different distributions on graphs to obtain hard instances for the problem at hand.

    Finally, we prove structural results that reveal qualitative differences between local node privacy and the standard local model for tabular data.
\end{abstract}

\ifnum\focsversion=1
\else
\fi

\ifnum\focsversion=1

\else
    \vfill
    {
    \pagestyle{empty}
    \setstretch{1}
    \newpage
    \small
    {
    \hypersetup{linkcolor=black}
    \setcounter{tocdepth}{2}
    \tableofcontents
    }
    
    \newpage
    }
\fi

\ifnum\focsversion=1

\else
    \clearpage
\fi

\pagestyle{plain}      %
\pagenumbering{arabic} %
\setcounter{page}{4}
\renumstuff{\setcounter{page}{1}}   %

\section{Introduction}
\label{sec:intro}

Many modern graph datasets containing sensitive information, such as social 
networks, collaboration graphs, and contact-tracing graphs, are naturally distributed:  each node knows its neighbors but no single authority sees the whole graph.
Such datasets can yield valuable 
insights, %
but these benefits must be balanced with protecting the privacy of the individuals represented in the graph.

Differential privacy (DP) \cite{DworkMNS16} is the standard framework for enabling data analyses while protecting individuals’ information. 
For distributed data, a common adaptation of DP is the \emph{local} model, in which each client randomizes its data before it is collected; this model is  widely used in industry deployments of DP \cite{ErlingssonPK14,BittauEMMRLRKTS17,DingKY17,apple-iconic-scenes23}.
The local model has been developed for tabular data \cite{KasiviswanathanLNRS11} and, more recently, considered for graph data. Most prior DP work on graphs, however, has focused instead on the central model, 
where a trusted curator holds 
the full dataset. In that model, 
differential privacy for graphs 
has been extensively studied with two canonical variants: 
{\em edge privacy} 
\cite{NissimRS07},
which, intuitively,
hides whether a particular relationship is present, and  
{\em node privacy}
\cite{BlockiBDS13,KasiviswanathanNRS13,ChenZ13}, which hides
an individual's entire set of relationships.

Only edge DP has been studied in the local model so far (see \Cref{sec:unrelated} for related work), even though node DP is strongly motivated in many distributed graph settings where nodes represent individuals.
In such settings, the sensitive unit is often a node's  
entire neighborhood, and even aggregate information about that neighborhood can be highly revealing (for example, exposing sexual orientation based on one's social network connections~\cite{JerniganM09}). 
Node privacy 
provides a strong guarantee in such settings but is especially challenging to achieve  in the local model: each node must privatize its entire neighborhood, making aggregation tasks require fundamentally different algorithmic approaches than for edge DP.

We provide the first investigation of node DP in the local model, 
capturing 
the concerns and constraints of distributed networks in which nodes represent individuals.
We develop algorithmic and lower bound techniques for this model, designing accurate algorithms for graph statistics based on the degree distribution, in some cases with optimal error,
and prove structural results that reveal qualitative differences between local node privacy 
and the standard local model for tabular data.

\paragraph{Local node differential privacy (\LNDP) model}
We study the local analogue of node differential privacy for %
graphs. There are $n$ parties, each corresponding to a node and receiving its incident edge list as input. Each party runs a {\em local randomizer} on its own input, using both public and local randomness, and releases the output to an (untrusted) %
central server, which  
postprocesses all reports to estimate the desired statistic.
For privacy parameters $\eps$ and $\delta$, the overall algorithm is \emph{\edLNDP} (\Cref{def:lndp}) if the joint distributions over the outputs of all parties are
$(\eps,\delta)$-indistinguishable (\Cref{def:indistinguishable}) 
for 
every pair of node-neighboring input graphs---that is, undirected graphs that 
can be obtained from one another by rewiring a single node (i.e., they differ only in the edges incident to a single node).

The \LNDP model is \textit{noninteractive}.\footnote{We call our primary (noninteractive) model \LNDP to distinguish it from interactive LNDP.} We focus on this setting since it captures 
existing deployments of local DP and has been studied extensively for tabular data (e.g., in \cite{DuchiJW13,BassilyS15,BunNS19,EdmondsNU20,ChenGKM21,FeldmanMRT25,CanonneGS26}). 
We also define an \emph{interactive} version of the model (\Cref{def:int-lndp}), in the style 
of
\cite{KasiviswanathanLNRS11,JosephMNR19} for tabular data, and show that our main lower bounds extend even to interactive LNDP algorithms with a constant number of rounds of interaction between the nodes and the server.

\paragraph{Problem formulation} We study the error achievable by \LNDP algorithms. Some of our error guarantees are worst-case over all graphs, while others are conditional on  a promise or distributional assumption on the input. However, in all cases, as is standard in the literature, we require privacy for \emph{all} input graphs; assumptions on the input are only needed for accuracy.

\subsection{Our Contributions}
\label{sec:results}
Our main contributions are an \LNDP algorithmic framework for answering arbitrary linear queries about the degree distribution of the input graph, lower bound techniques for \LNDP, and structural results that reveal qualitative differences between local node privacy and the standard local model for tabular data.

\paragraph{Algorithmic framework for answering linear queries}
We develop a novel algorithmic framework that allows us to accurately answer arbitrary linear queries about the degree distribution $\ddist_G$ of the input graph $G$ under \LNDP.
Concretely,
for any ``workload'' matrix $M \in \mathbb{R}^{k \times n}$ of $k$ linear queries, we give an \LNDP algorithm for estimating $M \ddist_G$. Examples of important graph statistics that can be represented as answers to linear queries include the PMF and CDF of the degree distribution, the edge count, and the parameter $p$ of the \ergraph graph drawn from $G(n,p)$.

Privately releasing statistics based on $\ddist_G$ is challenging due to its high sensitivity: changing the edge list of %
one node can change all nodes' degrees. In the central model, node-DP approximations to $\ddist_G$ are obtained by using Lipschitz extensions or projections that carefully prune the graph until it satisfies a given degree bound 
and then privately release the degree distribution of the pruned graph. 
For example,
\cite{RaskhodnikovaS16} obtain a Lipschitz extension of $\ddist_G$ via quadratic %
programming, and \cite{DayLL16} try to insert edges of $G$ in the pruned graph in a fixed order, keeping only those that obey the degree bound for both endpoints.
Such approaches do not work in the local model, since the nodes lack the information needed
to compute their contributions. (As discussed later in this section, our impossibility results rule out this type of approach 
entirely.)

To overcome this challenge, we introduce an  approximation of the degree distribution %
that we call the \emph{blurry degree distribution}.
Instead of answering queries about  $\ddist_G$ directly, our algorithms answer related queries about %
the blurry degree distribution. 
The blurry degree distribution, denoted $\fakedistnew$, is parametrized by
$s\in \N$ and is 
a ``smooth'' discretization of $\ddist_G$
to multiples of $s$, where each node's degree is
\cw{represented as}
a convex combination of the nearest multiples of $s$. Crucially, $\fakedistnew$ has lower sensitivity than $\ddist_G$, and each node can compute its contribution to $\fakedistnew$ locally. (We describe $\fakedistnew$ %
in \Cref{sec:techniques}.) To state our results, %
the only property of $\fakedistnew$ we highlight is that it is close to $\ddist_G$ in Wasserstein-$\infty$ distance, i.e., $\winf\bparen{\fakedistintro,\ddist_G}\leq s$ (see \Cref{lem:blur-deg-props-new}). Intuitively, when we replace $\ddist_G$ by $\fakedistnew$, it shifts each node's contribution to the degree distribution by at most $s$, resulting in ``left-right'' error quantified by $\winf$.

In our algorithmic framework, the algorithms estimate $M\ddist_G$  for a workload matrix $M$ by privately releasing $M\fakedistnew$. As a result, we obtain a bicriterion error guarantee: a shift of at most $s$ in each node's degree (i.e., a ``left-right'' error from blurring) and an $\ell_\infty$ error from the added noise, where smaller $s$ reduces the $\winf$ error, while larger $s$ reduces the $\ell_\infty$ error. 
Our framework allows us to leverage existing \emph{factorization mechanism}-based methods for answering linear queries \cite{HardtT10,BhaskaraDKT12,LiMHMR15,NikolovT016,EdmondsNU20}, which 
reduce the $\ell_\infty$ error
when $M$  
can be represented as a product of two matrices $L$ and $R$ with low relevant norms. The general guarantee of our framework is stated next.
Up to a factor of $\sqrt{1 + n/s^2}$, our accuracy matches that of the factorization mechanism for tabular data in the standard local model \cite{EdmondsNU20}.

\begin{theorem}[Linear queries about $\fakedistnew$; \cref{thm:fact-mech-blurry} (informal version)]
\label{thm:lin-queries-inf}
    For all $s \in \mathbb{N}$ and matrices $M \in \mathbb{R}^{k \times n}$ of linear queries, there is an $(\eps, \del)$-\LNDP algorithm $\cA$ such that, for all graphs $G$ on node set $[n]$,
    $$
        \BEx
       \left \|
       \cA(G) 
       - 
       M \fakedistintro 
       \right \|_\infty
       = 
       \Otilde\Big(
       \|\matR\|_{\gamma_2}\cdot \sqrt{\frac{1}{n} + \frac{1}{s^2}} \cdot \frac{\sqrt{\log(1/\del)}}{\eps}
       \Big),
    $$
    where $\| M \|_{\gamma_2} := \min_{L,R} \{\normtwoinf{L} \normonetwo{R} : LR = M \}$, and $\|\cdot\|_{2 \to \infty}$ and $\|\cdot\|_{1 \to 2}$ denote the maximum $\ell_2$ norm of a row and column, respectively.
\end{theorem}

This theorem yields \edLNDP algorithms that estimate the PMF and CDF of the degree distribution with a bicriterion error guarantee. 
We use $M = \mathbb{I}_n$ 
for estimating the PMF; %
for the CDF, we use the lower-triangular %
matrix $M = (\indic_{i \geq j})_{i,j \in [n]}$, 
relying on known factorizations of this matrix
(e.g., \cite{HenzingerKU25}).
The accuracy guarantees for estimating general linear queries and the PMF/CDF are summarized in \Cref{tab:deg-dist}.

Linear queries on the degree distribution %
enable many graph estimation tasks. As summarized in Table~\ref{tab:results}, 
we obtain \edLNDP algorithms for counting edges in $D$-bounded graphs (i.e., with maximum degree at most $D$), estimating the parameter of an \ergraph graph, and estimating the size of a clique in a graph that consists of a large clique and isolated nodes. 
These algorithms' accuracy bounds apply directly to the original problem; they are not bicriterion guarantees.
In all cases, privacy holds for \emph{all} graphs, while accuracy is guaranteed on the specified classes of graphs, as is common in the %
literature.

\newcommand{\fnmark}[1]{%
  \hypertarget{fn:#1}{}%
  \hyperlink{fntext:#1}{\textsuperscript{#1}}%
}

\newcommand{\fntext}[2]{%
  \par\noindent
  \hypertarget{fntext:#1}{\textsuperscript{#1}}\,#2%
}

\begin{table*}[t!]
\centering
\renewcommand{\arraystretch}{1.2}
\setlength{\tabcolsep}{6pt}

\begin{tabular}{
    c || c | c | c
}
\toprule
\rowcolor{gray!30}
\textbf{Statistic} & 
\textbf{$\winf$ error} & 
\textbf{$\linf$ error} & 
\textbf{Reference} \\
\midrule

linear queries $M \in \R^{k \times n}$ about $\ddist_G$

&\multirow{3}{*}{$\slopelength$}
& 
\rule[-0.05em]{0pt}{\dimexpr\ht\strutbox+0.05em}
$\!\Otilde\!\paren{
       \| \matR \|_{\gamma_2} \cdot \sqrt{\frac{1}{n} \!+\! \frac{1}{s^2}} \cdot 
       \smallepsdel
       }\!$
& \Cref{thm:fact-mech-blurry} \\

\cellcolor{gray!10}
CDF of degree distribution
& 
& 
\rule[-0.1em]{0pt}{\dimexpr\ht\strutbox+0.1em}
$\widetilde{O}\!\paren{\sqrt{\frac{1}{n} + \frac{1}{\slopelength^2}}\cdot \smallepsdel}$
& %
\cellcolor{gray!10} \Cref{cor:cumu-deg-dist-new} \\

PMF of degree distribution
& 
& \cellcolor{gray!10}
\rule[-0.05em]{0pt}{\dimexpr\ht\strutbox+0.05em}
$\widetilde{O}\!\paren{\frac{1}{s} \sqrt{\frac{1}{n} + \frac{1}{\slopelength^2}}\cdot \smallepsdel}$
&  \Cref{cor:funny-pdf} \\

\bottomrule
\end{tabular}

\caption{
Error of our \edLNDP algorithms for estimating %
linear queries, PMF, and CDF of the degree distribution $\ddist_G$
of an $n$-node graph $G$. The $\widetilde O$ hides a polylog$(n)$ factor. 
We use $\| M \|_{\gamma_2}$ as in \Cref{thm:lin-queries-inf}. 
Approximating the degree distribution in the central model has been explored in \cite{RaskhodnikovaS16,DayLL16}, focusing on accuracy for graphs with maximum degree at most $D$, with $\ell_\infty$ error $\wt{O}(D/\eps)$.
}
\label{tab:deg-dist}
\end{table*}

\begin{table*}[t!]
\centering
\renewcommand{\arraystretch}{1.2}
\setlength{\tabcolsep}{6pt}

\rowcolors{2}{gray!10}{white}
\begin{tabular}{
    >{\columncolor{gray!10}}c
    |c
    ||
    c|
    c|
    c|
    >{\columncolor{gray!5}}c
}
\toprule
\rowcolor{gray!30}
\textbf{\small Statistic}
& {\small \twoline{\textbf{Accuracy}}{\textbf{promise}}}
& {\small \twoline{\textbf\edLNDP}{\textbf{upper bound}}} & 
{\small \twoline{$(\eps,\del)$-LNDP}{\textbf{lower bound}}} & {\small \twoline{\textbf{Central}}{\textbf{$(\eps,\del)$-node-DP}}} &
\textbf{\small \scalebox{0.9}{Significance}}
\\
\midrule
{\small edge count}
    & {\small \twoline{$D$-bounded}{graph}}
        & \twoline{\scalebox{0.9}{$O\paren{\paren{D\sqrt{n} + n} \cdot \smallepsdel}$}}{\footnotesize{(\cref{thm:edge-ct-alg})}}  
    & \twoline{$\Omega\paren{\frac{n}{\eps}}$}{\footnotesize{(\cref{thm:edge-ct-lb-approx})}}
    & \twoline{$\Theta\paren{\frac{D}{\eps}}$}{{
    \ifnum\focsversion=1
        \footnotesize
    \else
        \tiny
    \fi
    \cite{BlockiBDS13,KasiviswanathanNRS13,ChenZ13}
    }} 
    & {\small \twoline{local-central}{gap}} \\
{\small \twoline{$\ergraph$}{parameter $p$}}
    & {\small $G \sim G(n,p)$}
    & \twoline{$\Otilde\paren{\frac{1}{n}\cdot \smallepsdel}$}{\footnotesize{(\cref{thm:er-ub-new})}}
    & \twoline{$\Omega\paren{\frac{1}{n\eps}}$}{\footnotesize{(\cref{thm:er-lower-bd})}} 
    & \twoline{$\Theta\paren{\frac{1}{n} + \frac{1}{n^{3/2}\eps}}$}{
    {
    \ifnum\focsversion=1
        \footnotesize
    \else
        \tiny
    \fi
    \cite{BorgsCSZ18,SealfonU21,ChenDHS24}
    }} 
    & {\small
    \threeline{matches$^\dagger$}{statistical}{error}
    } \\
{\small clique size $k$}
    & {\small \twoline{clique +}{isolated nodes}}
    & \twoline{$O\paren{\frac{\sqrt{\log(1/\del)}}{\eps}}$}{{\footnotesize for $k=\Omega(n)$} \scalebox{0.9}{\footnotesize{(Thm.~\ref{thm:clique-new})}}}
    & \twoline{$\Omega\paren{\frac{1}{\eps}}$}{\scalebox{0.9}{\footnotesize{(same as central)}}}
    & \twoline{$\Theta\paren{\frac{1}{\eps}}$}{
    {
    \ifnum\focsversion=1
        \footnotesize
    \else
        \tiny
    \fi
    \cite{BlockiBDS13,KasiviswanathanNRS13,ChenZ13}}
    } 
    & {\small \twoline{local-central}{match}} \\
\bottomrule
\end{tabular}

\caption{Additive error of local and central node-DP algorithms, with corresponding lower bounds. 
All local algorithms
are noninteractive, private for all graphs, and have the indicated additive error under the accuracy promise. %
Lower bounds hold even for constant-round interactive LNDP. \\
$^\dagger$When $p$ and $\eps$ are constant
and $\delta=1/\poly(n)$, the error of the \LNDP algorithm for estimating the \ergraph parameter matches the error $\Theta(1/n)$ of nonprivate estimation (up to a $\polylog(n)$ factor).
}
\label{tab:results}

\end{table*}

\paragraph{Central-level accuracy in a local model}
A notable feature of our algorithmic framework is that, for some natural problems,
our additive error under \LNDP is nearly the same as the error required for solving these problems under \emph{central} node DP.
For estimating the parameter $p$ in $G(n,p)$, we recover the same $1/n$ behavior as 
for estimation without any privacy requirement, i.e., statistical error,
up to a $\polylog(n)$ factor
when $p$ and $\eps$ are constant
and $\delta=1/\poly(n)$.
For clique-size estimation, we match the central model's $\frac{1}{\eps}$ dependence. 
Matching the central model in this style
is impossible for the standard (tabular) local model.\footnote{\label{fn:amplification-by-shuffling}To see why, note that 
\textit{amplification by shuffling} \cite{CheuSUZZ19,ErlingssonFMRTT19,FeldmanMT21}, which states that shuffling the outputs of an $(\eps_0, \del_0)$-LDP algorithm yields an $(\eps_0/\sqrt{n}, n \del_0)$-DP algorithm in the central model, 
shows that any problem on tabular data (invariant under relabeling of the individuals) with $\Omega(\frac 1\eps)$ error in the central model must have $\Omega({\sqrt{n}/\eps})$ error in LDP.
Our clique size result, in particular, rules out an analogous shuffling result for $\LNDP$.
}

\paragraph{Impossibility and separation from the central model} 
We complement our algorithms with lower bounds that apply even to \emph{interactive} LNDP algorithms with a constant number of rounds.

For %
edge counting, we prove that 
every \edLNDP algorithm that is private on \emph{all} graphs and accurate on $D$-bounded graphs, for any $D\geq 1/\eps$, has additive error $\Omega(n/\eps)$. In particular, the $n$ term in our upper bound $O\bigl((D\sqrt{n}+n)\smallepsdel\bigr)$ is unavoidable, so our algorithm is optimal for the sparse regime $D = O(\sqrt{n})$.

This %
result highlights a fundamental difference between designing algorithms for %
local and central node DP. 
A common paradigm in the central model is to first create an algorithm that is private and accurate for some set of ``nice'' graphs (e.g., $D$-bounded graphs)---so that the error depends on $D$ (as in the $\Theta(D/\varepsilon)$ error bound for edge counting with node privacy)---and to then ``extend'' the algorithm to be private on all graphs while retaining accuracy on ``nice'' graphs, using tools such as Lipschitz extensions and stable projections (e.g., \cite{BlockiBDS13,,KasiviswanathanNRS13,ChenZ13,RaskhodnikovaS16,DayLL16,JainSW24}).
It is natural to think this strategy could also apply to local algorithms.
However, our edge-counting lower bound shows
such a design strategy breaks down in the local model.
This necessitates developing new algorithmic tools, which we describe in \cref{sec:techniques}.

We prove a similar lower bound for \ergraph parameter estimation: every \edLNDP algorithm for estimating $p$ in $G(n,p)$ must have additive error $\Omega(\frac 1 {n\eps})$, which matches (up to a $\sqrt{\log n}$ factor) the accuracy achieved by our algorithm.
Together, these lower bounds show that our algorithms are essentially the best one can hope for under local node privacy, even with a constant number of rounds of interaction.

\paragraph{Structural properties of \LNDP}
Finally, we uncover %
behavior
of \LNDP algorithms
that does not appear in the standard local (tabular) setting. 
We show that approximate \LNDP  (i.e., when $\delta>0$) is strictly more powerful than pure \LNDP (i.e., when $\delta=0$), in contrast with the result of \cite{BunNS19} showing that every noninteractive approximate LDP algorithm can be simulated by a noninteractive pure LDP algorithm. Specifically, we prove an \emph{advanced grouposition} property for pure \LNDP: if two graphs differ in the incident edges of $k$ nodes, then an \ezLNDP algorithm produces
$\bparen{O\bparen{k \eps^2 + \eps\sqrt{k\log(1/\del)}}, \del}$-indistinguishable outputs (\cref{thm:adv-grp}), for all $\delta>0$. Thus group privacy degrades like $O(\eps\sqrt{k})$, %
preventing pure \LNDP algorithms from distinguishing cliques of sizes differing by about $1/\eps^2$.
In contrast, our \edLNDP clique size estimation algorithm distinguishes cliques whose sizes differ by $\Theta_\del(1/\eps)$, %
so %
approximate \LNDP %
is strictly more powerful.

We also separate {\em degrees-only} \LNDP algorithms---a powerful class that includes all our algorithms described in \Cref{tab:deg-dist,tab:results}---from unrestricted \LNDP algorithms.
In a degrees-only algorithm, each node's randomizer sees only that node's degree instead of its full edge list.
We consider 
two natural input
distributions: random $t$-regular graphs and random $t$-{\em starpartite} graphs. A graph is $t$-{\em starpartite} if it has $t$ {\em star center} nodes that are adjacent to every node, and has no other edges (see \cref{def:starpartite-graph}).
We show that unrestricted $(\eps,\del)$-\LNDP algorithms can distinguish these %
distributions for some $t = O_\delta(1/\eps^6)$, whereas any degrees-only $(\eps,\del)$-\LNDP algorithm needs $t=\Omega(\sqrt{n}/\eps)$. Since the two distributions can be %
easily distinguished non-privately 
based on 
their degree sequences---for example, by checking for
a node of degree $n-1$---the gap comes 
from the privacy constraint. Unrestricted \LNDP algorithms are thus strictly more powerful.

These structural results show that \LNDP is not simply LDP with a different adjacency relation: it has its own group privacy behavior, a separation
of local views
(i.e., degrees-only versus %
full edge lists), 
no general amplification by shuffling (see \Cref{fn:amplification-by-shuffling}), and a separation of pure and approximate \LNDP.

\subsection{Our Techniques}
\label{sec:techniques}

Our algorithms, lower bounds, and structural results require the development of new techniques specific to \LNDP. Unlike in the local model for tabular data, the inputs of the parties in an \LNDP computation necessarily overlap (e.g., each edge appears in the views of both endpoints).
The bulk of the technical challenges in \LNDP stem from this overlap.

\subsubsection{
Linear Queries about the Degree Distribution}

Recall that our \LNDP algorithmic framework allows us to estimate $M \ddist_G$, where $M\in \R^{k\times n}$ is a matrix of linear queries and $\ddist_G$ is the degree distribution of the input graph $G$.
Two key ideas help us %
achieve
good accuracy:
(1) working %
with the blurry degree distribution $\fakedistnew$ 
that has lower sensitivity than $\ddist_G$,
and (2) scaling %
noise to the $\ell_2$, rather than $\ell_1$, sensitivity of the vector of per-node outputs.\footnote{It is notable that $\ell_2$ sensitivity helps in our
noninteractive local model:
in contrast,
\cite{BunNS19} show that, in the noninteractive (tabular) local model, every approximate DP algorithm can be simulated by a pure DP algorithm. 
However, working with approximate DP, which permits using $\ell_2$-sensitivity, is necessary for achieving our error guarantees. For example, 
we obtain error for clique-size estimation that, by
our advanced grouposition result (\cref{sec:adv-grp}), %
is not achievable by
pure-\LNDP algorithms.}

To explain the framework, we first work with the actual degree distribution $\ddist_G$. We can represent  $\ddist_G$ as an average of the indicator vectors $e_{d_i}\in\R^n$, where $d_i$ is the degree of node $i$. By linearity, $M\ddist_G=\frac{1}{n}\sum_{i\in[n]} Me_{d_i}$.
A natural approach then is for each node to release a noisy version of $Me_{d_i}$ and for the server to average all node contributions. To satisfy \LNDP, the noise must scale with the sensitivity of the full vector $(Me_{d_1},\ldots,Me_{d_n})$. This sensitivity is large, under both the $\ell_1$ and $\ell_2$ norms:
rewiring one node can change the degree of every node, and hence all vectors $e_{d_i}$. Specifically, the $\ell_2$ sensitivity is $\Omega(\|M\|_{1\to 2}\sqrt{n})$, resulting in an overall error of $\Omega(\|M\|_{1\to 2})$ after averaging.
For many tasks, this is too large. 

Our main idea is to replace $\ddist_G$ with the blurry degree distribution $\fakedistnew$ and estimate $M\fakedistnew$ instead. To construct $\fakedistnew$, each node encodes its degree not as the indicator vector $e_{d_i}$, but as a convex combination $\tilde e_{d_i}$ of the indicator vectors for the two multiples of $s$ closest to $d_i$ (see \cref{fig:blur-deg-dist}).
Specifically, %
each node constructs the {\em blurry vector}
\begin{conddisplay}
    \tilde e_{d_i} = \paren{1-\set{\frac{d_i}{s}}} e_{s\floor{d_i/s}} + \paren{\set{\frac{d_i}{s}}} e_{s\ceil{d_i/s}},
\end{conddisplay}
where $\set{x} = x-\floor{x}$ denotes
the fractional part of
$x$. 
The
blurry degree distribution is 
the average of these blurry vectors, 
\begin{conddisplay}
    \fakedistnew = \frac{1}{n} \sum_{i\in[n]} \tilde e_{d_i}.
\end{conddisplay}

The blurry degree distribution $\fakedistnew$ is a close approximation of $\ddist_G$ in two key ways. %
First, it preserves the average degree, since each %
$\tilde e_{d_i}$ %
weighs the two multiples of $s$ nearest to $d_i$ so 
that the
weighted average is %
$d_i$. Second, it is close to $\ddist_G$ in $\winf$ distance, since mass on degree $d_i$ is only shifted by at most $s$. %
Crucially, it has low sensitivity: while the contribution $M\tilde e_{d_{i^*}}$ of the rewired node $i^*$ %
can still change by $\Theta(\|M\|_{1\to 2})$, each other node's contribution 
changes by at most $O(\|M\|_{1\to 2} \cdot \frac{1}{s})$, since consecutive blurry vectors satisfy $\| \tilde e_{d} - \tilde e_{d+1} \|_1 \leq 2/s$.
Thus, the overall $\ell_2$ sensitivity of the vector of contributions is $O\paren{\| M\|_{1\to 2} \cdot \sqrt{1 + \frac{n}{s^2}}}$.\footnote{The $\ell_1$ sensitivity is still $\Omega\paren{\| M\|_{1\to 1}\cdot (1 + \frac{n}{s})}$; scaling noise to this gives larger error for the tasks of interest.} Adding Gaussian noise scaled to this sensitivity gives an estimator for $M\fakedistnew$ with additive error $\tilde O\paren{\|M\|_{1\to2}\cdot \sqrt{\frac{1}{n} + \frac{1}{s^2}}\cdot \smallepsdel}$, which is the bound stated in \cref{thm:lin-queries-inf}, except for the dependence on $M$.
A notable feature of our techniques is that they allow us to leverage existing work on factorization mechanisms, %
leading to the final error bound in the theorem.

\begin{figure*}[t]
    \centering

    \begin{minipage}{0.49\textwidth}
        \ifnum\focsversion=1
            \centering
        \fi
        \usetikzlibrary{arrows.meta}

\begin{tikzpicture}[
    >=Latex,
    axis/.style={->, thick},
    tick/.style={very thin},
    every node/.style={font=\small},
    xscale=0.83,
    yscale=0.83
]

\def\xks{1.8}
\def\xdi{3.5}
\def\xkps{5.6}
\def\yone{3}
\def\yomf{2.2}
\def\yfrac{0.95}

\draw[axis] (-0.2,0) -- (8,0);
\node[below=2pt] at (7.5,0) {degree};
\draw[axis] (0,-0.15) -- (0,4);
\node[above] at (-0.8,3.5) {weight};

\draw[gray] (\xks,0)   -- (\xks,\yomf);
\draw[gray] (\xdi,0)   -- (\xdi,\yone);
\draw[gray] (\xkps,0)  -- (\xkps,\yfrac);
\draw[dashed, gray!40] (0,\yone)  -- (\xdi,\yone);
\draw[dashed, gray!40] (0,\yomf)  -- (\xks,\yomf);
\draw[dashed, gray!40] (0,\yfrac) -- (\xkps,\yfrac);

\draw[-{Stealth[scale=1.3]}, gray!70, thick]
  (\xdi-0.25,\yone)
    .. controls (\xdi-0.85,\yone) and (\xks+1,\yomf) ..
  (\xks+0.2,\yomf);

\draw[-{Stealth[scale=1.3]}, gray!70, thick]
  (\xdi+0.25,\yone)
    .. controls (\xdi+0.85,\yone) and (\xkps-1.2,\yfrac) ..
  (\xkps-0.2,\yfrac);

\filldraw[draw=black, fill=red!60,  line width=0.6pt]  (\xks,\yomf)  circle (3.5pt);
\filldraw[draw=black, fill=blue!60, line width=0.6pt]    (\xdi,\yone)   circle (3.5pt);
\filldraw[draw=black, fill=red!60,  line width=0.6pt]  (\xkps,\yfrac) circle (3.5pt);

\node[below] at (\xks,0)  {$\textstyle s \lfloor \frac{d_i}{s} \rfloor $};
\node[below] at (\xdi,0)  {$d_i$};
\node[below] at (\xkps,0) {$\textstyle s \lceil \frac{d_i}{s} \rceil$};

\node[left] at (0,\yone)  {$1$};
\node[left] at (0,\yomf)
  {$\textstyle 1 - \{\frac{d_i}{s}\}$
  };
\node[left] at (0,\yfrac)
  {$\textstyle \{\frac{d_i}{s}\}$ 
  };

\end{tikzpicture}
    \end{minipage}
    \begin{minipage}{0.49\textwidth}
        \ifnum\focsversion=1
            \centering
        \fi
        \usetikzlibrary{arrows.meta}
\begin{tikzpicture}[
    >=Latex,
    axis/.style={->, thick},
    every node/.style={font=\small},
    xscale=0.75,
    yscale=0.75,
]

\def\topBase{2.75}       %
\def\topScale{0.6}      %
\def\topYMax{1.6}       %
\def\botScale{0.75}      %
\def\botYMax{1.6}       %
\def\botStart{0.3}      %
\def\botBarW{0.4}      %
\def\botStep{1.6}       %

\foreach \x/\h in {
    0.2/1.2,  0.6/1.5,  1.0/1.1,  1.4/0.9,  1.8/0.5,
    2.2/0.3, 2.6/0.85, 3.0/0.6,  3.4/1.6,  3.8/1.85,
    4.2/1.35, 4.6/0.95, 5.0/0.85, 5.4/1.0,  5.8/0.7,
    6.2/1.1,  6.6/2.0,  7.0/1.85, 7.4/1.75
}{
    \draw[fill=blue!50, draw=black, line width=0.6pt] (\x, \topBase) rectangle ++(0.4, \h*\topScale);
}

\draw[axis] (-0.2,\topBase) -- (8.2,\topBase);
\node[below=2pt] at (8.25,\topBase) {degree};
\draw[axis] (0,\topBase-0.2) -- (0,\topBase+\topYMax);
\node[above left, align=right] at (-0.1,\topBase+\topYMax-0.7) {$\ddist_G$};

\draw[-{Stealth[scale=1.2]}, thick, black] (4.0, \topBase-0.35) -- (4.0, 1.3);
\node[left=4pt, black, align=right] at (4.0, {(\topBase-0.35+1.3)/2}) {apply blurring};

\foreach \i/\h/\lbl in {
    0/1.10/$0$,
    1/0.55/$s$,
    2/1.35/$2s$,
    3/0.8/$3s$,
    4/1.65/$4s$
}{
    \pgfmathsetmacro\xpos{\botStart + \i*\botStep}
    \draw[fill=red!50, draw=black, line width=0.6pt] (\xpos, 0) rectangle ++(\botBarW, \h*\botScale);
    \node[below] at (\xpos + 0.5*\botBarW, 0) {\lbl};
}

\draw[axis] (-0.2,0) -- (8.2,0);
\node[below=2pt] at (8.25,0) {degree};
\draw[axis] (0,-0.2) -- (0,\botYMax);
\node[above left, align=right] at (-0.1,\botYMax-0.7) {$\fakedistnew$};

\end{tikzpicture}
    \end{minipage}
    
    \caption{
        The left figure shows how to convert the indicator vector $e_{d_i}$ for degree $d_i$ to its corresponding blurry vector $\tilde e_{d_i}$, a %
        convex combination of the indicator vectors for the two multiples of $s$ nearest $d_i$, where $\set{x} = x - \floor{x}$ denotes the fractional part of $x$. 
        The right figure shows the blurry version (bottom) of the true degree distribution (top), obtained by
        averaging the blurry vectors of all
        nodes' degrees.
    }
    \label{fig:blur-deg-dist}
\end{figure*}

Our blurring technique allows us, for $s \geq \sqrt{n}$, to add noise scaled as if rewiring one node only affected that node's contribution to the output; changes induced by all other nodes are lower-order terms. For this regime of $s$, the error of \cref{thm:lin-queries-inf} thus matches the error required for answering linear queries 
in the standard local model \cite{EdmondsNU20}.

\paragraph{Counting edges, and estimating \ergraph parameters and clique sizes}

Our framework yields \LNDP algorithms with optimal or near-optimal error for counting edges in sparse graphs, estimating \ergraph parameters, and estimating the size of the clique in a graph consisting of a clique and isolated nodes. These algorithms rely on a subroutine (\Cref{lem:conc-deg}) that, for a graph whose nonzero degrees lie in an unknown interval of width at most $s$, estimates the average degree, scaled by $k/n$, of the $k$ nodes falling in that interval with error
$O\bparen{(1+\frac{s}{\sqrt{n}})\cdot\smallepsdel}$. 
(Such an estimator is most immediately useful when $k$ is known, as is the case for $D$-bounded graphs and \ergraph graphs, where $k=n$.) It does so by first estimating the blurry PMF and then using the value with largest mass as an anchor point to locate all nonzero-degree mass (which falls in a width-$O(s)$ interval by assumption and since $\winf(\fakedistnew, \ddist_G) \leq s$).
The
average degree
is then recovered as an appropriately scaled  weighted combination of the anchor point and the nearby PMF masses.
Since the blurry distribution preserves the average degree, this incurs only $\ell_\infty$ error from the PMF estimate---not the bicriterion error of our general framework---and yields an estimate for $\frac{k}{n}$ times the
average degree
with error
$O\bparen{(1+\frac{s}{\sqrt{n}})\cdot\smallepsdel}$.

Each application reduces to this subroutine since the relevant graph families' nonzero degrees are concentrated in a narrow interval: $D$-bounded graphs' degrees are in a width-$D$ interval, yielding error $O\bparen{(D\sqrt{n}+n)\cdot\smallepsdel}$ on the edge count; and $G(n,p)$ graphs' degrees are in a width-$\Otilde(\sqrt{n})$ interval w.h.p., yielding error $\Otilde\paren{\frac{1}{n\eps}}$ on the parameter estimate. While the number of nonzero-degree nodes $k$ is not known for cliques, because each node either has degree 0 or $k-1$, with some additional algebra we can recover the average degree, and thus the clique size, with error $O({\smallepsdel})$.

\subsubsection{Impossibility Results via Splicing}
\label{sec:techniques-imposs}

Existing lower-bound frameworks for tabular data in the local model \cite{BeimelNO08,DuchiJW13,BassilyS15,JosephMNR19} 
break down when working with graphs since edge lists held by different nodes 
necessarily overlap.\footnote{
\cite{EdenLRS25} gives a lower bound on triangle counting (later extended by \cite{SuppakitpaisarnPHH25} to subgraph counting) under noninteractive local edge-DP that does not follow via reduction from a local DP impossibility result, but that work is specific to edge privacy.
}
We develop a new lower-bound technique tailored to \LNDP.
One 
key result, \Cref{lem:hell-empty-dreg}, is that an empty $n$-node graph $\gempty$ and a random $d$-regular $n$-node graph (whose distribution is denoted $\gdreg$) are indistinguishable by \LNDP algorithms 
for $d$ up to about $1/\eps$.
That is,
\begin{equation}
\label{eq:intro-tv-splice-new}
    \BEx_{G\sim \gdreg} \Big[
        \dtv\Bparen{ \bparen{G,\cA(G)}, \bparen{G,\cA(\gempty)} }
    \Big]
    =
    O\paren{d\eps}
\end{equation}
for every $(\eps,\del)$-\LNDP algorithm $\cA$ (when $\del = O(d\eps^2)$).
Thus, even when given the random $d$-regular graph
$G$, an analyst seeing $\cA$'s output cannot reliably tell whether $\cA$ was run on $G$ or the empty graph (when $d = c/\eps$ for sufficiently small constant $c > 0$).
Since these
graphs' edge counts differ by $dn/2$,
taking
$d=c/\eps$ immediately gives our    
$\Omega(n/\eps)$ lower bound for counting edges in sparse graphs.
A generalization to symmetric distributions on $d$-bounded graphs (\Cref{lem:symmetric-graphs-to-starpartite}) yields the lower bound for \ergraph  parameter estimation. 

To show \cref{eq:intro-tv-splice-new}, which compares empty and $d$-regular graphs, we go through a third family: %
\emph{$d$-starpartite} graphs. \az{Recall that, in a $d$-starpartite graph, $d$ \emph{star center} nodes have edges to all other nodes, and there are no other edges.} %
This family has two key properties.
First, every $d$-starpartite graph is at node distance $d$ from the empty graph, so group privacy immediately gives a bound on
the distance between $\cA$'s output
distributions on starpartite and empty graphs.
Second, when the $d$ star centers are chosen uniformly at random, the neighborhood of each non-center node is a uniformly random set of $d$ other nodes---exactly the same distribution as it would have 
in a random $d$-regular graph.

Our ``splicing'' argument uses this local similarity to transfer a bound on the distance between $\cA$'s output distributions on $d$-starpartite and empty graphs to a bound on the distance between $\cA$'s output distributions on $d$-regular and empty graphs. (We dub the approach ``splicing'' since it involves stitching together views from locally similar but globally very different distributions.)

We establish our bound on TV distance by working with Bhattacharyya distance ($\bhatt$), which enjoys a \emph{tensorization} property (i.e., the $\bhatt$ between product distributions is the sum of the $\bhatt$ between each coordinate of the product distributions).
Let $\vcR(G)$ denote the vector of reports produced by $\cA$'s local randomizers. 
For any fixed graph $G$, this is a product distribution, so by tensorization,
$\bhatt(\vcR(G), \vcR(\gempty)) = \sum_{i=1}^n \bhatt(\cR_i(N_i^G), \cR_i(\varnothing))$,
where $N_i^G$ is %
the
neighborhood of
node $i$ in graph $G$.
For every node $i$,
\azsout{the splicing argument is}
\begin{align*}
\cond{&}
\BEx_{G\sim (d\text{-regular})} 
    \Big[
        \bhatt(\cR_i(N_i^G), \cR_i(\varnothing))
    \Big]
\condelse{
    \\ & \;\underbrace{\quad = \quad}_{\substack{\text{identically} \\ \text{distributed} \\ \text{views}}}\;
}{
    \;\;\;\underbrace{\quad = \quad}_{\substack{\text{identically} \\ \text{distributed} \\ \text{views}}} \;\;\;&
}
\BEx_{G\sim (d\text{-starpartite})} 
\Big[
    \bhatt(\cR_i(N_i^G), \cR_i(\varnothing))
    \Bigm| \text{$i$ non-center}
\Big] 
\condelse{
    \\
    &\underbrace{\quad \leq \quad}_{d \;\le \; n/2} \;\;\;
}
{
    \\
    \underbrace{\quad \leq \quad}_{d \;\le \; n/2} \quad \;&
}
2 \cdot  
\BEx_{G\sim (d\text{-starpartite})} 
\Big[
    \bhatt(\cR_i(N_i^G), \cR_i(\varnothing))
\Big].
\end{align*}
This is the splicing step: node $i$'s expected contribution to the distance from the empty graph is at most twice as large under the $d$-regular distribution as under the $d$-starpartite distribution.
Thus, the exceptional center nodes cost only a factor of two; all other nodes have the same view distribution as in a random $d$-regular graph.
Combining the splicing step %
with the tensorization of $\bhatt$ and linearity of expectation yields
\begin{align*}
\cond{&}
\BEx_{G\sim (d\text{-regular})}
\Big[
\bhatt\Bparen{\vcR(G), \vcR(\gempty)} 
\Big]
\condelse{\\ &\qquad\;\;\quad \leq \quad\;\;}{\;\; \leq \;\;}
2 \cdot 
\BEx_{G\sim (d\text{-starpartite})}
    \Big[
    \bhatt\Bparen{\vcR(G), \vcR(\gempty)} 
    \Big]
\cond{\\ & \qquad}
\;\;\underbrace{\quad = \quad}_{\substack{\text{group} \\ \text{privacy}}} \;\;
 O(d^2 \eps^2 + d \delta).
\end{align*}
Since $\cA$ is a postprocessing of $\vcR(\cdot)$, for $\del = O(d\eps^2)$ converting $\bhatt$ to TV distance gives \cref{eq:intro-tv-splice-new}.

This lower-bound argument is specific to the local model: even for $d=1$, central node-DP algorithms can distinguish empty and random regular graphs (e.g., via the maximum matching size).
It is also \emph{not} 
generally true that ``similar per-node views imply indistinguishability''---we show in \Cref{sec:unrestricted-star-vs-reg} that slightly denser
random regular graphs \emph{are} distinguishable from similarly dense random starpartite graphs.

In \cref{sec:int-imposs}, we lift %
these noninteractive lower bounds to the interactive setting.
Namely, because our TV bounds hold even when the graph $G$ is revealed to the distinguisher, we %
extend these bounds via a round-by-round hybrid argument to %
show that the TV bound increases by at most a factor of $\ell$ for a protocol with $\ell$ rounds of interaction.
Thus, solving these problems with constant-round LNDP algorithms requires the same asymptotic error as (noninteractive) \LNDP algorithms.

\subsubsection{Structural Results on \texorpdfstring{\LNDP}{\LNDP}}
\label{sec:structure-techniques}

\paragraph{Advanced grouposition for pure \texorpdfstring{\LNDP}{LNDP} algorithms \texorpdfstring{(\cref{sec:adv-grp})}{ }}
\Cref{thm:adv-grp}, on ``advanced grouposition'', demonstrates  a separation between pure %
and approximate \LNDP.
Analogously to the proof of advanced group privacy \cite{BunNS19} in the usual local model, the proof of \Cref{thm:adv-grp} considers how the contribution to the privacy loss---that is, the log of the ratio of the probabilities of a given randomizer's output under two different graphs---from each of the randomizers adds up as we rewire $k$ nodes in the graph. The argument is delicate since each rewiring may affect all 
nodes' inputs. The key insight is that the rigid constraints of pure \LNDP allow us to bound the sum of \textit{absolute values of the privacy losses} due to each of the changes by $O(\eps)$. 
(Crucially, this type of bound fails for approximate \LNDP.) We %
then argue that the expected values of (almost all of) these privacy losses are %
very small---about $O(\eps/\sqrt n)$. With additional work, this leads to the final bound.

\paragraph{Separating degrees-only and unrestricted \LNDP algorithms (\Cref{sec:dist-stars-regular})}
Our (unrestricted) \LNDP algorithm for distinguishing random $t$-regular and $t$-starpartite graphs is powered by the observation that nodes' neighborhoods in a starpartite graph are very similar 
(with the exception of star centers, each node has the same neighborhood), while nodes' neighborhoods in a random regular graph are almost entirely different. 
At a high level, our algorithm uses public randomness to 
generate $t$ random sets of nodes
$S_1,\dots,S_t$ of size $\Theta\bparen{\frac{n}{t}}$. For each set $S_j$, every node $i$ (noisily) reports a bit answering the query, ``Do you have at least one neighbor in set $S_j$?'' 
The vector of these bits can be thought of as a noisy locality-sensitive hash of each node's edge list.
In a random starpartite graph, these bits are highly correlated, while in a random regular graph they are roughly independent. Although
these bits must be released with 
considerable noise, there is enough signal in their correlation to reliably distinguish $t$-starpartite from $t$-regular graphs when %
$t\geq\poly(\log(1/\delta)/\eps)$, independent of $n$.

In contrast with this algorithmic result, we show that \textit{degrees-only} \LNDP algorithms cannot distinguish random $t$-regular and $t$-starpartite graphs %
for $t = o %
\bparen{\frac{\sqrt{n}}{\eps}}$. 
This also shows that degrees-only algorithms cannot estimate the number of edges with error $o(\frac{n \sqrt{n}}{\eps})$, pinning down the error of edge counting
for this special class of algorithms. 
We prove the lower bound, \cref{thm:od-indist}, by reducing from %
bit summation under standard LDP, using a novel ``hard distribution'' that mimics the correlation structure in a graph's degrees.

\subsection{Related Work}
\label{sec:unrelated}

We draw most heavily from work on local privacy and central-model node privacy. We briefly review key results in these areas as well as seemingly related lines of work that differ from ours in significant ways.

DP for graphs---in the setting of edge privacy---was introduced by \cite{NissimRS07}. The first nontrivial node-private algorithms appeared concurrently in \cite{BlockiBDS13,KasiviswanathanNRS13,ChenZ13}: they achieved accuracy under a structural promise (e.g., bounded maximum degree) and then used Lipschitz extensions and projections to extend privacy to all graphs while retaining accuracy on instances satisfying the promise.
This paradigm underlies most subsequent work on {\em node privacy,} which is powered by Lipschitz extensions \cite{RaskhodnikovaS16,DayLL16,BorgsCSZ18,CummingsD20,KalemajRST23} and projections \cite{DayLL16,JainSW24}. 
This work covers edge and subgraph counts \cite{BlockiBDS13,KasiviswanathanNRS13,ChenZ13}, degree distribution estimation \cite{RaskhodnikovaS16,DayLL16}, connected component counts \cite{KalemajRST23,JainSW24}, and  implicit $b$-matchings \cite{DinitzLLZ25}; and, for distributional tasks, includes parameter estimation for $\ergraph$ graphs \cite{SealfonU21,ChenDHS24}, block models \cite{ChenDDHLS24}, and graphons \cite{BorgsCS15,BorgsCSZ18}. Node-private algorithms 
are also known %
for the \emph{continual release} setting \cite{JainSW24}.
Our techniques necessarily depart from this approach: the indistinguishability of empty and regular graphs (\cref{lem:hell-empty-dreg}) 
shows that the usual design paradigm for node-private algorithms fails in the local model.

Local (LDP) algorithms can be traced back to Warner \cite{Warner65} and were formalized by \cite{DworkMNS16,KasiviswanathanLNRS11,JosephMNR19}.
Connections to information theory (e.g., entropy, mutual information, and KL divergence) were developed by \cite{DuchiJW13} and extended in \cite{BassilyS15,EdmondsNU20,ChenGKM21}.
Further investigation revealed the power of interactivity \cite{JosephMNR19} and properties such as advanced grouposition and the equivalence of pure and approximate LDP \cite{BunNS19}. 
While we leverage some LDP tools for tabular data, correlations inherent to distributed graph data required developing new algorithmic techniques and %
connections to information theory.

Prior work on local privacy for graphs has focused entirely on edge privacy. Papers evoking “local node privacy” do not match our definition and do not align with the notion of node privacy in the central model: \cite{QinYYKXR17,YeHAMX22,ZhangWCZB25}
protect only a node’s own edge list---and not the copies of the same edges held by its neighbors (in fact providing a definition equivalent to tabular LDP), %
while \cite{ZhangLBR20} protects node attributes under a public topology (with publicly known and thus unprotected edges). Work on local edge privacy began with an investigation of synthetic graphs \cite{QinYYKXR17} and subgraph  counting \cite{ImolaMC21}. Later, \cite{DhulipalaLRSSY22} formalized the model and linked it to parallelizable graph algorithms.  The first lower bounds specific to local edge-DP (that do not immediately follow from lower bounds for LDP) were shown by \cite{EdenLRS25} for triangle counting. Upper and lower bounds in \cite{EdenLRS25} for counting triangles were extended to subgraphs by \cite{SuppakitpaisarnPHH25}. Finally, \cite{MundraPSL25} advanced practical local algorithms that satisfy edge privacy.

\subsection{Open Questions}\label{sec:open}

Our work raises several concrete %
open questions
for \LNDP. The first is to fully pin down the optimal error needed for edge counting: is the $D\sqrt{n}$ term in our 
$O_\delta\big(\frac{D\sqrt{n}+n}\eps\big)$
error bound for $D$-bounded graphs inherent? We already showed that the $\frac n\eps$ %
term 
is necessary and that the bound is tight for sparse graphs.

On the structural side, we proved that pure and approximate \LNDP are genuinely different, but 
the gap between them is not fully understood.
Similarly, we have a separation between \LNDP algorithms that see full edge lists and those that only see degrees; strengthening this separation or finding additional tasks that witness it would clarify the relative power of different local views.

Finally, can interaction reduce the error of %
LNDP
algorithms? %
Our lower bounds for edge counting and $\ergraph$ parameter estimation give a partial answer: they hold for
$\ell$-round interactive protocols, showing that 
our (noninteractive) \LNDP algorithms are optimal even when constant-round 
interaction is
permitted. 
Notably, no noninteractive--interactive separations are known for %
local edge DP (LEDP), and the best known interactive LEDP algorithms are only constant round, with noninteractive LEDP algorithms matching their error up to $\eps$ dependence (e.g., $O(\frac{n^2}\eps +\frac{n^{3/2}}{\eps^2})$ for interactive vs.\ $O(\frac{n^2}\eps +\frac{n^{3/2}}{\eps^3})$ for noninteractive triangle counting
\cite{ImolaMC22,EdenLRS25}).
However, separations are known for %
tabular data %
\cite{KasiviswanathanLNRS11,JosephMNR19}, so understanding the power of interaction for LNDP remains an interesting open question.

\subsection{Organization}
\cref{sec:lndp} gives some background on DP and formalizes \LNDP. %
\cref{sec:blur-new} develops our 
algorithmic framework for releasing linear queries and applications to estimating
degree distributions' PMFs and CDFs, edge counts, $\ergraph$ parameters, and clique sizes.
\cref{sec:lbs} presents our lower-bounds framework. %
Structural results are in \cref{sec:adv-grp,sec:dist-stars-regular}: \cref{sec:adv-grp} separates pure and approximate \LNDP via advanced grouposition and derives pure \LNDP lower bounds; \cref{sec:dist-stars-regular} separates degrees-only from unrestricted \LNDP.

\section{Local Node Differential Privacy}
\label{sec:lndp}

In this section, we state the definition of DP and formalize our \LNDP model. We use $[n]$ to denote $\{1, \ldots, n\}$.

\paragraph{Node neighbors and differential privacy}%

Differential privacy is defined with respect to neighboring datasets, which differ in the data of a single individual. In the tabular setting, two datasets $x,x'\in\cX^n$ are neighbors if they differ in one entry. We focus on node privacy for graphs: graphs $G$ and $G'$ on node set $[n]$ are \emph{node neighbors}, denoted $G\sim G'$, if one can be obtained from the other by rewiring a single node, i.e., by changing only the edges incident to some node $i\in[n]$. More generally, $G$ and $G'$ are at \emph{node distance} $k$ if one can be obtained from the other by rewiring $k$ nodes.

\begin{definition}[$(\eps,\delta)$-indistinguishability]\label{def:indistinguishable}
Let $\eps>0$ and $\delta\in[0,1]$. Two distributions $P,Q$ over outcome space $\cY$ are \emph{$(\eps,\delta)$-indistinguishable}, denoted $P\simed Q$, if for all $Y\subseteq\cY$, we have
\[
    P(Y)\le e^\eps Q(Y)+\delta
    \ \ \text{and}\ \ 
    Q(Y)\le e^\eps P(Y)+\delta.
\]
We also write $R_1\simed R_2$ for random variables $R_1$ and $R_2$ with $(\eps,\delta)$-indistinguishable distributions.
\end{definition}

\begin{definition}[Differential privacy (DP) \cite{DworkMNS16}]
\label{def:dp batch}
Let $\eps>0$ and $\delta\in[0,1]$. A randomized algorithm $\cA:\cU^*\to\cY$ is $(\eps,\delta)$-DP if $\cA(x)\simed \cA(x')$ for every pair of neighboring inputs $x,x'\in\cU^*$.
\end{definition}

\Cref{sec:dp-background} reviews standard properties of differentially private algorithms, common mechanisms used as building blocks, and the usual definition of local DP for tabular data.

\paragraph{Our model: local node differential privacy}
\LNDP extends node-privacy for graphs to the local model.
In an \LNDP algorithm, each node runs a \emph{local randomizer} on its set of neighbors and reports the result to the untrusted server, which aggregates all reports to produce the final output. Local randomizers may depend on public randomness. 
We require that the joint distribution over all 
reports is $(\eps,\del)$-indistinguishable on node-neighboring
graphs. For an undirected graph $G=([n],E)$, let $\neigh{G}{i}$ be the neighborhood of node $i\in[n]$.

\begin{definition}[Noninteractive local node differential privacy (\LNDP)]
\label{def:lndp}
    Let $\eps > 0, \del\in[0,1]$, and $n\in \N$.
    A (randomized) algorithm $\cA$
    is \edLNDP
    if there exist (i) a distribution $\distpubrand$ over  
    strings $\pubrand$ \emph{(public randomness)}, 
    (ii) local randomizers $\cR_{1,\pubrand},\ldots, \cR_{n,\pubrand}$ (depending on $\pubrand$),
    and (iii) a postprocessing algorithm $\cP$ such that 
    \begin{enumerate}
        \item \emph{(Locality and noninteractivity)}
        for all graphs $G$ on node set $[n]$, the algorithm
        $\cA$ can be represented as
        \[
            \cP \bigl( \cR_{1,\pubrand}(\neigh{G}{1}), \ldots, \cR_{n,\pubrand}(\neigh{G}{n}) \bigr),
        \]
        where $\pubrand\sim \distpubrand$ and node $i$ runs a randomizer $\cR_{i,\pubrand}$ on its neighborhood $\neigh{G}{i}$, and 
        \item \emph{(Node privacy)} For all node-neighboring graphs $G$ and $G'$ on node set $[n]$ and all settings of public randomness $\pubrand$,
        the distributions of vectors of outputs released by the randomizers are $(\eps,\del)$-indistinguishable---that is, $\vec{\cR}_\pubrand(G)
        \approx_{\eps,\del}
        \vec{\cR}_\pubrand(G')$, where $\vec{\cR}_\pubrand(G) = \bigl( \cR_{1,\pubrand}(\neigh{G}{1}), \ldots, \cR_{n,\pubrand}(\neigh{G}{n}) \bigr)$.
        Equivalently, for all $\rho$, the map $\vec{\cR}_\pubrand$ is $(\eps, \del)$-DP under the node-neighboring relation.
    \end{enumerate}
    If  each randomizer $\cR_{i,\pubrand}$ in algorithm $\cA$ only needs degree $d_i=|N_i^G|$ as input (rather than the full neighborhood $\neigh{G}{i}$), we say $\cA$ is a \emph{degrees-only} \LNDP algorithm. 
\end{definition}

We omit $\pubrand$ when there is no public randomness (e.g., $\pubrand$ is always empty) or when it is clear from context.

Later (in \Cref{def:int-lndp}), we define an \emph{interactive} version of \LNDP, in the style of
\cite{KasiviswanathanLNRS11,JosephMNR19}, in which the server queries 
nodes adaptively based on previous messages from all nodes. 
In this work, we focus on algorithms and impossibility results for the noninteractive setting, with the exception of \cref{sec:lbs}, where we extend our impossibility results for several fundamental problems to the interactive setting.

\paragraph{Guarantees in the presence of malicious parties} If some parties (i.e., nodes) deviate from the protocol, the distributions of outputs on some node-neighboring graphs $G$ and $G'$ may become distinguishable, since malicious parties may see edges on which $G$ and $G'$ differ. Nevertheless,  the \LNDP guarantee still holds for the \emph{graph induced by the nodes corresponding to the honest parties}. That is, information visible only to honest parties remains protected.

\section{Algorithmic Tool: Blurry Degree Distributions}
\label{sec:blur-new}

We now present our %
algorithmic \LNDP framework for privately estimating %
linear queries about a graph's degree distribution, based on a new object we call the \textit{blurry degree distribution}. In \Cref{sec:blurry-degree-dist}, we introduce %
this new object and describe its properties. In \Cref{sec:prove-lin-quer}, we give our algorithm for answering linear queries, prove its guarantees (\Cref{thm:linear-queries-blurry,thm:fact-mech-blurry}), and apply it to privately releasing the PMF and CDF of the degree distribution (\Cref{cor:deg-dist-new,cor:cumu-deg-dist-new}). We then develop further applications: we give an algorithm for estimating the average degree of a graph with concentrated nonzero degrees in \Cref{sec:conc-deg} and 
apply it to  
estimating the parameter of an \ergraph graph
and clique size in \Cref{sec:er-new,sec:clique-new}, respectively.

\subsection{Blurry Degree Distributions}\label{sec:blurry-degree-dist}

To build our \LNDP framework for answering linear queries, we introduce a new object, the 
\emph{blurry degree distribution}. %
The degree distribution of a graph $G$ on node set $[n]$ is defined by $\ddist_G(d)=\frac{1}{n}\sum_{i\in[n]}\indic[d_i=d]$. Rather than answer queries about $\ddist_G$ directly, our algorithm answers related queries about the blurry degree distribution, which closely approximates $\ddist_G$ and has lower sensitivity.
The blurry degree distribution $\fakedistnew$
is obtained by rounding %
each degree in $G$ to the two nearest integer multiples of an analyst-specified parameter $s\in \N$ and splitting its contribution between them according to their distance to the degree. %
We formalize this via
the randomized rounding map $\randround \colon \R \to s\Z$ depicted in \Cref{fig:blur-deg-dist} in \Cref{sec:techniques} and defined as
\begin{equation}
\label{eq:blur-rand-map}
    \randround(x) = 
    s\bparen{ \floor{x/s} + \Bern(\{x/s\}) },
\end{equation}
where $\set{x/s} = x/s - \floor{x/s}$ is the fractional part of $x/s$. 
Then $\randround(x) \in \{s \floor{x/s}, s \ceil{x/s}\}$ for all $x \in \R$.

\begin{definition}[Blurry and compressed blurry degree distributions]
\label{def:blur-deg-dist-new}
Let $\slopelength\in\N$ and $G$ be a graph. Define the \emph{blurry degree distribution} $\fakedistnew$ as the probability mass function (PMF) of $\randround(X)$, where $X\sim\ddist_G$ and $\randround$ is as in \cref{eq:blur-rand-map}. Define the \emph{compressed blurry degree distribution} $\shortddist$ as the PMF of $\frac{1}{s}Y$, where $Y\sim\fakedistnew$.
\end{definition}

The distributions $\fakedistnew$ and $\shortddist$ are equivalent: a sample from one is obtained from a sample of the other by multiplying or dividing by $s$. The blurry degree distribution $\fakedistnew$ is defined on $\{0,\ldots,n-1\}$,  making it directly comparable to $\ddist_G$, but is supported only on multiples of $s$. Our algorithms therefore operate with the \emph{compressed blurry degree distribution} $\shortddist$, whose support consists of the $\blurdim:=\ceil{n/s}+1$ multiples of $s$ on which $\fakedistnew$ can be nonzero. Finally, although we define $\fakedistnew$ and $\shortddist$ as functions, it is often convenient to view them as $0$-indexed column vectors of dimension $n$ and $\blurdim$, respectively: the $\ord{i}$ entry of the vector form of $\fakedistnew$ is $\fakedistnew(i)$, and $\shortddist$ is obtained by restricting $\fakedistnew$ to its $\blurdim$ possibly nonzero coordinates.

Next we show that $\fakedistnew$ equals $\ddist_G$ in expectation and is close to it in Wasserstein $\infty$-distance, defined for distributions $F_P,F_Q$ of random variables $P,Q$ as
$\winf(F_P,F_Q)=\inf\{\sup|X-Y|:(X,Y)\text{ is a coupling of }P,Q\}.$

\begin{lemma}[Properties of the blurry degree distribution $\fakedistnew$]
\label{lem:blur-deg-props-new} For all $s\in \N$ and graphs $G$,
\begin{center}
    \begin{enumerate*}[label=(\alph*)]
    \item \label{item:blur-deg-exp-new} $\Ex_{Y \sim \fakedistnew}[Y] = \Ex_{X \sim \ddist_G}[X]$, 
    \ifnum\focsversion=1
        \qquad \quad
    \else
        \qquad \qquad \qquad \qquad \qquad
    \fi
    \item \label{item:blur-deg-winf-new} $\winf\Bparen{\ddist_G, \fakedistnew} \leq \slopelength$.
\end{enumerate*}
\end{center}

\end{lemma}

\begin{proof}
\emph{\Cref{item:blur-deg-exp-new}.}  
Fixing $d \in \Z^{\geq 0}$, %
we have $\E[\randround(d)] = s \floor{d/s} + s(d/s - \floor{d/s}) = d$.
The law of total expectation gives $\E_{X \sim \ddist_G} \left[\randround(X)\right] = \E_{X \sim \ddist_G} [X]$, implying \Cref{item:blur-deg-exp-new} by \Cref{def:blur-deg-dist-new}.

\emph{\Cref{item:blur-deg-winf-new}.} 
\Cref{def:blur-deg-dist-new} provides a coupling $(X,\randround(X))$ of $\ddist_G$ and $\fakedistnew$, where $X \sim \ddist_G$ and $\randround(X) \sim \fakedistnew$. Since $\randround(x)$ maps to %
$s\floor{x/s}$ or $s\ceil{x/s}$, which are at most $s$ away from $x$,
we get $|X - \randround(X)| \leq \slopelength$.
\end{proof}

\subsection{Answering Linear Queries
about the Blurry Degree Distribution
}
\label{sec:prove-lin-quer}

In this section, we describe our method for privately answering arbitrary linear queries about the compressed blurry degree distribution $\shortddist$. The resulting guarantees are stated in \Cref{thm:linear-queries-blurry,thm:fact-mech-blurry}, their implications for privately releasing the PMF and CDF of the degree distribution are given in \Cref{cor:deg-dist-new,cor:cumu-deg-dist-new}. The algorithm itself appears in \Cref{sec:algorithm-for-linear-queries}, and the proofs of \Cref{thm:linear-queries-blurry,thm:fact-mech-blurry} are in \Cref{sec:proofs-of-linear-query-thms}.

Let $\cedp := \frac{\sqrt{2\log(1.25/\del)}}{\eps}$.
For $\alpha, \beta \in \mathbb{N} \cup \{\infty\}$ and $A \in \R^{m \times n}$, define the $\alpha \to \beta$ %
norm as
$\|A\|_{\alpha \to \beta} = \max_
{v \in \R^n : \|v\|_{\alpha} \leq 1} 
\|A v\|_{\beta}$. \cref{thm:linear-queries-blurry,thm:fact-mech-blurry} use the following %
norms: 
$\normoneinf{\cdot}$ is the maximum absolute value of a matrix entry, and $\normtwoinf{\cdot}$ and $\normonetwo{\cdot}$ are the maximum $\ell_2$ norms of a row and column, respectively.

\begin{theorem}[Linear queries about $\shortddist$]
\label{thm:linear-queries-blurry}
    Let $\eps > 0,\del \in (0, 1]$, and $n, k, s \in \mathbb{N}$. Define $\blurdim = \ceil{n/s} + 1$. Let $\matR \in \R^{k\times \blurdim}$ be a matrix. There is an \edLNDP algorithm $\algmatrix^{\matR}$ (\cref{alg:matrix-local}) such that for all graphs $G$ on node set $[n]$, we have $\algmatrix^{\matR}(G) = \matR \, \shortddist + Z$, where $Z \sim \cN(\vec 0, \sigma^2 \, \mathbb{I}_k)$ and %
    \(
        \sigma
        =
        O \paren{\normonetwo{\matR}\cdot \sqrt{\frac{1}{n} + \frac{1}{s^2}} \cdot \cedp}.
    \)
\end{theorem}

In \Cref{thm:fact-mech-blurry}, %
we show that the factorization mechanism \cite{HardtT10,BhaskaraDKT12,LiMHMR15,NikolovT016,EdmondsNU20} can be used
when designing \LNDP algorithms.
Up to a factor of $\sqrt{1 + \frac{n}{s^2}}$, the accuracy guarantees of \cref{thm:fact-mech-blurry} match those of the factorization mechanism on tabular data in the standard local model \cite{EdmondsNU20}.

\begin{theorem}[Applying the factorization mechanism to $\shortddist$]
\label{thm:fact-mech-blurry}
    Let $\eps > 0,\del \in (0, 1]$, and $n, k, s \in \mathbb{N}$. Define $\blurdim 
    = \ceil{n/s} + 1$. Let $\factapprox \in \R^{\geq 0}$ and $\matW\in \R^{k\times \blurdim}$ be a workload matrix. There is an \edLNDP algorithm $\algfact^{\matW, \factapprox}$ such that for all graphs $G$ on node set $[n]$, 
    \begin{align*}
        \ifnum\focsversion=1
            &
        \fi
        \E \bracks{\|\algfact^{\matW, \factapprox}(G) - \matW\, \shortddist \|_\infty}
        \ifnum\focsversion=1
            \\
        \fi
        &=
        O \paren{\factapprox + \factnorm{\matW}{\factapprox}\cdot \sqrt{\paren{\frac{1}{n} + \frac{1}{s^2}}\cdot \log k} \cdot \cedp },
    \end{align*}
    where $\factnorm{\matW}{\factapprox} = \min \{\normtwoinf{\matL} \normonetwo{\matR} : \normoneinf{\matL\matR - \matW} \leq \factapprox\}$ is the \emph{$\factapprox$-approximate factorization norm}.
\end{theorem}

To interpret \cref{thm:linear-queries-blurry,thm:fact-mech-blurry}, if
a data analyst wants to evaluate a workload $\matW \in \R^{k \times n}$ of linear queries about the degree distribution $\ddist_G$ of a graph $G$, she
can reduce $\matW$ to a workload $\matW' \in \R^{k \times \blurdim}$ on the compressed blurry degree distribution $\shortddist$
(e.g., by keeping every $\ord{s}$ column of $\matW$, starting with the first). The theorems show that $\matW'$ can be answered under \LNDP with small $\ell_\infty$ error. Hence, up to a left–right shift of at most $s$ (i.e., $\winf(\fakedistnew,\ddist_G)\le s$), the analyst can accurately answer arbitrary linear queries about the degree distribution. Since $\fakedistnew$ and $\shortddist$ are equivalent up to rescaling by $s$, answering linear queries about one is equivalent to answering them about the other.

We now apply these theorems 
to obtain bicriterion approximations of the degree distribution's PMF and CDF (i.e., simultaneous guarantees in the $\winf$ distance between $\fakedistnew$ and $\ddist_G$, and the $\ell_\infty$ distance between the estimate %
and the true PMF/CDF of $\fakedistnew$).
Taking the workload to be the identity matrix $\mathbb{I}_\blurdim \in \R^{\blurdim \times \blurdim}$ yields an estimate of the PMF of $\fakedistnew$ and using the lower-triangular all-ones matrix 
$\mcount := (\indic_{i \geq j})_{i,j \in [\blurdim]} \in \R^{\blurdim \times \blurdim}$ 
as the workload yields an estimate of the CDF of $\fakedistnew$.

\begin{corollary}[\az{Blurry degree PMF approximation}]
\label{cor:deg-dist-new}
    Let $\eps > 0, \del \in (0, 1]$, and $n, s \in \N$. Define $\blurdim = \ceil{n/s} + 1$. For all graphs $G$ on node set $[n]$, the \edLNDP algorithm
    $\algmatrix^{\mathbb{I}_\blurdim}$ (as in \Cref{thm:linear-queries-blurry})
    satisfies
    \begin{flalign*}
    \E \bracks{\|\algmatrix^{\mathbb{I}_\blurdim}(G) \!- \shortddist\|_\infty}
    &\!= O\paren{\!\sqrt{\!\paren{\!\frac{1}{n} + \frac{1}{s^2}\!} \log (n/s)} \cdot \cedp \!}. %
    \end{flalign*}
\end{corollary}
In \cref{cor:cumu-deg-dist-new}, we use the well-known fact that $\factnorm{\mcount}{0} = \Theta(\log(n))$---see, e.g., \cite{HenzingerKU25,Mathias93b}.

\begin{corollary}[\az{Blurry degree CDF approximation}]
\label{cor:cumu-deg-dist-new}
    Let $\eps > 0, \del \in (0, 1]$, and $n, s \in \N$. Define $\blurdim = \ceil{n/s} + 1$ and $\mcount = (\indic_{i \geq j})_{i,j \in [\blurdim]} \in \R^{\blurdim \times \blurdim}$.
    For all graphs $G$ on node set $[n]$, the \edLNDP algorithm $\algfact^{\mcount, 0}$ (as in \Cref{thm:fact-mech-blurry}) satisfies
    \begin{align*}
        \ifnum\focsversion=1
            &
        \fi
        \E \bracks{\| \algfact^{\mcount, 0}(G) - \mcount (\shortddist) \|_\infty}
        \ifnum\focsversion=1
            \\
        \fi
        &=
        O\paren{\sqrt{\paren{\frac{1}{n} + \frac{1}{s^2}}\cdot \log (n/s)}\cdot \log (n) \cdot \cedp}
        .
    \end{align*}
\end{corollary}

\Cref{cor:funny-pdf} further improves our bicriterion approximation for PMF by reducing the $\ell_\infty$ error in \Cref{cor:deg-dist-new} by a factor of $s$.
The idea is to spread each point of mass of $\fakedistnew[G][\slopelength/2]$ along an interval of width $\slopelength/2$, resulting in a distribution with Wasserstein-$\infty$ distance at most $\frac{\slopelength}{2} + \frac{\slopelength}{2} \leq s$ from $\ddist_G$, but with %
smaller $\ell_\infty$ error.\footnote{\Cref{cor:funny-pdf} takes advantage of our specific $\ell_\infty$/$W_\infty$ error model to improve the error for PMF estimation.
This approach does not reduce the $\linf$ error for CDF estimation because, by the end of the interval over which each point mass is spread, the cumulative sum includes that point's entire mass and estimation error.
}
This gives our PMF estimation error in \Cref{tab:deg-dist}.

\begin{corollary}[\az{PMF approximation}] %
\label{cor:funny-pdf}
Let $\eps > 0, \del \in (0, 1]$, and $n, s \in \mathbb N$ with $s$ even.
There is an $(\eps, \del)$-\LNDP algorithm $\cB$ such that, for all graphs $G$ on node set $[n]$, there exists a distribution $S_G$ with $W_\infty(\ddist_G, S_G) \leq s$ and
\[
    \E \left[\|\cB(G) - S_G\|_\infty\right]
    =
    O \left(\frac{1}{s}\sqrt{\left(\frac{1}{n} + \frac{1}{s^2}\right) \log(n/s)}  \cdot \cedp \right).
\]
\end{corollary}
\begin{proof}
Let $h = \frac s 2$ and note that $h\in \mathbb{N}$. For every function $g \colon \N \to \mathbb{R}$ with support $h \N$, define $(F \circ g)(jh + r) = g(jh)/h$ for all $j \in \N$ and $r \in \{0, \ldots, h-1\}$,
i.e., spread the mass at $jh$ uniformly over $\{jh, \ldots, jh + (h-1)\}$.
Let $\cA$ be the algorithm of \Cref{cor:deg-dist-new} run with blur parameter $h$, with coordinate $j$ of its output at $jh$, so that $\cA(G)$ estimates $\fakedistnew[G][h]$.  %
Define $S_G = F \circ \fakedistnew[G][h]$ and $\cB(G) := F \circ \cA(G)$. Since $\cB$ is a post-processing of $\cA$, it is $(\eps, \del)$-\LNDP. 

For the Wasserstein-$\infty$ error, $F$ moves the mass of $\fakedistnew[G][h]$ by at most $h-1$, so $W_\infty(\fakedistnew[G][h], S_G) \leq h-1$. Using \Cref{item:blur-deg-winf-new} of \Cref{lem:blur-deg-props-new}, we get $W_\infty(\ddist_G, S_G) \leq W_\infty(\ddist_G, \fakedistnew[G][h]) + W_\infty(\fakedistnew[G][h], S_G) \leq h+(h-1) \leq s$.

For the $\ell_\infty$ error, the sets $\{jh,\dots,jh+h-1\}$ are disjoint and $F$ divides the mass at $jh$ evenly along each element in this set, so each coordinate of $F \circ g$ depends on a single coordinate of $g$, giving $\|F \circ g-F \circ g'\|_\infty=\|g-g'\|_\infty/h = \frac{2}{s} \|g-g'\|_\infty$ for any $g,g'$ supported on
$h\mathbb{N}$. Therefore, 
\begin{align*}
\E \left[\|\cB(G) - S_G\|_\infty\right]
\cond{&}
= \frac{2}{s} \cdot \E \left[\|\cA(G) - \fakedistnew[G][h]\|_\infty\right]
\cond{\\&}
= O \left(\frac{1}{s} \sqrt{\left(\frac{1}{n} + \frac{1}{s^2}\right) \log(n/s)}  \cdot \cedp \right),
\end{align*}
where we use %
\Cref{cor:deg-dist-new}, completing the proof.
\end{proof}

\subsubsection{Our Algorithm for Privately Answering Linear Queries}\label{sec:algorithm-for-linear-queries}
In this section we present \Cref{alg:matrix-local}, used to prove \cref{thm:linear-queries-blurry,thm:fact-mech-blurry}.
In the algorithm, we view the compressed blurry degree distribution $\shortddist$ as the application  of a \textit{blur matrix} $A_s$ to the exact degree distribution $\ddist_G$, i.e., $\shortddist = \blurmatrix \ddist_G$, as formalized in the following lemma. %
\begin{lemma}
\label{lem:shortddist-blurmatrix}
Let $n,  s \in \N$, and define $\blurdim = \ceil{n/s} + 1$.  Define the \emph{blur matrix} $\blurmatrix \in \R^{\blurdim \times n}$ as
\begin{equation}
\label{eqn:blurmatrix}
(\blurmatrix)_{i+1,j+1} := \Pr \left[\randround(j) = s \cdot i \right] %
\end{equation}
for all $i \in \{0, \ldots, \blurdim-1\}$ and $j \in \{0, \ldots, n-1\}$. Then, for every graph $G$ on node set $[n]$, we have $\shortddist = \blurmatrix \ddist_G$, where we view $\ddist_G$ and $\shortddist$ as column vectors in $\R^n$ and $\R^\nu$, respectively.
\end{lemma}
\begin{proof}[Proof of \Cref{lem:shortddist-blurmatrix}]
Fix $i \in \{0, \ldots, \blurdim - 1\}$. It suffices to show $\shortddist(i) = \angles{(\blurmatrix)_{i+1, \bullet} , \ddist_G}$, where $(\blurmatrix)_{i+1, \bullet}$
denotes %
row $i+1$
of $\blurmatrix$.
Let $(e_1, \ldots, e_n)$ denote the standard basis of $\R^n$. For every graph $G$ on node set $[n]$, we represent %
$\ddist_G \in \mathbb{R}^n$ as the linear combination $\ddist_G = \sum_{k=0}^{n-1} \ddist_G(k) \, e_{k+1}$.
Thus,
\begin{align*}
    \shortddist(i)
    &= \Pr_{X \sim \ddist_G}\bbracks{\tfrac{1}{s} \randround(X) = i} 
    = \Pr_{X \sim \ddist_G}\bbracks{\randround(X) = s\cdot i} 
    \condnewline
    \condalign
    = \sum_{k=0}^{n-1} (\blurmatrix)_{i+1, k+1} \cdot \ddist_G(k)
    = \angles{(\blurmatrix)_{i+1, \bullet} , \ddist_G},
\end{align*}
where the first equality follows from 
 \Cref{def:blur-deg-dist-new}, and the third equality from \cref{eqn:blurmatrix}. 
\end{proof}
We now present \Cref{alg:matrix-local}. At a high level, for a workload matrix $\matR$ of linear queries, each user applies $\matR A_s$
to the basis vector for its degree and releases a noisy version of the resulting vector, which the central server then averages to obtain an estimate for $\matR \, \shortddist$.

\begin{algorithm}[ht!]
    \caption{$\algmatrix^M$ for answering a workload $M$ of linear queries about the blurry degree distribution.}
    \label{alg:matrix-local}
    
    \begin{algorithmic}[1]
    \label{alg:matrix-local-algorithmic}
        \Statex \textbf{Parameters:} Privacy parameters $\eps > 0$, $\del\in (0,1]$;
        number of nodes $n\in\N$; $s \in \R^+$;
        \Statex matrix $\matR \in \R^{k \times \blurdim}$, where $\blurdim = \ceil{n/s} + 1$ and $k \in \mathbb{N}$.    
        \Statex \textbf{Input:}
        Graph $G$ on vertex set $[n]$.
        \Statex \textbf{Output:} Estimate $\wh v \in \R^k$ of the vector $v := \matR \, \shortddist$.
        \For{{\bf all} nodes $i\in[n]$}
            \State Node $i$ computes $r_i \gets \matR(\blurmatrix \, e_{d_i})$. 
            \cond{\Statex}
            \Comment{\commentstyle{$e_0, \ldots, e_{n-1}$ is the standard basis of $\R^n$, and $d_i$ is node $i$'s degree.}}
            \label{line:fact-compute-ri}
            \State Node $i$ sends $\wh r_i \gets r_i + Z_i$, where $Z_i \sim \cN(\vec 0, \tilde \sigma^2 \mathbb{I}_k)$ and $\tilde \sigma^2 = 4 \normonetwo{\matR}^2 (1 + \frac{n}{s^2}) \cdot  \cedp^2$.
            \label{line:fact-send-noisy}
        \EndFor
        \State Central server returns $\wh v \gets \frac{1}{n} \sum_{i=1}^n \wh r_i$.
    \end{algorithmic}
\end{algorithm}

\subsubsection{Proofs of \Cref{thm:linear-queries-blurry,thm:fact-mech-blurry}}\label{sec:proofs-of-linear-query-thms}

\begin{proof}[Proof of \Cref{thm:linear-queries-blurry}]

\textbf{(Accuracy.)}
The output of $\algmatrix^\matR(G)$ can be written as
\begin{align*}
    \hat v &= \frac{1}{n} \sum_{i=1}^n \hat r_i = \frac{1}{n} \sum_{i=1}^n (r_i + Z_i) = \paren{\frac{1}{n} \sum_{i=1}^n r_i} + Z 
    \condnewline 
    \condalign
    = \matR \blurmatrix \left(\frac{1}{n} \sum_{i=1}^n e_{d_i} \right) + Z = \matR \blurmatrix \ddist_G + Z = \matR \shortddist + Z,
\end{align*}
where
$Z = \frac{1}{n} \sum_{i=1}^n Z_i \sim \cN \left(\vec 0, \frac{\tilde \sigma^2}{n}\, \mathbb I_k\right)$ and $\tilde \sigma = O \left(\cedp \normonetwo{\matR} \sqrt{1 + \frac{n}{s^2}} \right)$ is defined as in \Cref{line:fact-send-noisy}.

\textbf{(Privacy.)} 
By \Cref{def:lndp}, to prove that the algorithm is $(\eps, \del)$-\LNDP, 
it suffices to show that releasing the randomizer outputs $\hat r_1, \ldots, \hat r_n$ is $(\eps, \del)$-DP.
Let $G$ and $G'$ be graphs on node set $[n]$ that differ only on the edges incident to a node $i^* \in [n]$. 
For $i\in[n]$, define $d_i,e_i,$ and $r_i$  as in \Cref{alg:matrix-local}. Set $\tilde e_{d_i} = \blurmatrix e_{d_i}$.
Let the vector $r_G 
\in \R^{n \blurdim}$ be the concatenation of $r_1, \ldots, r_n$. Similarly, define $d_i'$, $\tilde e_{d_i'}$, and $r_{G'}$ for $G'$.

Let $\Delta_2$ denote the
$\ell_2$ distance
between $r_G$ and $r_{G'}$. By privacy of the Gaussian mechanism (\Cref{lem:gaussianmech}), it suffices to show that $\Delta_2^2 \leq 4 \normonetwo{\matR}^2 (1 + \frac{n}{s^2})$. 
To see why this holds, we %
break $\Delta_2^2$ into two terms:
\begin{align*}
    \Delta_2^2
    &= \|r_G - r_{G'}\|_2^2
    = \sum_{i \in [n]} \|\matR (\tilde e_{d_i} - \tilde e_{d_i'})\|_2^2 
    \condnewline \condalign
    = \| \matR (\tilde e_{d_{i^*}} - \tilde e_{d_{i^*}'}) \|_2^2 + \sum_{i \neq i^*} \| \matR (\tilde e_{d_{i}} - \tilde e_{d_{i}'}) \|_2^2.
\end{align*}
The first term concerns the changed node $i^*$. Because each column of $\blurmatrix$ has nonnegative entries that sum to at most 1, we have $\|\tilde e_{d_{i}}\|_1 \leq 1$ and $\|\tilde e_{d_{i^*}} - \tilde e_{d_{i^*}'}\|_1\leq 2$. We can therefore bound the first term:
\begin{align}
\label{eqn:changed-node-l2}
\begin{split}
\| \matR (\tilde e_{d_{i^*}} \!\!\!- \tilde e_{d_{i^*}'}\!) \|_2^2 
    &\leq \! \max_{\substack{v \in \R^{n} \\ \|v\|_1 \leq 2}} \|\matR v\|_2^2
    = \Big(\!\max_{\substack{v \in \R^{n} \\ \|v\|_1 \leq 2}} \! \|\matR v\|_2 \Big)^2
    \condnewline \condalign
    \leq \Big(2 \!\!\max_{\substack{v \in \R^{n} \\ \|v\|_1 \leq 1}} \|\matR v\|_2 \Big)^2\! = 4 \normonetwo{\matR}^2.
\end{split}
\end{align}
We now show $\|\tilde e_{d_i} - \tilde e_{d_i'}\|_1\leq \frac{2}{s}$ for all %
nodes $i \neq i^*$.
Because
$|d_i - d_i'| \leq 1$ when $i \neq i^*$, it suffices to show
\begin{equation}
\label{eqn:other-node-l1}
    \|\tilde e_{d} - \tilde e_{d + 1}\|_1\leq \tfrac{2}{s}
\end{equation}
for all $d \in \{0, \ldots, n-2\}$, which is equivalent to the statement that
$\dtv\bparen{\randround(d), \randround(d+1)} \leq \frac{1}{s}$ (since $\|A - B\|_1 = 2 \dtv(A, B)$ for all distributions $A, B$).
Fix $d \in \{0, \ldots, n-2\}$, and let
$k^* = \floor{\frac{d}{s}}$.
Then 
\begin{flalign*}
    \condalign
    \condnegspace \condnegspace \condnegspace 
    \dtv(\randround(d), \randround(d+1))
    \condnewline
    &= \frac{1}{2} \sum_{k=0}^{\blurdim-1} \bigl|\Pr[\randround(d) = sk] - \Pr[\randround(d+1) = sk] \bigr| \\
    &\leq \frac{1}{2s} \left(\bigl||d+1 - sk^*| - |d - sk^*| \bigr|\right. 
    \condnewline \condalign \condquad \condquad
    + \left.\bigl||d+1 - s(k^*+1)| - |d - s(k^*+1)| \bigr|\right) \\
    &\leq \frac{1}{2s} \left(1 + 1\right) = \frac{1}{s},
\end{flalign*}
proving \Cref{eqn:other-node-l1} and implying $\|\tilde e_{d_i} - \tilde e_{d_i'}\|_1\leq \frac{2}{s}$ for all %
nodes $i \neq i^*$.
A similar calculation to that in \Cref{eqn:changed-node-l2} gives
$\| \matR (\tilde e_{d_{i}} - \tilde e_{d_{i}'}) \|_2^2 \leq \frac{4}{s^2} \normonetwo{\matR}^2$. Overall,
\[
    \Delta_2^2
     \leq
    4 \normonetwo{\matR}^2 + (n-1) \cdot \tfrac{4}{s^2} \normonetwo{\matR}^2
     \leq
    4 \normonetwo{\matR}^2 \left(1 + \frac{n}{s^2}\right),
\]
showing that $\algmatrix^M$
is $(\eps,\del)$-\LNDP and completing the proof.
\end{proof}

\begin{proof}[Proof of \Cref{thm:fact-mech-blurry}]

Consider the algorithm $\algfact^{\matW, \factapprox}$ described as follows: given input graph $G$, it finds $\matL \in \mathbb{R}^{k \times \ell}$ and $\matR \in \mathbb{R}^{\ell \times \blurdim}$ such that $\normoneinf{\matL\matR - \matW} \leq \factapprox$ and $\normtwoinf{\matL} \normonetwo{\matR} = \gamma_\factapprox(\matW)$,\footnote{Matrices $\matL, \matR$ can be found in time polynomial in the size of $\matW$ by semidefinite programming \Cite{LinialS09}, as noted in \cite{EdmondsNU20}.
}
runs $\hat v = \algmatrix^\matR(\eps, \del, n, G)$, and returns $\matL \hat v$.

Since $\algfact^{\matW, \factapprox}$ is a postprocessing of $\algmatrix^\matR$ which is \edLNDP, so is $\algfact^{\matW, \factapprox}$.
We now prove the accuracy guarantee.
Let $\matL \in \mathbb{R}^{k \times \ell}$ and $\matR \in \mathbb{R}^{\ell \times \blurdim}$ be the matrices chosen as the $\factapprox$-approximate factorization of $\matW$. Define $E := \matL\matR - \matW$. Then $\normoneinf{E} \leq \factapprox$. So, we can write the output of $\algfact^{\matW, \factapprox}$ as
\begin{equation}\label{eq:error-of-factoring-mechanism}
    \algfact^{\matW, \factapprox}(G) = \matL \hat v = \matL\matR \, \shortddist + \matL Z = \matW \, \shortddist + E \, \shortddist + \matL Z,
\end{equation}
where $Z \sim \cN(0, \sigma^2 \I_k)$ and $\sigma$ is as in \Cref{thm:linear-queries-blurry}. The two error terms are $E\,\shortddist$, from the approximate factorization, and $\matL Z$, introduced by $\algmatrix^\matR$.
By \Cref{eq:error-of-factoring-mechanism}, the overall expected error is
\begin{align}
\label{eqn:fact-mech-error-decomp}
\E \left[\|\algfact^{\matW, \factapprox}(G) - \matW \, \shortddist \|_\infty\right]
\condalign
\leq \E \Big [\|E \, \shortddist\|_\infty \Big ] 
\condnegspace+ \condnegspace
\E \Big[ \|\matL Z\|_\infty \Big]
\cond{\nonumber} \condnewline \condalign
\leq \normoneinf{E} + \E \left[\|\matL Z\|_\infty\right]
\cond{\nonumber} \condnewline \condalign
\leq \factapprox + \E \left[\|\matL Z\|_\infty\right],
\end{align}
where the second inequality uses $\|\shortddist\|_1 = 1$. 
Note that %
$(\matL Z)_i \sim \cN(0, \|\matL_{i, \bullet}\|_2^2 \, \sigma^2)$ for all $i \in [k]$, where $\matL_{i, \bullet}$ is row $i$
of $\matL$. By a standard Gaussian tail bound (\Cref{lem:max-gauss}),
\begin{align*}
\E[\|\matL Z\|_\infty] 
\condalign
\leq \sigma \normtwoinf{\matL} \sqrt{2 \log(2k)}
\condnewline
&= O \left(\frac{ \cedp \normtwoinf{\matL} \normonetwo{\matR} \sqrt{1 + \frac{n}{s^2}} \sqrt{\log(k)}}{\sqrt{n}}\right).
\end{align*}
Using $\normtwoinf{\matL} \normonetwo{\matR} = \gamma_\factapprox(\matW)$ and \Cref{eqn:fact-mech-error-decomp} gives the desired bound on $\E [\|\algfact^{\matW, \factapprox}(G) - \matW \shortddist \|_\infty]$.
\end{proof}

\subsection{Estimating the Average Degree of a Concentrated-Degree Graph}
\label{sec:conc-deg}

In this section, we describe an \LNDP algorithm (\Cref{alg:conc-deg}) %
that estimates the average degree of a graph whose nonzero degrees are concentrated in some interval. The guarantees of the algorithm %
are summarized in \Cref{lem:conc-deg}.
We then use this algorithm as a subroutine in our algorithms for \ergraph parameter estimation and clique size estimation in \Cref{sec:er-new,sec:clique-new}, respectively.

\begin{lemma}[Estimating average degree in concentrated-degree graphs]
\label{lem:conc-deg}

    Let $\eps \in (0, 1]$, $\del \in (0, 1]$ and $n, s \in \mathbb{N}$ such that $s \geq \sqrt{n}$.
    There exists an algorithm  $\algconcdeg$ (\Cref{alg:conc-deg}) such that:
    \begin{enumerate}[label=(\alph*)]
        \item $\algconcdeg$ is $(\eps,\del)$-\LNDP for all graphs $G$ on node set $[n]$.
        \item 
        Let $\ccon > 0$ be a constant
        such that, for all sufficiently large $n\in\N$, the maximum magnitude of $\ceil{n/s}$ Gaussians from \cref{thm:linear-queries-blurry} with $s = \ceil{\sqrt{n}}$ is at most
        $\ccon \cdot \sqrt{\log (\frac{n}{s})/n} \cdot c_{\eps,\del}$ 
        with probability at least $\frac{19}{20}$ (such $\ccon$ exists by \Cref{lem:max-gauss}). %
        
        Suppose $\del\in \left( 0, \frac{1}{2} \right]$
        and $n$ is sufficiently large. Let $u$ denote the maximum degree of $G$, and suppose the interval $[\ell, u] := [u-s, u]$ contains the degrees of all non-isolated nodes, with at least
        $6 \ccon n\log(\frac{n}{s})\cdot \cedp$
        nodes in it. %
        Let $(\wh x,\wh v)\gets \algconcdeg(G)$. Then, there exists $\alpha = O\bparen{ \bparen{1 + \frac{s}{\sqrt{n}}} \cdot \cedp}$ such that, with probability at least $\frac{9}{10}$,
        \[
        [\ell,u]\subseteq[\wh x, \wh x + 4s]
        \quad \text{ and }
        \quad
        \left| \frac{k}{n}\cdot \wh x + \wh v - \frac{1}{n} \sum_{i\in[n]} d_i \right| \leq \alpha,
        \]
        where $k$ is the number of nodes with degree at least $\wh x$.
    \end{enumerate}

    In particular, condition (b) is also satisfied when
    $n$ is sufficiently large,
    $\delta \in \bigl(0, \frac{1}{2}\bigr]$,
    $\eps \geq c\cdot \sqrt{\frac{\log n \log(1/\del)}{n}}$ for some constant $c > 0$, and there are at least $\frac{n}{10}$ nodes with degrees in $[\ell,u]$.

\end{lemma}

This statement generalizes \cref{thm:edge-ct-alg}, albeit for a more restricted setting of $\eps$. Specifically, if the input graph has maximum degree at most $D$, then $\wh x + \wh v$ is an estimate of the average degree with error $O\bparen{\bparen{1 + \frac{D}{\sqrt{n}}}\cdot \cedp}$, which can be converted to an edge count with error $O\paren{\paren{n + D\sqrt{n}}\cdot \cedp}$. \cref{alg:conc-deg}, though, also accurately counts edges in graphs where all nodes' degrees are in an interval of width $s$, even for graphs of arbitrary maximum degree.

\begin{algorithm}[ht!]
    \caption{$\algconcdeg$ for estimating the average degree of a concentrated-degree graph.}
    \label{alg:conc-deg}
    
    \begin{algorithmic}[1]
        \Statex \textbf{Parameters:} Privacy parameters $\eps \in (0,1]$, $\del\in (0,1]$;
        number of nodes $n\in\N$; width $s\in \N$ s.t.\ $s\geq \sqrt{n}$; smallest candidate index $q \in \Z$.
        \Statex \textbf{Input:}
        Graph $G$ on vertex set $[n]$.
        \Statex \textbf{Output:} Average degree estimate $\wh d\in \R$.
        \State Let $\blurdim = \ceil{n/s} + 1$ and $c_1>0$ be the absolute constant from \cref{lem:conc-deg}.
        \State $(\wh v_0,\ldots, \wh v_{\blurdim - 1}) \gets \algmatrix^{\I_{\blurdim}}
        \paren{\eps, \del, G}$.
        \Comment{\commentstyle{Call to \cref{alg:matrix-local}.}}
        \label{line:vhat}
        \State $\jnz \gets \arg\max_{j\in[\blurdim - 1]}(\wh v_1,\ldots, \wh v_{\blurdim - 1})$. \label{line:jhat}
        \cond{\Statex}
        \Comment{\commentstyle{Choose nonzero index with largest value.}}
        \State
            $\jhat =
            \begin{cases}
                \jnz & \text{if $\wh v_{\jnz} \geq c_1\sqrt{\log(\frac{n}{s})/n}\cdot \cedp$,} \\
                0 & \text{otherwise.}
            \end{cases}$
        \cond{\Statex}
        \Comment{\commentstyle{Use index holding largest value if it's sufficiently large ($0$ otherwise).}}
        \State Central server returns
        $\wh x \gets s(\jhat - 2)$ and 
        \cond{\Statex \;} $\displaystyle \wh v \gets \sum_{i\in[4]}i\cdot s \cdot v_{\jhat - 2 + i}$.
        \Comment{\commentstyle{Let $\wh v_j = 0$ if $j\not\in\set{0,\ldots, \blurdim - 1}$.}}
    \end{algorithmic}
\end{algorithm}

\begin{proof}[Proof of \cref{lem:conc-deg}]
    \textbf{(Privacy.)}
    By \cref{thm:linear-queries-blurry} and postprocessing, $\algconcdeg$ is $(\eps,\del)$-\LNDP.

    \textbf{(Accuracy.)}
    Let $\overline{d} = \frac{1}{n}\sum_{i\in[n]} d_i$ denote the average degree of graph $G$.
    By assumption, $G$ has an interval of width $s$ as described in \cref{lem:conc-deg}.
    Let $u$ be the maximum degree of $G$, let $\min_{\msf{nz}}$ be its minimum nonzero degree, and $\ell$ be the largest value in $\set{0, \min_{\msf{nz}}}$ such that at least $6\ccon \sqrt{n\log(\frac{n}{s})}\cdot \cedp$ nodes have degree in $[\ell,u]$. Thus, $[\ell,u]$ is an interval satisfying the conditions of \cref{lem:conc-deg}.
    
    We first show that, with probability at least $\frac{19}{20}$ over the randomness of $\jhat$, we have $[\ell, u] \subseteq [s(\jhat - 2), s(\jhat + 2)]$.
    We next show that if $[\ell, u] \subseteq [s(\jhat - 2), s(\jhat + 2)]$, then the additive error on the estimate $\wh d$ of the average degree is $O\bparen{\bparen{ 1 + \frac{s}{\sqrt{n}} } \cdot \cedp}$, with probability at least $\frac{19}{20}$.
    \cref{lem:conc-deg} follows by a union bound.
    
    Throughout this proof, let $(v_0,\ldots, v_{\blurdim-1}) = \shortddist$. Let $(\wh v_0,\ldots, \wh v_{\blurdim-1})$ be as defined on Line~\ref{line:vhat} and let $\wh v_j = 0$ and $v_j = 0$ for all $j\not\in\set{0, \ldots, \blurdim - 1}$.
    Note that $\sum_{j\in\set{0,\ldots, \blurdim-1}} v_j = 1$ and that $M \shortddist = \shortddist$ for $M = \I_\blurdim$.
    
    \paragraph{Probability of $[\ell,u] \subseteq [s(\jhat - 2), s(\jhat + 2)]$}
    We first show $(v_1, \ldots, v_{\blurdim - 1})$ has at most three nonzero entries, and that these entries must be consecutive. Let $b = \floor{\frac{\ell}{s}}$. Then $\ell\in [b\cdot s, (b+1)\cdot s)$, so the degrees of all non-isolated
    nodes are in $[b\cdot s, (b+2)\cdot s)$. 
    Thus, the only nonzero elements are $v_b, v_{b+1}$, and $v_{b+2}$.

    We next show $\jhat \in \set{b, b+1, b+2}$. By the assumption on $[\ell, u]$ and the definition of $\shortddist$ (in particular, the rounding function $\randround$ in \cref{eq:blur-rand-map}), we have $v_{b} + v_{b+1} + v_{b+2} \geq 6\ccon \sqrt{\log(\frac{n}{s})/n}\cdot \cedp$.
    Therefore, at least one value in $(v_b,v_{b+1},v_{b+2})$ must be at least $2\ccon \sqrt{\log(\frac{n}{s})/n}\cdot \cedp$. Because these are the only nonzero elements in $(v_1,\ldots, v_{\blurdim - 1})$, to show $\jhat \in \set{b, b+1, b+2}$ it suffices to show that, with probability at least $\frac{19}{20}$ we have $|v_j - \wh v_j| < \ccon \sqrt{\log(\frac{n}{s})/n}\cdot \cedp$ for all $j\in[\blurdim - 1]$.
    However, this is immediate by the definition of $\ccon$.

    Since $\jhat\in \set{b, b+1, b+2}$ with probability at least $\frac{19}{20}$, and $[\ell,u]\subseteq [b\cdot s, (b+2)\cdot s]$, this means $[\ell,u]\subseteq \bracks{s(\jhat - 2) , s(\jhat + 2) }$, with probability at least $\frac{19}{20}$.

    \paragraph{Additive error on average degree}
    Assume $[\ell,u]$ contains the degrees of all non-isolated nodes and $[\ell, u] \subseteq [s(\jhat - 2), s(\jhat + 2)]$.
    By $\E_{Y \sim \fakedistnew}[Y] = \E_{X \sim \ddist_G}[X]$ (see \Cref{lem:blur-deg-props-new}) %
    and the law of total expectation,
    \begin{align*}
        \overline{d} \;
        &= 
        \sum_{j\in\set{0,\ldots,\blurdim-1}} sj \cdot v_j %
        = 
        \sum_{j = \jhat - 2}^{\jhat + 2} \bparen{sj\cdot v_{j}} %
        \condnewline \condalign
        = 
        \sum_{i=0}^4
        \paren{s(\jhat - 2) + i\cdot s} \cdot v_{\jhat - 2 + i} %
        \condnewline \condalign
        = 
        \tfrac{k}{n}\cdot s(\jhat - 2) + \sum_{i=1}^4 i \cdot s \cdot v_{\jhat - 2 + i},
    \end{align*}
    where the second equality
    follows from the assumption $[\ell, u] \subseteq [s(\jhat - 2), s(\jhat + 2)]$,
    so $sj \cdot v_j = 0$ for all integers
    $j\not\in[\jhat - 2, \jhat + 2]$.
    Thus, when $k$ is known, with probability at least $\frac{19}{20}$ we can bound the additive error of $\wh d$ as
    \begin{align*}
        \condalign
        \Babs{\wh d - \overline{d}}
        \condnewline
        &=
        \Babs{
        \tfrac{k}{n} s(\jhat - 2) + 
        \condnegspace \condnegspace
        \sum_{i\in[4]} i \!\cdot\! s \!\cdot\! \wh v_{\jhat - 2 + i}
        -
        \tfrac{k}{n} s(\jhat - 2) - 
        \condnegspace \condnegspace
        \sum_{i\in[4]} i \!\cdot\! s \!\cdot\! v_{\jhat - 2 + i}
        } \\
        &=
        \Babs{
        \sum_{i\in[4]} i \cdot s \cdot \wh v_{\jhat - 2 + i}
        -
        \sum_{i\in[4]} i \cdot s \cdot v_{\jhat - 2 + i}
        } \\
        &=
        O\paren{s \cdot \paren{\frac{1}{\sqrt{n}} + \frac{1}{s}}\cdot \cedp}
        = 
        O\paren{\paren{1 + \frac{s}{\sqrt{n}}}\cdot \cedp},
    \end{align*}
    where the first big-O expression follows from \az{\cref{cor:deg-dist-new}}.
\end{proof}

\subsection{Estimating the Parameter of an \texorpdfstring{$\ergraph$}{Erdos-Renyi} Graph}
\label{sec:er-new}

In this section, we show how to privately estimate the parameter $p$ of an $\ergraph$ graph $G\sim G(n,p)$. Our algorithm has near-optimal additive error: up to a multiplicative factor of $O\bparen{\sqrt{\log n}}$, the accuracy of our algorithm matches the additive error required by any \edLNDP algorithm that estimates the parameter of an $\ergraph$ graph (see \cref{thm:er-lower-bd} for the lower bound).

\begin{theorem}[$\ergraph$ parameter estimation]
\label{thm:er-ub-new}
    Let $c > 0$ be some absolute constant.
    Let $\alger$ be \cref{alg:er-new} with parameters $n\in \N$, $\del\in \left( 0, \frac{1}{2} \right]$, and $\eps \geq c\cdot \sqrt{\frac{\log n \log(1/\del)}{n}}$. 
    Then $\alger$ is $(\eps,\del)$-\LNDP for all graphs $G$ on node set $[n]$. Moreover, there exists $\alpha = O\paren{\frac{\sqrt{\log n}}{n} \cdot \cedp}$ such that for all $p\in [0,1]$ and sufficiently large $n$, we have 
    \(
        \Pr\bracks{ \babs{\alger(G) - p } \leq \alpha } \geq \tfrac{2}{3},
    \)
    with the probability taken over the randomness of $\alger$ and $G\sim G(n,p)$.
\end{theorem}

\begin{algorithm}[ht!]
    \caption{$\alger$ for privately estimating the parameter $p$ of an $\ergraph$ graph.}
    \label{alg:er-new}
    
    \begin{algorithmic}[1]
        \Statex \textbf{Parameters:} Privacy parameters $\eps > 0$, $\del\in (0,1]$;
        number of nodes $n\in\N$.
        \Statex \textbf{Input:}
        Graph $G$ on vertex set $[n]$.
        \Statex \textbf{Output:} $\ergraph$ parameter estimate $\wh p\in \R$.
        \State $(\wh x,\wh v) \gets \algconcdeg\bparen{\eps,\del, n, s = \ceil{2\sqrt{3n\ln(10n)}}, G}$. 
        \cond{\Statex} \Comment{\commentstyle{Call to \cref{alg:conc-deg}.}}
        \State Central server returns $\wh p = \frac{\wh x + \wh v}{n}$.
    \end{algorithmic}
\end{algorithm}

\begin{proof}[Proof of \cref{thm:er-ub-new}]
    
    \textit{(Privacy.)} By
    \cref{lem:conc-deg} and postprocessing, $\alger$ is $(\eps,\del)$-\LNDP.

    \textit{(Accuracy.)}
    Let $G\sim G(n,p)$ and $m$ be the number of edges in $G$. The error $|\alger(G) - p|$ comes from two sources: 
    sampling of $G$ and the algorithm's error on $G$. %
    Let $c' > 0$ be some absolute constant to be specified later.
    Let $E_1$ be the event $|m - \binom{n}{2}p| \leq \sqrt{3\ln(10)\binom{n}{2}}$;
    $E_2$ be the event that the degree of each node in $G$ is in the interval $\bbracks{ (n-1) p \pm \sqrt{3n \ln(10n)} }$; and $E_3$ be the event that $\abs{\wh p - m / \binom{n}{2}} \leq c'\cdot \frac{\sqrt{\log n}}{n}\cdot \cedp$.
    It suffices to prove 
    $\Pr[\overline{E_1} \cup \overline{E_3}] < \frac{1}{3}$,
    which follows from showing $\Pr[\overline{E_1}] \leq \frac{1}{10}$,
    $\Pr[\overline{E_2}] \leq \frac{1}{10}$,
    and $\Pr[\overline{E_3} | E_2 ] \leq \frac{1}{10}$,
    and applying the union bound and the law of total probability:
    $
    \Pr[\overline{E_1} \cup \overline{E_3}] \leq \Pr[\overline{E_1}] + \Pr[\overline{E_3}] 
    \leq \Pr[\overline{E_1}] + \Pr[\overline{E_3} \mid E_2] + \Pr[\overline{E_2}].$

    \paragraph*{Bounding $\Pr[\overline{E_1}]$}
    By \cref{lem:bin-tails} (using the term $\sqrt{3np \ln(1/\beta)}$ in Item 2 for sufficiently large $n$), %
    \begin{align*}
        \Pr_{G\sim G(n,p)}[\overline{E_1}] &=     \Pr_{G\sim G(n,p)} \cond{\!} \left[ \abs{m  - p {\cond{\textstyle} \binom{n}{2} }} 
        \cond{\!}
        > 
        \cond{\!}
        \sqrt{3\ln(10) {\cond{\textstyle} \binom{n}{2} }} \right] 
        \cond{\!}
        \leq \cond{\!} \frac{1}{10}.
    \end{align*}
    
    \paragraph{Bounding $\Pr[\overline{E_2}]$}
    The degree of each node in an $\ergraph$ graph is distributed as $\Bin(n-1, p)$. So, by \Cref{lem:bin-tails} (Chernoff--Hoeffding for binomials), %
    a fixed node's degree is in $\bbracks{ (n-1) p \pm \sqrt{3n \ln(10n)} }$ with probability at least $1 - \frac{1}{10n}$ for sufficiently large $n$. Taking a union bound over all $n$ nodes gives
    $\Pr[\overline{E_2}] \leq \frac{1}{10}$.

    \paragraph{Bounding $\Pr[\overline{E_3} | E_2 ]$}
    Conditioning on $E_2$, all nodes' degrees are in an interval of width %
    $s = \ceil{2\sqrt{3n \ln(10n)}}$. Thus, setting $k = n$ in \cref{lem:conc-deg}, %
    there is some absolute constant $c' > 0$ such that for all sufficiently large $n$, we have $\abs{ m - \wh p \binom{n}{2}} \leq c' n\sqrt{\log n} \cdot \cedp$ with probability at least $\frac{9}{10}$. Therefore, $\Pr[\overline{E_3} | E_2] \leq \frac{1}{10}$.
\end{proof}

\subsection{Estimating Clique Size}
\label{sec:clique-new}

In this section, we show how to privately estimate the size of a clique with additive error $O_\delta(1/\eps)$, where accuracy holds under the condition that the graph consists of a clique of size $\Theta(n)$ and isolated nodes.
The error of this \LNDP algorithm matches the error required for solving this problem in the central model, up to a factor of $\sqrt{\log(1/\del)}$.

\begin{theorem}[Clique size estimation]
\label{thm:clique-new}
    Let $c > 0$ be some absolute constant.
    Let $\algcliquesize$  be \cref{alg:clique-new} with parameters $n\in \N$, $\del\in \left( 0, \frac{1}{2} \right]$, and $\eps \geq c\cdot \sqrt{\frac{\log n \log(1/\del)}{n}}$. 
    Then $\algcliquesize$ is $(\eps,\del)$-\LNDP for all graphs $G$ on node set $[n]$. Moreover, there exists some
    $\alpha = O\Big(\frac{\sqrt{\log(1/\del)}}{\eps}\Big)$
    such that, if $G$ is a $K$-clique with $k^* = |K| \geq \frac{n}{10}$ and $n$ is sufficiently large, then
    $\Pr \left[|\algcliquesize(G) - k^*| \leq \alpha\right] \geq \frac{2}{3}$.
\end{theorem}

\begin{algorithm}[ht!]
    \caption{$\algcliquesize$ for privately estimating the size of a $K$-clique.}
    \label{alg:clique-new}
    
    \begin{algorithmic}[1]
        \Statex \textbf{Parameters:} Privacy parameters $\eps > 0$, $\del\in (0,1]$;
        number of nodes $n\in\N$.
        \Statex \textbf{Input:}
        Graph $G$ on vertex set $[n]$.
        \Statex \textbf{Output:} Clique size estimate $\wh k^*$.
        \State $(\wh x, \wh v) \gets \algconcdeg\bparen{\eps,\del, n, s = \ceil{\sqrt{n}\,}, G}$. 
        \cond{\Statex}
        \Comment{\commentstyle{Call to \cref{alg:conc-deg}.}}
        \State Return $\wh k^* \gets
        \frac{\wh x + 1}{2} + 
        \sqrt{\max \cond{\!}\set{0,\paren{\frac{\wh x + 1}{2}}^2 \cond{\!\!}+ n \wh v }}$.
        \cond{\Statex}
        \Comment{\commentstyle{$\max$ ensures that $\sqrt{\cdot}$ is well defined.}}
    \end{algorithmic}
\end{algorithm}

\begin{proof}[Proof of \cref{thm:clique-new}]
    
    \textit{(Privacy.)} By
    \cref{lem:conc-deg} and postprocessing, $\algcliquesize$ is $(\eps,\del)$-\LNDP.

    \textit{(Accuracy.)}
    For $|K|\geq \frac{n}{10}$, the accuracy conditions of  \cref{lem:conc-deg} are satisfied, so with probability at least $\frac{9}{10}$ there is some $\alpha' = O\bparen{\frac{\sqrt{\log(1/\del)}}{\eps}}$ such that %
    $\frac{k^* (k^*-1)}{n} = \frac{k^*}{n}\cdot \wh x + \wh v + \alpha'$. Let $E_1$ be the event that there is such an 
    $\alpha'$, and condition on it. %
    Solving for $k^*$ gives us
    \(
        k^* = \frac{\widehat x+1}{2} + \sqrt{\Big(\frac{\widehat x+1}{2}\Big)^2 + n\widehat v + n\alpha'}.
    \)
    The algorithm outputs the same expression but without the $n\alpha'$ term (and with a truncation at $0$ inside the square root). Since $k^* \ge n/10$, the quantity inside the square root is $\Omega(n)$, so the square-root function is $O(1/\sqrt{n})$-Lipschitz on this range.
    Therefore, removing $n\alpha'$ changes the value by at most $O(\sqrt{n}\cdot |\alpha'|) = O\Bparen{\frac{\sqrt{\log(1/\delta)}}{\eps}}$.
    Thus, with probability at least $\frac{9}{10}$, we have $|\widehat k^* - k^*| \leq \alpha$ for $\alpha = O\Bparen{\frac{\sqrt{\log(1/\delta)}}{\eps}}$.
    
    Combining this bound with the probability of $E_1$ and with the fact that the $\max$ operation will not increase error, we see that with probability at least $2/3$ the estimate satisfies $|\wh k^* - k^*| \leq \alpha$, as claimed.
\end{proof}

\section{Lower Bounds on Error Necessary for \texorpdfstring{\LNDP}{Noninteractive LNDP}}
\label{sec:lbs}

In this section, we prove lower bounds on the additive error required by $(\eps,\del)$-\LNDP algorithms for edge counting and for estimating the parameter $p$ of $\ergraph$ graphs. These results show that our edge-counting algorithm is asymptotically tight on $D$-bounded graphs for all $D=O(\sqrt n)$, and that our $\ergraph$ parameter estimation algorithm is tight up to a factor of $\sqrt{\log n}$.

In \cref{sec:int-imposs}, we show that the same asymptotic lower bounds %
hold for interactive constant-round \LNDP algorithms. %
Thus, our noninteractive \LNDP algorithms remain optimal even with limited interactivity.

\begin{theorem}
[Error for private edge counting]
\label{thm:edge-ct-lb-approx}
    There exists a constant $c>0$ such that, for all $0 \leq \delta \leq \eps \leq c$, $d \in \mathbb{Z}^+$ such that $d \leq \frac{c}{\eps}$, sufficiently large $n\in\N$, and every \edLNDP algorithm $\cA$, the following holds.
    If 
    \begin{conddisplay}
        \Pr_{\text{coins of $\cA$}} [\babs{\cA(G)- |E(G)|} \leq \alpha] \geq \frac{2}{3}
    \end{conddisplay}
    for every $n$-node $d$-bounded graph $G$, then
    \(
        \alpha
        \geq
        \frac{1}{2} \min\set{\frac{dn}{2},\binom n 2}.
    \)
    Furthermore, for $d = \floor{\min\set{\frac{c}{\eps}, n-1}}$, we have $\alpha = \Omega\paren{\min\set{\frac{n}{\eps},n^2}}$.
\end{theorem}
    
\begin{theorem}
[Error for private ER parameter estimation]
\label{thm:er-lower-bd}
    There exists a constant $c > 0$ such that, for sufficiently large $n\in\N$, $\eps\in(0,c]$, $\del \in \left[0, \frac{c\eps}{n\log n} \right]$, and every \edLNDP algorithm $\cA$, then the following holds.
    If 
    \begin{conddisplay}
    \Pr_{\substack{G\sim G(n,p); \\ \text{coins of $\cA$}}}
        \Bigl[
            |\cA(G) - p| \leq \alpha
        \Bigr]
        \geq \frac{2}{3}
    \end{conddisplay}
    for every $p \in [0, 1]$, then \(
        \alpha
        =
        \Omega\paren{\min\set{\frac{1}{n\eps}, 1}}
    \).
\end{theorem}

To prove these theorems (in \cref{sec:empty-reg-indist,sec:er-lb}, respectively), we show it is hard to distinguish the distributions that result from running an \LNDP algorithm on graphs from two families.
Specifically, we upper bound the TV distance\footnote{The \emph{total variation (TV) distance} between distributions $P$ and $Q$ on domain $\cX$ is $\dtv\paren{P, Q} := \sup_{A\subseteq \cX} \abs{P(A) - Q(A)}$.}
between the distributions that result from running an \LNDP algorithm on (1) the empty graph and a random regular graph (\cref{lem:hell-empty-dreg}); and (2) the empty graph and a random $\ergraph$ graph (\cref{thm:er-empty-tv-bd}).
We then use these bounds on TV distance to prove \Cref{thm:edge-ct-lb-approx,thm:er-lower-bd}.

Before proving these theorems, we present \cref{lem:bhatt-dp-bound}, which shows a relationship between Bhattacharyya distance and $(\eps,\del)$-indistinguishability.
To bound the TV distance between output distributions, we in fact bound their Bhattacharyya distance, which is a function of the Hellinger distance between two distributions.
Bhattacharyya distance has the nice property that it \emph{tensorizes}---that is, the Bhattacharyya distance between product distributions is the sum of the distances between each coordinate of the product distributions.
Hellinger distance is closely tied to TV distance,
so converting a statement about Bhattacharyya distance to a statement about TV distance is straightforward.

Formally, for probability distributions $P$ and $Q$, let $\haff(P\|Q)=\int_x \sqrt{P(x)Q(x)}dx$ denote Hellinger affinity (also known as the Bhattacharyya coefficient).
The quantity $\bhatt(P, Q) = -\ln \haff(P\|Q)$ is called the \emph{Bhattacharyya distance}, or the R{\'e}nyi-$\frac{1}{2}$ divergence, between $P$ and $Q$.

\begin{lemma}
\label{lem:bhatt-dp-bound}
    For all $\eps>0$ and $\delta \in [0,1)$,
    if $P$ and $Q$ are $(\eps,\delta)$-indistinguishable distributions, then 
    \begin{align*}
    \bhatt\paren{P,Q} &\leq \ln\paren{\frac{e^{\eps/2} + e^{-\eps/2}}{2}} + \ln\paren{\frac{1}{1-\delta}} 
    \\
    &\leq \min\set{\frac{\eps^2}{8}, \frac{\eps}{2}} + \frac{\delta}{1-\delta} \, .    
    \end{align*}
    The first inequality is tight when $P$ and $Q$ are the output distributions of the \textit{leaky randomized response} mechanism on inputs 0 and 1. 
    Assuming $\delta$ is bounded away from 1, the right-hand side is $\Theta\paren{\min\set{\eps^2,\eps} +\delta}$.
\end{lemma}

\begin{proof}[Proof of \cref{lem:bhatt-dp-bound}]
    Recall that 
    \begin{conddisplay}
        \bhatt(P,Q) = -\ln \haff(P \| Q).
    \end{conddisplay}
    Thus, it suffices to show
    \[
        \haff(P \| Q) \geq
        \frac{2(1-\del)}{e^{\eps/2} + e^{-\eps/2}}.
    \]

    By the simulation lemma of \cite{KairouzOV15} (\cref{lem:sim-via-lrr}), we have
    $\haff(P \| Q)\geq \haff\bparen{\lrr(0) \| \lrr(1)}$, where $\lrr$ is the leaky randomized response functionality from \cite{KairouzOV15,MurtaghV18} (see \cref{def:leaky-rr}). (This inequality holds since the squared Hellinger distance is an $f$-divergence.)
    By direct calculation, we get 
    \begin{align*}
    \cond{&}
        \haff\paren{\lrr_{\eps,\delta}(0) \| \lrr_{\eps,\delta}(1)}
        \cond{\\ &}
        =
        \delta \cdot 0 + 2 \sqrt{(1-\delta)\frac{e^\eps}{e^\eps + 1} \cdot (1-\delta)\frac{1}{e^\eps + 1}  }
        = \frac{2(1-\del)}{e^{\eps/2} + e^{-\eps/2}}.
    \end{align*}
    The second term in \cref{lem:bhatt-dp-bound} follows immediately since for all $\eps > 0$ and $\del\in[0,1)$,
    we have 
    \condeqnsand{
        \ln\paren{\frac{e^{\eps/2} + e^{-\eps/2}}{2}} \leq \min\set{\frac{\eps^2}{8}, \frac{\eps}{2}}
    }
    {
        \ln\paren{\frac{1}{1-\delta}} \leq \frac{\delta}{1-\delta}.
            \qedhere
    }
\end{proof}

\subsection{Error Needed for Counting Edges}
\label{sec:empty-reg-indist}

In this section, we bound the TV distance between the output distributions of an \LNDP algorithm on the empty graph $\gempty$ and on a uniformly random $d$-regular graph, for $d\approx 1/\eps$ (\Cref{lem:hell-empty-dreg}), and then use this bound to prove \cref{thm:edge-ct-lb-approx}. Let $\gdreg$ denote the uniform distribution over $d$-regular graphs on node set $[n]$.

\begin{lemma}[Random $d$-regular graphs are indistinguishable from the empty graph]
\label{lem:hell-empty-dreg}
    Let $d,n\in \N$.
    Let $\eps > 0$ and $\del \in \big[0,\frac{1}{d e^{d \eps}}\big)$, and set $\tilde \del = d\del \cdot e^{d\eps}$.
    Let $\cA$ be an \edLNDP algorithm.
    When $G\sim\gdreg$ is a random $d$-regular graph, the Bhattacharyya distance between the distributions of the pairs $(G,\cA(G))$ and $(G,\cA(\gempty))$ is bounded. That is, for $G\sim\gdreg$:
    \begin{equation}
    \label{eq:ex-hell-dreg}
        \bhatt
        \Bparen{\cond{\!}\bparen{G,\cA(G)}, \bparen{G,\cA(\gempty)} \cond{\!}}
        \leq 
        \frac{1}{
            1 \cond{\!}-\cond{\!}\frac{d}{n}
            }
        \ln\left(\cond{\!}\frac{e^{d\eps/2} \cond{\!}+\cond{\!} e^{-d\eps/2}}{2(1-\tilde \delta)} \cond{\!}\right) 
            .
    \end{equation}
    Furthermore, if  $d\leq n/2$ then, as $d\eps\to 0$ and $d \del\to 0$, when $G\sim \gdreg$, we have
    \begin{align*}
        \bhatt \Bparen{\bparen{G,\cA(G)}, \bparen{G,\cA(\gempty)}}  
        = O\Bparen{(d\eps)^2 + d\del}
    \end{align*}
    and
    \begin{align*}
        \dtv\Bparen{\bparen{G,\cA(G)}, \bparen{G,\cA(\gempty)}} 
        = O\paren{d\eps + \sqrt{d\del}}\, .
    \end{align*}
\end{lemma}

This lemma states that no outside analyst, seeing the output of an \LNDP algorithm, can tell apart a random $d$-regular graph $G$ from  the empty graph, \textit{even given access to the graph $G$}. We use this strong formulation when extending the result to interactive protocols in \Cref{sec:int-imposs}.

We use \cref{def:wt-hell} to quantify the (in)distinguishability between these distributions of outputs.
Recall from \cref{def:lndp} that every \edLNDP algorithm
$\cA$
is specified by a sequence of randomizers $\vcR_\pubrand:= \cR_{1,\pubrand},\ldots, \cR_{n,\pubrand}$
and a postprocessing algorithm $\cP$.
(We use the notation $\vcR := \cR_{1},\ldots, \cR_{n}$ when there is no---or fixed---public randomness.)
Intuitively, this ``weight function'' captures how much the output of randomizer $i$ changes when run on the empty graph and when run on a graph where node $i$ has edge set $S$.

\begin{definition}
\label{def:wt-hell}
For $S\subseteq [n] =: V$, define the \emph{weight of node $i$ for edge set $S$} as
\begin{equation}
\label{eq:wt-hell}
    \whell_i(S) := 
    \bhatt\bparen{ \cR_i(S) , \cR_i(\varnothing) }.
\end{equation}
\end{definition}

We also use the following definition in our proof of \cref{lem:hell-empty-dreg}.

\begin{definition}[Starpartite graph]
\label{def:starpartite-graph}
    Let $n\in \N$ and $T\subseteq[n]$.
    A \emph{starpartite graph on nodes $[n]$ with center $T$}, denoted $S_T$, has edge set $\{\{i,j\}: i\in T, j\in [n],  i\neq j \}$.\footnote{That is, every node in $T$ is a star (with an edge to every node in the graph), and there are no additional edges in the graph.}
    We also call this graph \emph{$d$-starpartite,} where $d= |T|$.
\end{definition}

Our proof of \cref{lem:hell-empty-dreg} makes use of the following intuition: A $d$-starpartite graph is node distance $d$ from the empty graph, so running an ($\eps',\del')$-\LNDP algorithm on this graph and on the empty graph must result in distributions of outputs that are similar---namely, they are $(d\eps', d\del'\cdot e^{d\eps'})$-indistinguishable.
Moreover, most nodes in this $d$-starpartite graph have the same ``view'' as in a $d$-regular graph (e.g., aside from the star centers, each node has edges to a uniformly random set of $d$ nodes). In particular, only the $d$ star centers (i.e., a $\frac{d}{n}$-fraction of nodes) have a different degree in the $d$-starpartite graph, as compared to in a $d$-regular graph.
Intuitively, this means an \LNDP algorithm returns similar distributions of outputs when run on a random $d$-regular graph and on a $d$-starpartite graph, and thus also when run on the empty graph. We now formalize this intuition.

\begin{proof}[Proof of \cref{lem:hell-empty-dreg}]
    By postprocessing and the convexity of $\bhatt$, where $\vcR_{\pubrand} := (\cR_{1,\pubrand},\ldots,\cR_{n,\pubrand})$ denotes the randomizers of $\cA$ with public randomness $\pubrand$, there exists some fixed public randomness $\pubrand^*$ such that, where $G\sim \gdreg$,
    \begin{align*}
        \cond{&}\bhatt\left(\bparen{G,\cA(G)}, \bparen{G, \cA(\gempty)}\right)        
        \cond{\\ & \qquad \qquad } \leq
        \bhatt\Bparen{\bparen{G,\vcR_{\pubrand^*}\paren{G}}, \bparen{H,\vcR_{\pubrand^*}(\gempty)}}.
    \end{align*}
    For the remainder of this proof, define $\vcR := \vcR_{\pubrand^*}$.
    
    The tensorization property of the inner product, which defines the Bhattacharyya coefficient, means that the Bhattacharyya distance between product distributions is the sum of the distances between the individual terms. As a result, for every fixed graph $G$, we have 
    \[
        \bhatt\paren{\vcR(G),\vcR(\gempty)}
        =
        \sum_{i\in[n]} \whell_i(\neigh{G}{i}),
    \]
    where $\whell_i$ are the weights from \Cref{def:wt-hell}.

    We first consider the sum of these weights when $G$ is starpartite.
    Let $T\subseteq [n]$ such that $|T| = d$, and let $S_T$ denote the $d$-starpartite graph with center $T$.
    Let $\tilde \eps = d\eps$ and $\tilde \del = d\del e^{\tilde\eps}$ (as in the statement of \cref{lem:hell-empty-dreg}).
    Because $S_T$ and the empty graph differ only on the edges incident to the nodes in $T$, these graphs are at node distance $d$.
    Since $\cA$ is \edLNDP (for all settings of public randomness), group privacy implies that the randomizers $\vcR := (\cR_1,\ldots,\cR_n)$ satisfy $\vcR(S_T) \approx_{\tilde\eps, \tilde\del} \vcR(\gempty)$, and therefore, applying \Cref{lem:bhatt-dp-bound},
    \begin{equation}
        \label{eq:expected-weights-starpartite}
        \sum_{i\in[n]} \whell_i\bigl(\neigh{S_T}{i}\bigr)
        =
        \bhatt\paren{\vcR(S_T), \vcR(\gempty)}
        \leq \ln \paren{\frac{e^{\tilde \eps/2} + e^{-\tilde \eps/2}}{2(1-\tilde \delta)}} \, .
    \end{equation}

    We now turn to random $d$-regular graphs.
    Since Bhattacharyya distance is convex, by Jensen's inequality for $G\sim \gdreg$, we have
    $\bhatt
    \bparen{\bparen{G,\vcR(G)}, \bparen{G,\vcR(\gempty)}}\leq \Ex_{G\sim \gdreg} \big[\bhatt\bparen{\vcR(G), \vcR(\gempty)}\big]$.
    We can relate this to the expected sum of the weights for a uniformly random $d$-regular graph
    from $\gdreg$, and bound that expectation as
    \begin{align}
    \begin{split}
    \label{eq:bd-regular}
    \BEx_{G\sim \gdreg}
        \big[ \bhatt\bparen{\vcR(G), \vcR(\gempty)} \big]
        &= 
        \BEx_{G\sim \gdreg}
        \bigg[
        \sum_{i\in[n]} \whell_i\bigl(\neigh{G}{i}\bigr)\bigg]
        \cond{\\ &}
         =
        \sum_{i\in[n]}
        \BEx_{\substack{S\subseteq[n] \setminus\set{i} \\ |S| = d}}
        \bracks{ \whell_i(S) },
    \end{split}
    \end{align}
    where the final term follows by linearity of expectation, with the expectations in the left and center over a uniformly selected $d$-regular graph, and the expectation on the right over the neighbors of a given node $i$, which form a uniformly random subset of $[n]\setminus \set{i}$ with size $d$.

    We now exploit the fact that the view of any given node is distributed nearly identically in a random $d$-regular graph and in a random $d$-starpartite graph.
    Specifically, if the set $T$ of $d$ star centers is selected uniformly at random, then the neighborhood of a given node $i$ in the graph $S_T$, \textit{conditioned on $i\not\in T$}, is a uniformly random set of size $d$, as it would be when $G\sim\gdreg$. Thus, for every node $i$ in a random regular graph, 
    \begin{align*}
    \BEx_{\substack{S\subseteq[n] \setminus\{i\} \\ |S| = d}}
    \bracks{
        \whell_i(S)
    }
    &= \BEx_{\substack{T \subseteq[n] \\ |T|=d }} \bracks{\whell_i(T) \mid i \not \in T} 
    \cond{\\ &}
    = \BEx_{\substack{T \subseteq[n] \\ |T|=d }} \bracks{\whell_i(N_i^{S_T}) \big| i \not \in T}
    \cond{\\ &}
    \leq \BEx_{\substack{T \subseteq[n] \\ |T|=d }} \bracks{\whell_i(N_i^{S_T})} \cdot \frac{1}{\Pr(i \not\in T)},
    \end{align*}
    where the inequality uses that $\whell_i(\cdot)$ takes only nonnegative values. 
    We can now ``\emph{splice}'' together these expressions for distances between per-randomizer outputs on random inputs to obtain an expression for the distances between global outputs on random inputs.
    We take the sum over the nodes $i$ to bound the distance, and use the fact that we have $i \not \in T$ with probability $(1-\frac d n)$:
    \begin{align*}
    \cond{&}\BEx_{G\sim \gdreg}
        \big[ \bhatt\bparen{\vcR(G), \vcR(\gempty)} \big]
         \cond{ \\&\qquad}
         =
        \sum_{i\in[n]}
        \BEx_{\substack{S\subseteq[n] \setminus\{i\} \\ |S| = d}}
        \bracks{
            \whell_i(S)
        }
        \cond{\\}
        &\cond{\qquad}\leq 
        \sum_{i\in[n]}  \BEx_{\substack{T \subseteq[n] \\ |T|=d }} \bracks{\whell_i(N_i^{S_T})} \cdot \frac{1}{\Pr(i \not\in T)}
        \\
        &\cond{\qquad}= 
        \frac{1}{1-\frac d n} \cdot  \BEx_{\substack{T \subseteq[n] \\ |T|=d }} \bracks{\bhatt\bparen{\vcR(S_T) , \vcR(\gempty)}} \, .
    \end{align*}
    The last equality holds by the tensorization of Bhattacharyya distance.
    By \Cref{eq:expected-weights-starpartite} (a bound on the distance between outputs on empty and starpartite graphs), $\Ex_{G\sim\gdreg} \big[\bhatt\bparen{\vcR(G), \vcR(\gempty)}\big]$ is at most $\frac{1}{1-\frac d n} \cdot \ln \paren{\frac{e^{\tilde \eps/2} + e^{-\tilde \eps/2}}{2(1-\tilde \delta)}}$, as desired.

    Under the assumptions that $d\leq \frac n 2$, and $\tilde \eps$, $\tilde \delta$ go to 0, this expression simplifies to $O(\tilde\eps^2 + \tilde\delta)$ (since $1-\frac{d}{n}\geq \frac 12$ and $e^{\tilde \eps/2} + e^{-\tilde\eps/2} = 2 + \tilde \eps^2 + O(\tilde\eps^4)$). We can further simplify this using the observation that, as $\tilde\eps = d\eps$ goes to zero, $\tilde\del = d\del \exp(d\eps)$ is $O(d\del)$. This yields the desired asymptotic bound on the Bhattacharyya distance $\bhatt$.

    It remains to bound the TV distance between the pairs $(G,\vcR(G))$ and $(G,\vcR(\gempty))$, where $G\sim\gdreg$.
    Recall that     
    for any distributions $P$ and $Q$, we have $\dtv\paren{P,Q} \leq \sqrt{2(1-\exp(-\bhatt(P,Q)))}$. 
    Substituting in the bound on the Bhattacharyya distance, and using the fact that $\exp(-x)=1-O(x)$ for bounded $x$, shows that the TV distance is $O(\tilde\eps +\sqrt{\tilde\del}) = O(d\eps +\sqrt{d\delta})$, as desired. 
\end{proof}

\begin{proof}[Proof of \cref{thm:edge-ct-lb-approx}]
    Fix
    a graph size $n$ and positive (integer) degree $d\leq n-1$, and let $\cA$ be an \edLNDP algorithm that estimates the edge count with additive error at most $\alpha$, with probability at least $\frac{2}{3}$,
    on $d$-bounded graphs.
    Observe that a $d$-regular graph has $\frac{dn}{2}$ more edges than $\gempty$.
    Thus, if $\alpha < \frac 1 2 \cdot \frac{dn}{2}$,  algorithm $\cA$ can be used to correctly determine, with probability at least $\frac 2 3$, if the input is $\gempty$ or a random graph in $\gdreg$. The TV distance between $(G,\cA(G))$ and $(G,\cA(\gempty))$, where $G\sim \gdreg$,
    is thus at least $\frac 2 3 - \frac 1 3 = \frac 1 3$. 
    By linearity of expectation, there exists a fixed value $\pubrand_0$ of $\cA$'s public randomness (if it uses any) such that, where $G\sim \gdreg$, $\dtv\bbracks{(G,\vcR_{\pubrand_0}(G)), (G,\vcR_{\pubrand_0}(\gempty))}\geq \frac 1 3$.
    
    On the other hand, because $\cA$ remains differentially private even when the public randomness is fixed, 
    \Cref{lem:hell-empty-dreg} shows that, where $G\sim \gdreg$, 
    $\dtv\bbracks{(G,\vcR_{\pubrand_0}(G)), (G,\vcR_{\pubrand_0}(\gempty))}$
    is $O(d\eps + \sqrt{d\del})$. 
    Let $a>0$ be a constant such that the TV distance is at most $\frac 1 4$ when $d\eps + \sqrt{d\del}<a$. If $\eps \leq \frac{a}{2d}$ and $\del \leq \frac{a^2}{4d}$, then we get a contradiction with the TV lower bound implied by $\cA$'s error guarantee. 

    Let $c = \min\set{\frac{a}{2},\frac{a^2}{4}}$. For all even $d\leq n-1$, there is a $d$-regular graph on $n$ nodes. Setting $d = 2 \cdot \floor{\min\set{\frac{c}{2\eps}, \frac{n-1}{2}}}$,  %
    we get a lower bound of $\alpha = \Omega\paren{\min\set{\frac{n}{\eps}, n^2}}$ when $0\leq \delta\leq \eps\leq c$, as desired.
\end{proof}

\subsection{Error Needed for Estimating the Parameter of an \texorpdfstring{$\ergraph$}{Erdős-Rényi} Graph}
\label{sec:er-lb}

We now prove \cref{thm:er-empty-tv-bd}, which upper bounds the TV distance between the distributions that result from running an algorithm on the empty graph and a random $G(n,p)$ $\ergraph$ graph, for sufficiently small $p$.
Because the TV distance between these distributions is small, distinguishing ER graphs from the empty graph (which has 0 edges and corresponds to an ER graph with $p=0$) must be difficult, giving us the lower bound on privately estimating the parameter $p$ of an $\ergraph$ graph in \cref{thm:er-lower-bd}.

\begin{lemma}
\label{thm:er-empty-tv-bd}
    Let $n\in \N$ and $p \in \bracks{0,1}$.
    There exists a constant $c \geq 1$ such that for all $a\in(0,1]$, $\eps \in (0,1]$, and $\del \in \left[0,
    \cw{a^2}/\paren{25\cdot\dmax e^{\dmax\eps}}
    \right)$ where $\dmax = \ceil{np + \max\bigl\{3\ln(10n\cw{/a}), \sqrt{3np\ln(10n\cw{/a})}\bigr\}}$, the following holds. If $\cA$ is $(\eps,\del)$-\LNDP,
    $n\in\N$ is sufficiently large, $\dmax \leq \frac{n}{2}$, and $p\leq \frac{a^2}{6 \cdot n\eps}$, then for $G\sim G(n,p)$ the TV distance between the distributions of the pairs $(G,\cA(G))$ and $(G,\cA(\gempty))$ satisfies
    \[
        \dtv
        \Bparen{\bparen{G,\cA(G)},\bparen{G,\cA(\gempty)}}
        \leq c\cdot a.
    \]
\end{lemma}

To prove \cref{thm:er-empty-tv-bd} we use \cref{lem:symmetric-graphs-to-starpartite}, which relates the distance between the empty graph and a random graph with maximum degree $\dmax$, and the distance between the empty graph and a random starpartite graph.
It generalizes an intermediate step used in our proof of \cref{lem:hell-empty-dreg}.

\begin{lemma}
\label{lem:symmetric-graphs-to-starpartite}
    Let $\cG$ be a distribution on $n$-node undirected graphs of maximum degree $\dmax<n$ that is symmetric under permutation of the nodes, and
    let $\ddist_\cG$ denote the distribution (on $\{0,\ldots,\dmax\}$) of the degree of a node in a graph selected according to $\cG$. 
    Let $\cA$ be an \edLNDP algorithm with randomizers $\vcR := (\cR_1,\ldots, \cR_n)$ with no (or fixed) public randomness.
    Then for $G\sim \cG$ the Bhattacharyya distance between the distributions of the pairs $(G,\vcR(G))$ and $(G,\vcR(\gempty))$ satisfies
    \begin{align*}
        \cond{&}
        \bhatt\Bparen{\bparen{G,\vcR(G)}, \bparen{G,\vcR(\gempty)} }
        \cond{\\ & \quad}\leq 
        \paren{\frac{1}{1-d_{max}/n}}
        \BEx_{\substack{X\sim \ddist_\cG \\ T\subseteq[n], |T|=X}} \bhatt\Bparen{\vcR(S_T), \vcR(\gempty)} \, .
    \end{align*}
    Furthermore, where $a\in(0,1]$ is an arbitrary constant, for $\dmax \leq \frac n 2$, $\eps\leq 1$, and $\delta \leq \cw{a^2}/\paren{25 \cdot \dmax e^{\dmax \eps}} $, we have 
    \[
        \bhatt\Bparen{\bparen{G,\vcR(G)}, \bparen{G,\vcR(\gempty)} }
        \leq
        \tfrac{\cw{a^2}}{12} + \tfrac {1}{8} \eps^2 \Ex_{X\sim \ddist_{\cG}}\paren{X^2}. 
    \]
\end{lemma}

\begin{proof}[Proof of \cref{lem:symmetric-graphs-to-starpartite}]
    We begin by proving the first inequality.
    This proof borrows ideas from the second half of the proof of \cref{lem:hell-empty-dreg}.
    Since the Bhattacharyya distance is convex, by Jensen's inequality for $G\sim \cG$ we have
    $\bhatt\bparen{(G,\vcR(G)), (G,\vcR(\gempty))} \leq \Ex_{G\sim \cG} \bhatt\bparen{\vcR(G), \vcR(\gempty)}$.
    We relate this to the expected sum of the weights for a random graph $G\sim \cG$, and bound that expectation as
    \begin{align}
    \begin{split}
    \label{eq:bd-dmax}
        \BEx_{G\sim \cG}
        \bhatt\bparen{\vcR(G), \vcR(\gempty)}
        &= 
        \BEx_{G\sim \cG}
        \bigg[
        \sum_{i\in[n]} \whell_i\bigl(\neigh{G}{i}\bigr)\bigg]
        \cond{\\ &}=
        \sum_{i\in[n]}
        \BEx_{\substack{S\subseteq[n] \setminus\set{i} \\ |S| \sim \ddist_\cG}}
        \bracks{ \whell_i(S) },
    \end{split}
    \end{align}
    where the final term follows by linearity of expectation, with the expectations in the left and center over the selection of $G\sim \cG$, and the expectation on the right over the neighbors of a given node $i$, which form a uniformly random subset of $[n]\setminus \set{i}$ with size $d\sim\ddist_\cG$.

    We can now exploit the fact that the view of any given node $i$ is distributed nearly identically in a random $d$-starpartite graph, for $d\sim\ddist_\cG$.
    Specifically, if the set $T$ of $d$ star centers is selected uniformly at random, then the neighborhood of a given node $i$ in the graph $S_T$, \textit{conditioned on $i\not\in T$}, is a uniformly random set of size $d\sim \ddist_\cG$, as it would be when $G\sim\cG$. Thus, for every node $i$, where $\whell_i(\cdot)$ is defined in \cref{eq:wt-hell},
    \begin{align*}
    \label{eq:exp-weight-dmax}
        \BEx_{\substack{S\subseteq[n] \setminus\{i\} \\ |S| \sim \ddist_\cG}}
        \bracks{
            \whell_i(S)
        }
        &= \BEx_{\substack{T \subseteq[n] \\ |T|\sim\ddist_\cG }} \bracks{\whell_i(T) \mid i \not \in T} 
        \cond{\\&}
        = \BEx_{\substack{T \subseteq[n] \\ |T|\sim\ddist_\cG}} \bracks{\whell_i(N_i^{S_T}) \big| i \not \in T}
        \cond{\\&}
        \leq \BEx_{\substack{T \subseteq[n] \\ |T|\sim\ddist_\cG }} \bracks{\whell_i(N_i^{S_T})} \cdot \frac{1}{\Pr(i \not\in T)},
    \end{align*}
    where the inequality uses that $\whell_i(\cdot)$ takes only nonnegative values. We can now take the sum over the nodes $i$ to bound the distance, using the fact that the event $i \not \in T$ occurs with probability at most $(1-\frac{\dmax}{n})$: 
    \begin{align*}
    \begin{split}
    \cond{&}\BEx_{G\sim \cG}
        \bhatt\bparen{\vcR(G), \vcR(\gempty)}
         \cond{\\ & \qquad\qquad}=
        \sum_{i\in[n]}
        \BEx_{\substack{S\subseteq[n] \setminus\{i\} \\ |S| \sim \ddist_\cG}}
        \bracks{
            \whell_i(S)
        }
        \cond{\\} &\cond{\qquad\qquad}\leq 
        \sum_{i\in[n]}  \BEx_{\substack{T \subseteq[n] \\ |T|\sim\ddist_\cG }} \bracks{\whell_i(N_i^{S_T})} \cdot \frac{1}{\Pr(i \not\in T)}
        \\
        &\cond{\qquad\qquad}\leq  
        \frac{1}{1-\frac{\dmax}{n}} \cdot  \BEx_{\substack{T \subseteq[n] \\ |T|\sim\ddist_\cG }} \bracks{\bhatt\bparen{\vcR(S_T) , \vcR(\gempty)}} \, ,
    \end{split}
    \end{align*}
    completing the proof of the first inequality in \cref{lem:symmetric-graphs-to-starpartite}.

    We now prove the second inequality. 
    Because $\dmax\leq \frac{n}{2}$ and thus $1/(1-\frac{\dmax}{n}) \leq 2$, it suffices to show that the expectation in the first expression is bounded above by $\tfrac{\cw{a^2}}{24} + \tfrac {1}{8} \eps^2 \Ex_{X\sim \ddist_{\cG}}\paren{X^2}$.
    
    Fix $T\subseteq [n]$ such that $t := |T| \leq \dmax$. Let $\tilde \del = t\del e^{t\eps}.$ 
    Then $\tilde \del \leq \frac{\cw{a^2}}{25}$. 
    By group privacy and \cref{lem:bhatt-dp-bound},
    \[
        \bhatt\bparen{\vec \cR(S_T), \vcR(\gempty)}
        \leq
        \frac{t^2 \eps^2}{8} + \frac{\tilde \del}{1 - \tilde \del }
        \leq
        \frac{t^2\eps^2}{8} + \frac{\cw{a^2}}{24}.
    \]
    By the law of total expectation,\cond{\let\qed\relax}
    \begin{align*}
        \cond{&} \BEx_{\substack{X\sim \ddist_\cG \\ T\subseteq[n], |T|=X}} \bhatt\Bparen{\vcR(S_T), \vcR(\gempty)}
        \cond{\\} 
        &= \cond{\!\!\!\!}\sum_{\!t\in\set{\!0,\ldots,\dmax\!}} \cond{\!\!} \BEx_{\substack{T\subseteq[n] \\ |T| = t}} \cond{\!\!} \bracks{ \bhatt\bparen{\vcR(S_T), \!\vcR(\gempty) \cond{\!}} \Big| |T|\cond{\!}=\cond{\!} t} \cond{\!} \Pr_{|T|\sim\ddist_\cG}[|T| \cond{\!}=\cond{\!} t] \\
        &\leq 
        \sum_{t\in\set{0,\ldots,\dmax}} \paren{ \frac{t^2\eps^2}{8} + \frac{\cw{a^2}}{24} }\cdot \Pr_{|T|\sim\ddist_\cG}[|T| = t] \\
        &= \frac{\cw{a^2}}{24} + \frac{\eps^2}{8} \BEx_{|T|\sim\ddist_\cG}(|T|^2). \condelse{\qquad \qquad \qquad \qquad \qquad \qquad \quad\;\; \square}{\qedhere}
    \end{align*}
\end{proof}

We now prove \cref{thm:er-empty-tv-bd}. We use the following definition in this proof.

\begin{definition}[Bounded-degree $\ergraph$ graphs]
\label{def:bdd-er}
    Let $n\in \N$, $p\in[0,1]$, and $\dmax\in \Z^{\geq 0}$.
    Define $H(n,p,\dmax)$, the distribution of \emph{$\ergraph$ graphs with maximum degree at most $\dmax$}, as the distribution $G\sim G(n,p)$ conditioned on the event that all nodes in $G$ have degree at most $\dmax$.
\end{definition}

\begin{proof}[Proof of \Cref{thm:er-empty-tv-bd}]
    For the setting of $\dmax$ in \cref{thm:er-empty-tv-bd}, a standard bound on the maximum degree of an $\ergraph$ graph (\Cref{thm:er-tails}) shows that, with  probability greater than $1-\cw{a}/6$, $G\sim G(n,p)$ has maximum degree at most $\dmax$; that is, it lies in the support of $H(n,p,\dmax)$. Hence, it suffices to show there is some constant $c\geq 1$ such that, where $H\sim H(n,p,\dmax)$, the TV distance between pairs $(H,\cA(H))$ and $(H,\cA( \gempty ))$ is at most $\frac{c\cdot a}{2}$. 
    By postprocessing and the convexity of $\bhatt$, where $\vcR_{\pubrand} := (\cR_{1,\pubrand},\ldots,\cR_{n,\pubrand})$ denotes the randomizers of $\cA$ with public randomness $\pubrand$, there exists some fixed public randomness $\pubrand$ such that the first inequality holds in
    \begin{align*}
    \begin{split}
        \cond{&}\bhatt\Bparen{(H,\cA\bparen{H}), (H,\cA\bparen{\gempty})}
        \cond{\\ & \qquad}
        \leq        \bhatt\Bparen{(H,\vcR_\pubrand\bparen{H}), (H,\vcR_\pubrand\bparen{\gempty})}
        \cond{\\ & \qquad} \leq
        \textstyle
        \frac {\cw{a^2}} {12} + \frac{1}{8} \eps^2
        \BEx_{X\sim \ddist_{H(n,p,\dmax)}} \paren{X^2},
    \end{split}
    \end{align*}
    with the second inequality holding by \Cref{lem:symmetric-graphs-to-starpartite} (since the distribution $H(n,p,\dmax)$ is symmetric under node permutations), where $\ddist_{H(n,p,\dmax)}$ is the distribution of the degree of any given node in a graph drawn from $H(n,p,\dmax)$.

    We now bound the expectation on the right-hand side above.
    We claim that $\ddist_{H(n,p,\dmax)}$ is stochastically dominated by the distribution $\Bin(n,p)$. This follows from Harris' inequality \cite{Harris1960}
    for product measures. The special case we need states that, for any product distribution on $\{0,1\}^N$, every two events $A,B\subseteq \{0,1\}^N$ that are \textit{decreasing} (i.e., closed under switching $1$s to $0$s) are positively correlated, that is $P(A\cap B) \geq P(A)P(B)$. Taking $N=\binom n2$ and interpreting bit strings of length $N$ as the edge list of a graph on $n$ nodes (where $0$ and $1$ indicate the absence and presence of an edge, respectively), and setting $A=\{G: \text{max-degree}(G) \leq \dmax\}$ and $B = \set{G: \deg_G(u) \leq k}$ for fixed node $u$ and integer $k$, we see that the CDF of the degree distribution only increases when we condition on $A$.

    Thus, the expectation of $X^2$, for $X\sim \ddist_{H(n,p,\dmax)}$, is at most $(n-1)^2p^2 + (n-1)p(1-p)$, which is the expected square of a random draw from $\Bin(n-1,p)$. By assumption $np \eps < \frac {\cw{a^2}}{6}$, so we have $\eps^2((n-1)^2p^2 + (n-1)p(1-p)) \leq 2np\eps \leq \frac{\cw{a^2}}{3}$. Consequently, 
    $\bhatt\bparen{(H,\cA(H), (H,\vcR(\gempty))} \leq \frac{\cw{a^2}}{12} + \frac{\cw{a^2}}{24} < \frac{\cw{a^2}}{6}$.

    It remains to bound the TV distance. Recall that     
    for any distributions $P$ and $Q$, we have $\dtv\paren{P,Q} \leq \sqrt{2(1-\exp(-\bhatt(P,Q)))}$. Substituting in the bound on the Bhattacharyya distance, and using the fact that $\exp(-x)=1-O(x)$ for bounded $x$, shows there is some $c\geq 1$ such that for all choices of $a$ the TV distance is at most $c\cdot a$, as desired. 
\end{proof}

\begin{proof}[Proof of \cref{thm:er-lower-bd}]
    Fix $a>0$ small enough that $c\cdot a < \frac{1}{3}$, a graph size $n$, and $\ergraph$ parameter $p= \frac{\cw{a^2}}{6\cdot n\eps}$, and let $\cA$ be an \edLNDP algorithm that estimates $p$ with error $\alpha$. If $\alpha < \frac{p}{2}$,  algorithm $\cA$ can be used to correctly determine, with probability at least $2/3$ over the randomness of $\cA$ and $G\sim G(n,p)$, if the input is $\gempty$ or a random graph $G\sim G(n,p)$. The TV distance between $(G,\cA(G))$ and $(G,\cA(\gempty))$, where $G\sim G(n,p)$, is thus at least $\frac 2 3 - \frac 1 3 = \frac 1 3$.

    On the other hand, \cref{thm:er-empty-tv-bd} shows that, where $G\sim G(n,p)$, $\dtv\bbracks{
        (G,\cA(G)), (G,\cA(\gempty))
        }
        < \frac{1}{3}$, which contradicts the TV lower bound implied by $\cA$'s error guarantee.
    (Note that, for $p = \frac{\cw{a^2}}{6n\eps}$, $n$ sufficiently large, and $\eps$ at most a sufficiently small constant, we have $\dmax\leq \frac{\ln(10n/a)}{\eps}$, so \cref{thm:er-empty-tv-bd} holds for all $\del \in \left[0, a^2 \eps / (25\cdot (10n/a)\ln (10n/a)) \right)$.)
    Thus, for sufficiently large $n$, $\cA$ must estimate $p$ with additive error $\alpha \geq \frac{p}{2} = \Omega\paren{\frac{1}{n\eps}}$.
\end{proof}

\subsection{Impossibility Results for Interactive LNDP}
\label{sec:int-imposs}

Our impossibility results for edge counting and \ergraph parameter estimation also extend to \emph{interactive} LNDP algorithms. An interactive LNDP algorithm proceeds by rounds, in which each node runs a 
randomizer
that takes as input its %
neighborhood, public randomness, and outputs from itself and other nodes in previous rounds. Node privacy requires the overall transcript of randomizer outputs, choices of nodes, and choices of randomizers to be $(\eps,\del)$-indistinguishable on node-neighboring input graphs. The definition below follows the style of those in \cite{KasiviswanathanLNRS11,JosephMNR19,EdenLRS25} for the tabular and edge-privacy settings.

\begin{definition}[Interactive LNDP]
\label{def:int-lndp}
    Let $n,\ell\in\N$.
    A \emph{transcript} $\pi$ is a vector consisting of some initial public randomness $\pubrand$, and a 3-tuple for each round $t\in[\ell]$ of the form $(S^t_U, S^t_R, S^t_Y)$.
    Each element in the tuple encodes, respectively, the set $S^t_U\subseteq[n]$ of nodes chosen, the per-node algorithms\footnote{A ``per-node algorithm'' run by node $i$ is an algorithm whose output is a function only of the neighborhood of node $i$, public randomness, and the outputs from and choices of per-node algorithms run in previous rounds of the algorithm.} used by each of the chosen nodes (this description includes the per-node algorithm's parameters---e.g., its privacy parameters), and the (randomized) output produced.
    An \emph{algorithm} in this model is a function $\cA$
    that maps each possible transcript to public randomness, a set of nodes, and per-node algorithms for those nodes.

    Given $\eps > 0$ and $\del\in[0,1]$, a randomized algorithm $\cA$ satisfies \emph{$(\eps,\del)$-local node differential privacy (LNDP)} if the algorithm $\cT_{\cA}$ that outputs the entire transcript generated by $\cA$ has the property that, for all choices of initial public randomness $\pubrand$ and all pairs of node-neighboring graphs $G$ and $G'$ on node set $[n]$,
    \[
        \cT_{\cA}(G) \approx_{\eps,\del} \cT_{\cA}(G').
    \]
\end{definition}

This privacy definition assumes that all players follow the 
protocol
(what cryptographers dub the \textit{honest-but-curious} model). This assumption only strengthens the lower bounds we present.

\paragraph{Lifting Noninteractive Lower Bounds to the Interactive Setting}

The following lemma shows that our bounds on the TV distance between outputs from \LNDP algorithms also apply to interactive, $\ell$-round LNDP algorithms, up to a factor of $\ell$.

\begin{lemma}[TV bounds for interactive LNDP algorithms]
\label{lem:int-extend}
    Let $\cG$ be a distribution on $n$-node undirected graphs, and let $H$ be a fixed $n$-node undirected graph.
    Fix $\eps,\delta\geq 0$. Suppose there exists $\eta_{\eps,\delta,\cG}$ such that, for all 
    $(\eps,\del)$-\LNDP algorithms $\cA$ on $n$-node graphs, when $G \sim \cG$, %
    \begin{equation}
    \label{eq:tv-nonint-assume}
        \dtv\bracks{\bparen{G,\cA(G)}, \bparen{G,\cA(H)}} \leq \eta_{\eps,\delta,\cG}. %
    \end{equation}
    Then, for %
    all (interactive) 
    $\ell$-round $(\eps,\del)$-LNDP algorithms $\cB$ on $n$-node graphs, when $G\sim \cG$,
    \begin{equation}
    \label{eq:tv-int}
        \dtv\bracks{\bparen{G,\cB(G)}, \bparen{G,\cB(H)}} \leq \ell \cdot \eta_{\eps,\delta,\cG}. %
    \end{equation}
\end{lemma}

Before proving the lemma, we state its consequences for two problems: edge-counting (\Cref{cor:int-edge-ct-lb}) and Erdős-Rényi parameter estimation (\Cref{cor:int-er-lower-bd}). Both follow from combining the lemma above with appropriate statements for noninteractive algorithms---\Cref{lem:hell-empty-dreg,thm:er-empty-tv-bd}, respectively. We omit detailed proofs of the corollaries, since they are similar to those of \Cref{thm:edge-ct-lb-approx,thm:er-lower-bd}.

\begin{corollary}
[Error for interactive private edge counting]
\label{cor:int-edge-ct-lb}
    There exists a constant $c>0$ such that, for all $\ell\geq 1$, $0 \leq \delta \leq \eps \leq c$, $d \in \mathbb{Z}^+$ such that $d \leq \frac{c}{\eps\ell}$, sufficiently large $n\in\N$, and every (interactive) $\ell$-round $(\eps,\del)$-LNDP algorithm $\cB$, the following holds.
    If $\Pr_{\text{coins of $\cB$}} [\babs{\cB(G)- |E(G)|} \leq \alpha] \geq \frac{2}{3}$
    for every $n$-node $d$-bounded graph $G$, then
    \(
        \alpha
        \geq
        \frac{1}{2} \min\set{\frac{dn}{2},\binom n 2}.
    \)
    Furthermore, for $d = \floor{\min\set{\frac{c}{\eps\ell}, n-1}}$, we have $\alpha = \Omega\paren{\min\set{\frac{n}{\eps\ell},n^2}}$.
\end{corollary}

\begin{corollary}
[Error for private ER parameter estimation]
\label{cor:int-er-lower-bd}
    There exists a constant $c > 0$ such that, for all $\ell \in \N$,  sufficiently large $n\in\N$, $\eps\in(0,1]$, $\del \in \left[0, \frac{c \eps}{n \ell^3 \log (n\ell)} \right]$, and every (interactive) $\ell$-round $(\eps,\del)$-LNDP algorithm $\cB$, the following holds.
    If $\Pr_{\substack{G\sim G(n,p); \\ \text{coins of $\cB$}}}
        \Bigl[
            |\cB(G) - p| \leq \alpha
        \Bigr]
        \geq \frac{2}{3}$ for every $p \in [0, 1]$, then \(
        \alpha
        =
        \Omega\paren{\min\set{\frac{1}{n\eps\ell^2}, 1}}
    \).
\end{corollary}

\begin{proof}[Proof of \cref{lem:int-extend}]
    Every $\ell$-round LNDP algorithm has the following form, by \cref{def:int-lndp}: for each round $k\in[\ell]$, an algorithm $\wt \cA_k$ selects a subset of nodes; has each selected node release a function of the following: its neighborhood, public randomness, its and other nodes' outputs from previous rounds, and internal state held by that node in previous rounds; and releases the output. Apart from the persistent internal state for each node, we see that $\wt\cA_k$ is an \LNDP algorithm. 
    
    We claim that every $\ell$-round LNDP algorithm can be simulated by the composition of $\ell$ \LNDP algorithms. 
    To see this, we introduce the following notation.
    Let $y_k$ denote the part of the ``transcript'' produced in round $k$ (i.e., the nodes selected, per-node algorithms chosen by each node, outputs produced by each node, and public randomness).
    Let $\cA_k(y_1,\ldots, y_{k-1}, G)$ denote an \LNDP algorithm that takes as input the transcript $y_1,\ldots,y_{k-1}$ and graph $G$.
    
    An algorithm that runs $\cA_k$ for each round $k\in[\ell]$ nearly matches the form of the interactive algorithm $\cB$ that runs $\wt \cA_k$ for each round $k\in[\ell]$. One key difference is that each party's internal state may persist across rounds in $\cB$. To see that a composition of \LNDP algorithms can match this behavior, we can have each party $i\in[n]$ choose for round $k\in[\ell]$ an internal state, uniformly at random, that is consistent with the set of outputs it has released so far (and choices of randomizers, etc.). Party $i$ can then use this internal state for round $k$, and the resulting output distribution will be equal to the output distribution it would have if it had maintained its internal state. Thus, if we define $\cA_k$ as the algorithm where each node simulates persistent internal state in this manner, we see that $\cA_k$ and $\wt \cA_k$ have identical distributions over outputs. (The resulting \LNDP-based algorithm may incur a blowup in time complexity, but this is irrelevant to the argument here since our lower bounds apply to all $\ell$-round LNDP algorithms, even inefficient ones.)

    We also note that each \LNDP algorithm $\cA_k$ must be specifically $(\eps,\del)$-\LNDP, since otherwise the outputs from round $k$ would violate the guarantee that the overall algorithm's transcript is $(\eps,\del)$-indistinguishable.     
    Thus, every $\ell$-round $(\eps,\del)$-LNDP algorithm can be simulated by the composition of $\ell$ algorithms that are each $(\eps,\del)$-\LNDP. 

    We now prove \cref{eq:tv-int}, by induction over the number of rounds in the interactive LNDP algorithm.
    
    Let $Y_k$ denote the distribution over parts of the transcript produced in round $k$ (i.e., where $y_k$ denotes the part of the transcript produced in round $k$, we let $Y_k$ denote the corresponding distribution), and let $Y_{[k]}:=Y_1,\ldots, Y_k$; additionally, let $Y_k$ and $Y'_k$, respectively, correspond to the distribution over transcripts produced in round $k$ when the input graph is drawn, respectively, from $\G$ and $\bH$. Define $\phi_k: \bparen{G,y_1,\ldots, y_{k-1}}\mapsto \bparen{G,y_1,\ldots, y_{k-1},\cA(G,y_1, \ldots, y_{k-1})}$. 
    
    The base case follows immediately from our assumption in \cref{eq:tv-nonint-assume}.
    We now complete the proof. By the inductive hypothesis, where $G\sim \cG$, we have
    \begin{align}
        \cond{&}(k-1)\cdot \eta_{\eps,\delta,\cG}
        \cond{\notag \\} &\cond{\qquad}\geq
         \dtv\left[
            \Bparen{G, Y_{[k-2]}, \cA_{k-1}\bparen{Y_{[k-2]},G}}
            ,\right. \cond{\notag \\&\qquad  \qquad \quad \;\;\;\;\;}
            \left.\Bparen{G, Y'_{[k-2]}, \cA_{k-1}\bparen{Y'_{[k-2]},H}} \right] \notag \\
        &\cond{\qquad}  =
         \dtv\Bbracks{
            \Bparen{G, Y_{[k-1]}}
            ,
            \Bparen{G, Y'_{[k-1]}}} \notag \\
        &\cond{\qquad} \geq
        \dtv \left[
            \Bparen{G, Y_{[k-1]}, \cA_{k}\bparen{Y_{[k-1]},G}}
            ,\right. \cond{\notag \\& \qquad \qquad \quad \;\;\;\;\;}
            \underbrace{\Bparen{G, Y'_{[k-1]}, \cA_{k}\bparen{Y'_{[k-1]},G}}}_{(*)} \Big]
            \label{eq:tv-induct}
            ,
    \end{align}
    with the second line following by postprocessing with $\phi_k$ and the data processing inequality for TV distance.
    By \cref{eq:tv-nonint-assume},
    \begin{align*}
        \eta_{\eps,\delta,\cG}
        \geq 
        \dtv\Big[
            \cond{&}\underbrace{\Bparen{G, Y'_{[k-1]}, \cA_{k}\bparen{Y'_{[k-1]},G}}}_{(*)} ,
            \cond{\\ &\;\;}
            \underbrace{\Bparen{G, Y'_{[k-1]}, \cA_{k}\bparen{Y'_{[k-1]},H}}}_{(**)}
        \Big].
    \end{align*}
    Thus, by the triangle inequality for TV distance, where we substitute $(**)$ for $(*)$ in \cref{eq:tv-induct}, 
    \begin{align*}
                k\cdot \eta_{\eps,\delta,\cG} 
        \geq
        \dtv\Big[
            \cond{&}\Bparen{G, Y_{[k-1]}, \cA_{k}\bparen{Y_{[k-1]},G}}
            , \cond{\\ &\;\;}
            \Bparen{G, Y'_{[k-1]}, \cA_{k}\bparen{Y'_{[k-1]},H}} \Big]
            ,
    \end{align*}
    which completes the proof.
\end{proof}

\section{Advanced Grouposition for Pure \texorpdfstring{\LNDP}{Noninteractive LNDP}}
\label{sec:adv-grp}

In this section, we prove an analogue of ``advanced grouposition'' for pure \LNDP. In the standard LDP setting for tabular data, advanced grouposition \cite{BunNS19} states that group privacy guarantees for $k$ users degrade proportional to $\sqrt{k}$ rather than $k$. We show an analogous result for pure \LNDP in \Cref{thm:adv-grp}, which we prove in \Cref{sec:pf-adv-grp}. We also prove that advanced grouposition cannot apply to approximate \LNDP,
showing a separation between the pure- and approximate-\LNDP settings---in contrast, \cite{BunNS19} show that pure and approximate LDP are essentially equivalent.

Our proof of advanced grouposition requires insights specific to pure-\LNDP algorithms. 
In the standard local model, changing the data of $k$ individuals only affects the inputs to those $k$ randomizers, so only these randomizers' outputs contribute to the privacy loss.
The analysis for pure \LNDP is more delicate, as rewiring $k$ nodes may affect all nodes' edge lists, and thus the inputs to all randomizers. 
However, we show that the brittle structure of pure \LNDP means that, even though many nodes' edge lists can change, only a few of these nodes' randomizers contribute significantly to the privacy loss.

\begin{theorem}[Advanced grouposition for pure \LNDP]
\label{thm:adv-grp}
Let $\eps > 0$. If $\cA$ is an \ezLNDP algorithm and $G$ and $G'$ are at node distance $k$, then for all $\del \in (0, 1]$, we have 
$$
    \cA\paren{G}
        \approx_{(\eps', \del)}
    \cA(G')
$$
for some $\eps' = 
O\bparen{k\eps^2 + \eps \sqrt{k \log(1/\del)}}
$. Furthermore, there is a constant $c > 0$ such that for all $\delta \in (0, \frac{1}{20})$ and $k \leq \frac{c}{\eps^2 \ln(2/\delta)}$, we have $\dtv(\cA(G), \cA(G')) \leq \frac{1}{3}$. 
\end{theorem}

\Cref{thm:adv-grp} immediately implies a \emph{separation between pure and approximate \LNDP}. To see this, consider a $K$-clique $G$ with $|K| = \frac{n}{2} + \frac{c}{\eps^2}$ for $\eps < 1$ and some sufficiently small constant $c > 0$. \Cref{thm:adv-grp} implies that $G$ is indistinguishable from the $K'$-clique $G'$ where $|K'| = \frac{n}{2}$, since they are node distance $\Theta(\frac{1}{\eps^2})$ apart. On the other hand, our approximate-\LNDP algorithm in \Cref{thm:clique-new}, which estimates clique sizes with additive error $O_\delta(\frac{1}{\eps})$, can distinguish them.

\Cref{thm:adv-grp} also shows that pure-\LNDP algorithms require error $\Omega\bparen{\frac{n}{\eps^2}}$ for counting edges.
To see this, define $G''$ as a $t$-starpartite graph (\cref{def:starpartite-graph}) with $t = \Theta(\frac{1}{\eps^2})$. It is indistinguishable from the empty graph $\gempty$ under pure \LNDP since they are at node distance $\Theta(\frac{1}{\eps^2})$, and differ in edge count by $\Theta(nt) = \Theta(\frac{n}{\eps^2})$, implying the lower bound.%

In \Cref{sec:pure-priv-loss}, we prove general properties about the \textit{privacy losses} of randomizers in pure-\LNDP algorithms, and then prove \Cref{thm:adv-grp} in \Cref{sec:pf-adv-grp}.

\subsection{Privacy Loss of Pure \texorpdfstring{\LNDP}{Noninteractive LNDP}}
\label{sec:pure-priv-loss}

In this section, we state and prove \cref{lem:priv-loss-lndp}, which says that the sum of the absolute values of each randomizer's privacy loss in an \ezLNDP algorithm is bounded by $\Theta(k\eps)$. Intuitively, this means that for any two graphs at node distance $k$, the bulk of the privacy loss comes only from the $k$ ``rewired'' nodes.
We first define the privacy loss of a randomizer, and then prove several facts about the privacy loss of \LNDP algorithms, which we then use in our proof of \cref{thm:adv-grp}.

Throughout this section, to simplify notation we fix the public randomness $\pubrand$ provided to the randomizers and omit it from our definitions.

\begin{definition}[Privacy loss]
\label{def:priv-loss}
    Let $G$ and $G'$ be graphs on node set $[n]$. For a randomizer $\cR_i: \cX\to\cZ$ and %
    $z\in \cZ$, define the \emph{privacy loss} of $\cR_i$ between $G$ and $G'$ as
    \[
    \privloss_i^{G, G'}(z)
        =
        \ln\paren{
        \frac
        {\pr{\cR_i(\neigh{G}{i}) = z }}
        {\pr{ \cR_i(\neigh{G'}{i}) = z }}
        }.
    \]
\end{definition}

\begin{lemma}[Privacy loss of pure \LNDP]
\label{lem:priv-loss-lndp}
    Let $\cA$ be an \ezLNDP algorithm with randomizers $\cR_1,\ldots, \cR_n$.
    For all graphs $G$ and $G'$ at node distance $k \in \mathbb{N}$ and all $(z_1,\ldots,z_n)\in \cZ^n$,
    \[
        \sum_{i\in[n]} \abs{ \privloss_i^{G, G'}(z_i) } \leq 3k\eps.
    \]
\end{lemma}

To prove \cref{lem:priv-loss-lndp}, we use \Cref{clm:priv-loss-facts}. 
\begin{claim}
\label{clm:priv-loss-facts}
Let $\cA$ be an \edLNDP algorithm with randomizers $\cR_1, \ldots, \cR_n$. Then:
\begin{enumerate}[label=(\alph*)]
    \item For all $i \in [n]$ and all pairs $X_i, X_i'$ of edge lists for node $i$, we have $\cR_i(X_i) \approx_{\eps,\del} \cR_i(X'_i)$.
    \item If $\del = 0$, then $\abs{\privloss^{G, G'}_i(z_i)} \leq \eps$ for all $i \in [n]$, $(z_1,\ldots, z_n)\in \cZ^n$, and graphs $G$, $G'$ on node set $[n]$.
    \item If $\del = 0$, then $\abs{ \sum_{i\in[n]} \privloss_i^{G, G'}(z_i) } \leq \eps$ for all node neighbors $G,G'$ and $(z_1,\ldots, z_n)\in \cZ^n$.
\end{enumerate}
\end{claim}

\begin{proof}[Proof of \cref{clm:priv-loss-facts}]
We prove each item separately.

\textbf{Item (a).} 
Let $G_i$ and $G'_i$ be undirected graphs on node set $[n]$, where the edges incident to node $i$ are given by $X_i$ and $X'_i$, respectively. Both graphs contain no other edges.
The graphs $G_i$ and $G'_i$ are node neighbors, 
so by the definition of \LNDP the output distributions of $\cR_i(X_i)$ and $\cR_i(X'_i)$ must be $(\eps,\del)$-indistinguishable.

\textbf{Item (b).} This follows immediately from Item (a) and \cref{def:priv-loss}.

\textbf{Item (c).}
    By the definition of \LNDP and independence of the randomizers,
    \begin{align*}
    \cond{&}\abs{ \sum_{i\in[n]} \privloss_i^{G, G'}(z_i) }
    \cond{\\& \quad}=
    \abs{
            \ln\paren{
            \prod_{i\in[n]}
            \frac{\Pr[\cR_i(\neigh{G}{i}) = z_i]}
            {\Pr[\cR_i(\neigh{G'}{i}) = z_i]}
            }
        }
    \cond{\\& \quad}=
    \abs{
        \ln\paren{
        \frac{
            \Pr[\cR_1(\neigh{G}{1}) = z_1\wedge\cdots \wedge \cR_n(\neigh{G}{n}) = z_n]}
            {\Pr[\cR_1(\neigh{G'}{1}) = z_1\wedge\cdots \wedge \cR_n(\neigh{G'}{n}) = z_n]}
        }
    }
    \cond{\\& \quad} \leq \eps. \qedhere
    \end{align*}
\end{proof}

\begin{proof}[Proof of \cref{lem:priv-loss-lndp}]
    We first prove the statement for $k = 1$. Fix two node-neighboring graphs $G$ and $G'$ and $(z_1,\ldots,z_n)\in \cZ^n$. Assume w.l.o.g.\ that $G$ and $G'$ differ only on (some) edges incident to node 1. %
    Note that the induced subgraphs of $G$ and $G'$ on nodes $2,\ldots, n$ are identical. 
    
    Our argument proceeds as follows: since the privacy losses $L_i^{G, G'}(z_i)$ could be positive or negative (making them tricky to analyze), we construct two new node-neighboring graphs $H$ and $H'$ by permuting the neighborhoods of nodes $2, \ldots, n$ in $G$ and $G'$ such that $L_i^{H, H'}(z_i) = \left|L_i^{G, G'}(z_i) \right| \geq 0$ for all $i \geq 2$, and defining the neighborhoods of node $1$ to agree with these new neighborhoods. Since the induced subgraph on nodes $2, \ldots, n$ is unaffected, the resulting graphs $H$ and $H'$ are still node neighbors (they only differ in the neighborhood of node 1), allowing us to bound the original privacy losses between $G$ and $G'$ using $H$ and $H'$.

    We now define the neighborhoods of the graphs $H$ and $H'$ on node set $[n]$. For each $i = 2, \ldots, n$, define 
    \begin{gather*}
    \neigh{H}{i} = \begin{cases}
            \neigh{G}{i} & \text{if $\privloss_i^{G, G'}(z_i) \geq 0$;} \\
            \neigh{G'}{i} & \text{otherwise,}
        \end{cases}
    \condelse{\\ \vspace*{1em} \text{and} \vspace*{1em} \\ }{\quad\quad \text{and} \quad\quad}
        \neigh{H'}{i} = \begin{cases}
            \neigh{G}{i} & \text{if $\privloss_i^{G, G'}(z_i) < 0$;} \\
            \neigh{G'}{i} & \text{otherwise.}
        \end{cases}
    \end{gather*}
    Next, to ensure that $H$ and $H'$ define valid graphs, we set
    $$
    \neigh{H}{1} = \left\{i \in [n] : 1 \in \neigh{H}{i} \right\}, 
    \condelse{\quad}{\qquad} 
    \neigh{H'}{1} = \left\{i \in [n] : 1 \in \neigh{H'}{i} \right\}.
    $$
    Note that $H$ and $H'$ are node neighbors. Furthermore, since $\privloss_i^{G, G'}(z_i) = -\privloss_i^{G', G}(z_i)$,
    for all $i\geq 2$ we have $
    \privloss_i^{H, H'}(z_i)
    =
    \abs{\privloss_i^{G, G'}(z_i)} \geq 0$.
    We now bound
    \begin{align*}
        \cond{&} \sum_{i=1}^n \abs{\privloss^{G, G'}_i(z_i)} 
        \cond{\\}&= \abs{\privloss^{G, G'}_1(z_1)} + \sum_{i=2}^n \abs{\privloss^{G, G'}_i(z_i)} \\
        &= \abs{\privloss^{G, G'}_1(z_1)} + \abs{\sum_{i=2}^n \privloss^{H, H'}_i(z_i)} \tag{using $\privloss^{H, H'}_i(z_i)
    =
    \abs{  \privloss_i^{G, G'}(z_i)} \geq 0$}\\
        &= \abs{\privloss^{G, G'}_1(z_1)} + \abs{\left(\sum_{i=1}^n \privloss^{H, H'}_i(z_i)\right) - \privloss^{H, H'}_1(z_1)} \\
        &\leq \abs{\privloss^{G, G'}_1(z_1)} + \abs{\sum_{i=1}^n \privloss^{H, H'}_i(z_i)} + \abs{\privloss^{H, H'}_1(z_1)} \\
        &\leq \eps + \eps + \eps = 3 \eps,
    \end{align*}
    where in the last inequality we use that each term is bounded above by $\eps$ (by \Cref{clm:priv-loss-facts}). The statement for general $k \in \mathbb{N}$ follows by a straightforward induction and the triangle inequality.
\end{proof}

We take the maximum of each term in the sum to achieve the same inequality, formalized in \cref{cor:max-priv-loss-lndp-group}.

\begin{corollary}
\label{cor:max-priv-loss-lndp-group}
    Let $\cA$ be a \ezLNDP algorithm with randomizers $\cR_1,\ldots, \cR_n$.
    For all pairs of graphs $G$ and $G'$ at node distance $k$,
    \[
        \sum_{i\in[n]} \max_{z\in\cZ^n} \abs{ \privloss_i^{G, G'}(z_i) } \leq 3k\eps.
    \]
\end{corollary}

\begin{proof}[Proof of \cref{cor:max-priv-loss-lndp-group}]
    The statement of \Cref{lem:priv-loss-lndp} is equivalent to %
    $\max_{z\in \cZ^n} \sum_{i\in[n]} \abs{ \privloss_i^{G, G'}(z_i) } \leq 3k\eps$. Define $\widehat{z} = \arg\max_{z\in\cZ^n}\sum_{i\in[n]} \abs{ \privloss_i^{G, G'}(z_i) }$, and $z^*\in\cZ^n$ by $z^*_i = \arg\max_{z_i \in \cZ} \abs{ \privloss_i^{G, G'}(z_i) }$. It suffices to show
    \begin{equation}
    \label{eqn:priv-loss-sum-equal}
        \sum_{i\in[n]} \abs{ \privloss_i^{G, G'}(\widehat z_i)%
        }
        =
        \sum_{i\in[n]} \abs{ \privloss_i^{G, G'}(z^*_i)%
        }
        .
    \end{equation}
    By definition of $\widehat{z}$ and $z^*$, we have $\abs{ \privloss^{G, G'}_i(\widehat{z}_i)} \leq \abs{ \privloss^{G, G'}_i(z^*_i)}$
    for all $i\in[n]$. Assume for contradiction that the inequality is strict for some $j\in[n]$: $\abs{ \privloss^{G, G'}_j(\widehat{z}_j)} < \abs{ \privloss^{G, G'}_j(z^*_j)}$.
    Considering the vector $\widetilde{z}$ given by $\widetilde{z}_j = z^*_j$ and $\widetilde{z}_i = \widehat{z}_i$ for all $i \neq j$, we get
    \[
        \sum_{i\in[n]} \abs{ \privloss_i^{G, G'}(\hat{z_i}) }
        <
        \sum_{i\in[n]} \abs{ \privloss_i^{G, G'}(\widetilde{z}_i) },
    \]
    contradicting that $\widehat{z}$ is the maximizer. This implies \Cref{eqn:priv-loss-sum-equal}, completing the proof. 
\end{proof}

\subsection{Proof of Advanced Grouposition}
\label{sec:pf-adv-grp}

To prove \cref{thm:adv-grp}, we use the following standard facts connecting $(\eps,\del)$-indistinguishability to TV distance and KL divergence.\footnote{The \emph{Kullback--Leibler (KL) divergence} between distributions $P$ and $Q$ on domain $\cX$ is $\dkl\paren{P \| Q} := \int_{x\in \cX} P(x) \cdot \ln \paren{\frac{P(x)}{Q(x)}} \mathrm dx$.}

\begin{fact}[$\dtv$ and $(\eps,\del)$-DP]
\label{lem:tv-dp}
    Let $\eps > 0$ and $\delta \in [0, 1)$, and let $P$ and $Q$ be two probability distributions such that $P\simed Q$.
    Then
    \[
        \dtv\bigl( P , Q \bigr) \leq (e^\eps - 1) + \del.
    \]
    Moreover, if $\eps\leq 1$, then
    $
        \dtv\bigl( P , Q \bigr) \leq 2\eps + \del.
    $
\end{fact}

\begin{fact}[$\dkl$ and $(\eps,0)$-DP {\cite[Proposition 3.3]{BunS16}}]
\label{lem:kl-dp}
    Let $\eps > 0$, and let $P$ and $Q$ be two probability distributions such that $P\approx_{\eps,0} Q$.
    Then
    \[
        \dkl\bigl( P \,\| \, Q \bigr) \leq \frac{1}{2}\eps^2.
    \]
\end{fact}

\begin{proof}[Proof of \cref{thm:adv-grp}]
    Since graphs $G$ and $G'$ are at node distance $k$, w.l.o.g.\ they differ only on (some) edges incident to nodes  in $[k]$. 
    For each $i \in [n]$, define $\fakeeps_i = \max_{z\in\cZ^n} \abs{\privloss_i^{G, G'}(z_i)}$. 
    By \Cref{cor:max-priv-loss-lndp-group}, we have $\sum_{i=1}^n \fakeeps_i \leq 3k \eps$. Since the algorithm is \LNDP, we have
    $\cR_i(\neigh{G}{i}) \approx_{\fakeeps_i} \cR_i(\neigh{G'}{i})$ for all $i \in [n]$, giving
    \(
        \dkl\paren{ \cR_i(\neigh{G}{i}) \| \cR_i(\neigh{G'}{i}) } \leq \frac{1}{2} \fakeeps_i^2
    \) by \cref{lem:kl-dp}.
    The KL divergence between running randomizer $\cR_i$ on $\neigh{G}{i}$ and $\neigh{G'}{i}$ is exactly equal to the expected privacy loss, giving %
    \begin{align*}
        \cond{&}\underset{z_i\sim \cR_i(\neigh{G}{i})}{\Ex}\bracks{ \privloss_i^{G, G'}(z_i) } 
        \cond{\\ &\qquad } 
        = \sum_{z_i} \Pr[\cR_i(\neigh{G}{i}) = z_i] \cdot \ln \left(\frac{\Pr[\cR_i(\neigh{G}{i}) = z_i]}{\Pr[\cR_i(\neigh{G'}{i}) = z_i]}\right)
        \cond{\\&\qquad} 
        = \dkl\bigl( \cR_i(\neigh{G}{i}) \;\|\; \cR_i(\neigh{G'}{i}) \bigr) 
        \cond{\\&\qquad}
        \leq \frac{1}{2}\fakeeps_i^2.
    \end{align*}
    We now separately compute high-probability upper bounds on the sum of privacy losses for nodes $1, \ldots, k$ and nodes $k+1, \ldots, n$.
    Recall from \Cref{clm:priv-loss-facts}
    that $\fakeeps_i \leq \eps$.
    Applying Hoeffding's inequality gives
    \begin{align*}
        \cond{&}\Pr_{z\sim \vec{\cR}(G)}
        \bracks{
            \biggl|\sum_{i = 1}^k
            \privloss^{G, G'}_i(z_i)%
            \biggr|
            >
            \frac{k\eps^2}{2} + \eps \sqrt{2k\ln (2/\del)}
        }
        \cond{\\&\qquad\qquad}
        \leq \exp\paren{-\frac{2 (\eps \sqrt{2k\ln (2/\del)})^2}{\sum_{i=1}^k (2\eps)^2}}
        \cond{\\&\qquad\qquad}
        \leq \exp\paren{-\frac{\eps^2 k \ln(2/\del)}{k\eps^2}} 
        \cond{\\&\qquad\qquad}
        = \frac{\del}{2}.
    \end{align*}

    For nodes $k+1, \ldots, n$, we use $\sum_{i=1}^n \fakeeps_i \leq 3k \eps$ and $\fakeeps_i \leq \eps$ to get $\sum_{i=k+1}^n \fakeeps_i^2 \leq \sum_{i=1}^n \fakeeps_i^2 \leq 3k\eps^2$.
    Using this and applying Hoeffding's inequality, we get 
    \begin{align*}
        \cond{&}
        \Pr_{z\sim \vec{\cR}(G)}
        \bracks{
            \biggl|\sum_{i=k+1}^n
            \privloss^{G, G'}_i(z_i)
            \biggr|
            >
            3k\eps^2 + \eps \sqrt{6k\ln(2/\del)}
        }
        \cond{\\&\qquad}
        \leq \exp\paren{-\frac{2(\eps \sqrt{6k\ln(2/\del)})^2}{\sum_{i=k+1}^n (2\fakeeps_i)^2}}
        \cond{\\&\qquad}
        \leq \exp\paren{-\frac{12k \eps^2  \ln(2/\del)}{12k\eps^2}} 
        \cond{\\&\qquad}
        = \frac{\del}{2}.
    \end{align*}
    Combining the above bounds through a union bound gives
    \[
        \Pr_{z\sim \vec{\cR}(G)}
        \bracks{
            \biggl|\sum_{i=1}^n
            \privloss^{G, G'}_i(z_i)
            \biggr|
            >
            \frac{7k\eps^2}{2} + 2\eps \sqrt{6k\ln (2/\del)}
        }
        \leq \del.
    \]
    By the definition of \LNDP and \Cref{def:priv-loss}, this implies that 
    $\cA(G) \approx_{\eps',\del} \cA(G')$
    for $\eps' = 7k\eps^2 / 2 + 2\eps \sqrt{6k\ln(2/\del)}$.
    Because the bound from \cref{cor:max-priv-loss-lndp-group} holds for all fixed strings of public randomness $\pubrand$, it also holds for all distributions over public randomness.
    For $\del \in \paren{0, \frac{1}{20}}$ and $k \leq \frac{1}{1738 \eps^2 \ln(2/\del)}$ we have $\eps' < \frac{1}{8}$, so \cref{lem:tv-dp} implies that $\dtv(\cA(G), \cA(G')) \leq 2\eps' + \del < \frac{2}{8} + \frac{1}{20} < \frac{1}{3}$, completing the proof.  
\end{proof}

\section{Separating Degrees-Only and Unrestricted \texorpdfstring{\LNDP}{Noninteractive LNDP}}
\label{sec:dist-stars-regular}

In this section, we show that degrees-only \LNDP algorithms are strictly weaker than unrestricted ones: some problems can be solved only if randomizers get nodes' adjacency lists rather than just their degrees.
Recall from \cref{def:lndp} that an \LNDP algorithm $\cA$ is \emph{degrees-only} if each randomizer $\cR_i$ receives only the degree $d_i$ of node $i$ in a graph $G$, rather than its neighborhood $\neigh{G}{i}$. In this section, we call standard \LNDP algorithms (as in \Cref{def:lndp}) \emph{unrestricted}, since each randomizer $\cR_i$ may see the full neighborhood $\neigh{G}{i}$.

We describe a problem unsolvable by degrees-only \LNDP algorithms, but solvable if either the degrees-only restriction or the privacy requirement is relaxed. Thus, the hardness comes from combining these two restrictions. The task is to distinguish two distributions on undirected graphs: $\distgreg{t}$, the uniform distribution over $t$-regular graphs on nodes $[n]$, and $\distgstar{t}$, the uniform distribution over \emph{$t$-starpartite} graphs on nodes $[n]$ (\cref{def:starpartite-graph}), which have $t$ ``star center'' nodes of degree $n-1$ and $n-t$ nodes of degree $t$ (see \Cref{fig:3starpartite}).

\usetikzlibrary{shapes.geometric}
\begin{figure}[h!]
\centering

\begin{tikzpicture}[scale=1,
    every node/.style={circle, draw, inner sep=1.5pt, minimum size=6mm, font=\small}]

\node (h1) at (0,2) {\Large$\star$};
\node (h2) at (2,2) {\Large$\star$};
\node (h3) at (4,2) {\Large$\star$};

\node (l1) at (-0.5,0) {};
\node (l2) at (0.75,0) {};
\node (l3) at (2,0) {};
\node (l4) at (3.25,0) {};
\node (l5) at (4.5,0) {};

\draw (h1) to (h2);
\draw (h2) to (h3);
\draw (h1) to[bend left=20] (h3);

\foreach \hub in {h1,h2,h3} {
  \foreach \leaf in {l1,l2,l3,l4,l5} {
    \draw (\hub) -- (\leaf);
  }
}

\end{tikzpicture}

\caption{A $3$-starpartite graph on $n=8$ nodes, where the three starred nodes connect to all other nodes.}
\label{fig:3starpartite}
\end{figure}

For non-private algorithms, distinguishing $\distgreg{t}$ from $\distgstar{t}$ is easy for all $t\in[0,n-1]$, even in the degrees-only setting: if the input graph $G$ has a node of degree $n-1$, then $G$ is starpartite; otherwise, it is regular. (In contrast, some problems, such as distinguishing two different perfect matchings, remain hard in the degrees-only setting even without privacy constraints.)

For \textit{unrestricted} \edLNDP algorithms, we show the distributions are also distinguishable for some $t = \poly\bigl(\frac{\log(1/\del)}{\eps}\bigr)$. 
The structure of our distinguishing algorithm, which is inspired by locality-sensitive hashing, is described in \Cref{sec:unrestricted-star-vs-reg}. Its guarantees are summarized in the following theorem.

\begin{theorem}[Unrestricted \LNDP distinguisher for $\distgreg{t}$ and $\distgstar{t}$]
\label{thm:reg-star-dist-unrestr}
    Let $\eps \in (0, \frac{1}{2})$ and $\del \in (0, \frac{1}{10})$.
    There exists an \textbf{unrestricted} algorithm $\algstarreg$ that is \edLNDP and, moreover, $\algstarreg$ satisfies
    \begin{gather*}
        \Pr_{G\sim \distgstar{t}}[\algstarreg(G) = \text{``starpartite''}] \geq 2/3 %
        \condelse{\\  \text{and} \\ }{\quad \text{and} \quad}
        \Pr_{G\sim \distgreg{t}}[\algstarreg(G) = \text{``regular''}] \geq 2/3 %
    \end{gather*}
    for some $t = O \left(\frac{\log^5(1/\delta)}{\eps^6}\right)$ and $n \geq \frac{c \ln^{10}(1/\delta)}{\eps^{10}}$ for some constant $c>0$.
    
\end{theorem}

In contrast, \cref{thm:od-indist} (\Cref{sec:indist-star-reg}) shows that, for some constant $c>0$, no degrees-only \LNDP algorithm can distinguish $G_{reg}\sim\distgreg{t}$ from $G_{star}\sim\distgstar{t}$ for all $t\le \frac{c\sqrt{n}}\eps$. This is tight: the degrees-only edge-counting algorithm based on the Laplace mechanism (\Cref{sec:laplace-mech-algo}) has error $O(\frac{n\sqrt{n}}\eps)$, while a $t$-starpartite graph and a $t$-regular graph differ in edge count by at least $\frac{t(n-2)}2$. Hence, for some constant $C>0$, all sufficiently large $n\in\N$, and all $t\ge \frac{C\sqrt{n}}\eps$, this algorithm distinguishes these graphs.

\begin{theorem}[Hardness %
for degrees-only \LNDP]
\label{thm:od-indist}
    Let $n\in \N$, $\eps \in (0,1]$ and $\del \in \bigl(0, \frac{1}{100n}\bigr]$.
    Let $\cA$ be a \textbf{degrees-only} \edLNDP algorithm.
    There \az{is a} %
    constant $c > 0$ such that, for all sufficiently large $n$ and $t\leq \frac{c\sqrt{n}}{\eps}$,
    \begin{gather*}
    \begin{split}
        \Pr_{G\sim \distgstar{t}}[\cA(G) = \text{``starpartite''}] < 2/3
        \condelse{\\ \text{or} \\}{\quad \text{or} \quad}
        \Pr_{G\sim \distgreg{t}}[\cA(G) = \text{``regular''}] < 2/3.
    \end{split}
    \end{gather*}
\end{theorem}

We prove \cref{thm:reg-star-dist-unrestr,thm:od-indist} in \cref{sec:unrestricted-star-vs-reg,sec:indist-star-reg}, respectively.

\subsection{Distinguishing Starpartite and Regular Graphs in Unrestricted \texorpdfstring{\LNDP}{Noninteractive LNDP}}
\label{sec:unrestricted-star-vs-reg}

In this section, we present our algorithm for distinguishing the distribution $\distgstar{t}$ of uniformly random $t$-starpartite graphs from the distribution of $\distgreg{t}$ uniformly random $t$-regular graphs, for some $t = O\paren{\frac{\log^5(1/\del)}{\eps^6}}$.

\paragraph{Overview of \cref{alg:star-vs-reg}}
\Cref{alg:star-vs-reg} distinguishes starpartite graphs from regular graphs as follows. Given an input graph $G$, we use public randomness to sample $s$ multisets $S_1,\ldots,S_s$, each containing $\frac{n}{t}$ elements drawn independently and uniformly with replacement from $[n]$, where $s$ is a parameter set later.

For all $j\in[s]$, each node $i$ reports a noisy bit $a_{i,j}$ indicating whether it has a neighbor in $S_j$. The server then computes the noisy averages $\overline{a_j}=\frac{1}{n}\sum_{i=1}^n a_{i,j}$.
Finally, the server computes the fraction of indices $j$ for which $\overline{a_j}\in[0,1]$, namely $\frac{1}{s}\sum_{j=1}^s \indic[\overline{a_j}\in[0,1]]$, and compares it to a threshold $\tau$. As shown later in \Cref{sec:acc-starp-dist}, the probabilities $\Pr[\overline{a_j}\in[0,1]\mid \text{$G$ is regular}]$ and $\Pr[\overline{a_j}\in[0,1]\mid \text{$G$ is starpartite}]$ differ by a noticeable gap, so choosing $\tau$ between them lets us distinguish the two distributions with high probability.

\begin{algorithm}[h!]
         \caption{$\algstarreg$ for distinguishing starpartite %
         and regular graphs}
    \label{alg:star-vs-reg}
    
    \begin{algorithmic}[1]
        \label{alg:star-vs-reg-algorithmic}
        \Statex \textbf{Parameters:} $\eps \in (0, \frac{1}{2})$, $\delta\in (0,\frac{1}{10})$,
        $t \in \N$.
        \Statex \textbf{Input:} Undirected graph $G$ on vertex set $[n]$.
        \Statex \textbf{Output:} ``starpartite'' or ``regular''.
        \Statex \textbf{Set Hyperparameters:} 
        \cond{\Statex \;}
        \(
        \left\{
        \begin{array}{rcl}
            \cedp &\gets& \frac{\sqrt{2 \ln (2.5/\delta)}}{\eps} \text{ and }
            s \gets 3t \ln(2/\delta) \\ %
            \sigma_{priv} &\gets& c_{\eps,\delta} \sqrt{s + \left(\frac{s}{t} + \sqrt{\frac{3s}{t} \ln(\frac{2}{\delta})} \right)n}\\
            \tau &\gets& \frac{1}{2}(\preg+\pstar), 
            \condelse{\\ &&}{\;\;}\text{where $\preg$ and $\pstar$ are defined in \Cref{eqn:preg-pstar}.}
        \end{array}
        \right . 
        \)

        \smallskip
        \State Publish $S_1, \ldots, S_s$ as multisets of $\frac{n}{t}$ elements of $[n]$ chosen uniformly at random \emph{with replacement}.
        \cw{\Statex\Comment{\commentstyle{\cw{These multisets act as the algorithm's public randomness.}}}}
        \label{step:starreg-sample-sets}
        \For{\textbf{all} nodes $i \in [n]$}
            \For{\textbf{all} $j \in [s]$}
                \State $b_{i,j} \gets \indic[N_i^G \cap S_j \neq \varnothing]$
                and  $a_{i,j} \gets b_{i,j} + Z_{i,j}$, \text{where $Z_{i, j} \sim \cN(0, \spriv^2)$}.
            \EndFor
            \State Node $i$ sends $(a_{i, 1}, \ldots, a_{i, s})$ to the central server.
        \EndFor
        \For{\textbf{all} $j \in [s]$}
            \State $\overline{a_j} \gets \frac{1}{n} \sum_{i =1}^n a_{i,j}$, and $Y_j \gets \indic[\overline{a_j} \in [0, 1]]$.
        \EndFor
        \State If $\frac{1}{s} \sum_{j=1}^s Y_j \geq \tau$, \textbf{output} ``regular''. Otherwise, \textbf{output} ``starpartite''.
    \end{algorithmic}
\end{algorithm}

In \cref{sec:priv-starp-reg,sec:acc-starp-dist}, we separately analyze the privacy and accuracy of \Cref{alg:star-vs-reg}.

\subsubsection{Privacy of \Cref{alg:star-vs-reg}}
\label{sec:priv-starp-reg}

\begin{lemma}
\label{lem:star-vs-reg-dp}
Let $\eps \in (0, \frac{1}{2})$ and $\del \in (0,\frac{1}{10})$.
Algorithm $\algstarreg$ (\Cref{alg:star-vs-reg}) is  (unrestricted) \edLNDP.
\end{lemma}

\begin{proof}
To show that $\algstarreg$ is \edLNDP, %
we consider simply releasing the entire matrix $[a_{i,j}]$, as the subsequent calculations are a postprocessing of this matrix.

Let $G$ be a graph on node set $[n]$, and let $G'$ be obtained by 
rewiring  %
node $i^*$. %
Let random variable $X$ denote the number of times $i^*$ occurs in all multisets $S_1, \ldots, S_s$. Let $\cE$ be the (good) event that $X \leq \frac{s}{t} + \sqrt{\frac{3s}{t} \ln(\frac{2}{\delta})}$.%

We first show that $\cA(G)$ and $\cA(G')$ are $(\eps, \frac{\delta}{2})$-indistinguishable when the event $\cE$ occurs. By the privacy of the Gaussian mechanism (\Cref{lem:gaussianmech}), it suffices to show that the $\ell_2$-sensitivity $\Delta_2$ of the matrix $[b_{i,j}]$ between $G$ and $G'$ is at most $\sqrt{s + \left(\frac{s}{t} + \sqrt{\frac{3s}{t} \ln(\frac{2}{\delta})} \right)n}$.
Changing the adjacency list of node $i^*$ affects:
\begin{itemize}
    \item The row $(a_{i^*, 1}, \ldots, a_{i^*, s})$ consisting of $s$ entries, and
    \item At most $X$ columns $(a_{k, 1}, \ldots, a_{k, n})$ for all $k$ such that $i^* \in S_k$, consisting of $n$ entries each.
\end{itemize}
This results in at most $s+nX$ entries of $[b_{i,j}]$ changing when the connections of $i^*$ are changed. So, conditioning on $\cE$, the $\ell_2$-sensitivity of $[b_{i,j}]$ is
$$
\Delta_2 \leq \sqrt{s + nX} \leq \sqrt{s + \left(\frac{s}{t} + \sqrt{\frac{3s}{t} \ln \left(\frac{2}{\delta} \right)} \right)n},
$$
showing that $\cA(G)$ and $\cA(G')$ are $(\eps, \frac{\delta}{2})$-indistinguishable when conditioned on $\cE$.

Next, note that since $|S_j| = \frac{n}{t}$ and there are $s$ multisets $S_1, \ldots, S_s$ chosen uniformly at random with replacement, we have $X \sim \Bin(\frac{sn}{t}, \frac{1}{n})$.
By \Cref{lem:bin-tails}, we have $\Pr[\bar{\cE}] \leq \frac{\delta}{2}$ (using the second term in the ``max'' since $\frac{s}{t} = 3 \ln(2/\del)$).
A standard conditioning argument then gives the unconditional privacy bound: For any (measurable) set $T$, we have
\begin{align*}
    \Pr[\cA(G) \cond{\!}\in \cond{\!} T] 
    &= \Pr[\cA(G) \cond{\!}\in \cond{\!} T \cond{\!} \;|\; \cond{\!} \cE] \Pr[\cE] \cond{\!}+ \cond{\!}\Pr[\cA(G) \cond{\!} \in\cond{\!}  T\cond{\!}  \;|\;\cond{\!} \overline{\cE}] \Pr[\overline{\cE}] \\
    &\leq \left(e^{\eps} \Pr[\cA(G') \in T \;|\; \cE] + \frac{\delta}{2}\right) \cdot \Pr[\cE] + \Pr[\overline{\cE}] \\
    &\leq e^\eps \Pr[\cA(G') \in T] + \frac{\delta}{2} + \frac{\delta}{2},
\end{align*}
giving that $\cA(G)$ and $\cA(G')$ are $(\eps, \delta)$-indistinguishable, completing the proof.
\end{proof}

\subsubsection{Accuracy of \Cref{alg:star-vs-reg}}
\label{sec:acc-starp-dist}

We now analyze the accuracy of \Cref{alg:star-vs-reg}. 

\paragraph{Intuition for the analysis}

The essence of our accuracy analysis is a bound of $\Omega(1/\savg^3)$ on the gap
between the means $\underset{S_j \subseteq [n]}{\E}[Y_j \;|\; \text{$G$ is regular}]$ and $\underset{S_j \subseteq [n]}{\E}[Y_j \;|\; \text{$G$ is starpartite}]$. Once this gap is established, we argue that the algorithm is able to determine whether $G$ is starpartite or regular with high probability by computing the sample mean $\overline{Y} = \frac{1}{s} \sum_{j=1}^s Y_j$ with $s = O(\savg^6)$ samples and seeing whether it lies closer to the true mean for regular graphs, or to the true mean for starpartite graphs. 

In this subsection, 
\Cref{lem:prob-star-in-sj,lem:bern-pstar-preg} give values
$\preg$ and $\pstar$ such that $\underset{S_j \subseteq [n]}{\E}[Y_j \;|\; \text{$G$ is regular}] \geq \preg$ and $\underset{S_j \subseteq [n]}{\E}[Y_j \;|\; \text{$G$ is starpartite}] = \pstar$, and \Cref{lem:integral-diff-easy} shows that the gap $\preg - \pstar$ is $\Theta(1/\savg^3)$.

To analyze the indicators $Y_j = \indic[\overline{a_j} \in [0, 1]]$, we begin by understanding the distributions of both the non-noisy and noisy averages $\overline{b_j}$ and $\overline{a_j}$ in the
starpartite and regular cases. First, consider the non-noisy averages $\overline{b_j} = \frac{1}{n} \sum_{i=1}^n b_{i,j}$. 
For the case when $G$ is $t$-starpartite, each $\overline{b_j}$ is bimodal: it is either $\frac{t}{n} \approx 0$ or $1$, depending on whether a star node with degree $n-1$ is contained in $S_j$. 
In contrast, when $G$ is $t$-regular, the distribution of the $\overline{b_j}$'s behaves almost like a binomial distribution that is concentrated around its mean $p_{n,t} \approx 1-e^{-1}$.

To obtain the noisy averages $\overline{a_j}$, we add Gaussian noise $\frac{1}{n} \sum_{i=1}^n Z_{i,j} \sim \cN(0, \savg^2)$ to $\overline{b_j}$, where $\savg^2 = \spriv^2/n$. For the $t$-starpartite case, the noisy averages $\overline{a_j}$ are distributed as a mixture of two Gaussians centered at $\frac{t}{n}$ and $1$ with variance $\savg^2$, whereas for the $t$-regular case, they are distributed closely to a single Gaussian centered at $p_{n,t} \approx 1-e^{-1}$ also with variance $\savg^2$. See \cref{fig:reg-star-compare} for a (stylized) visualization of the distributions.

\begin{figure*}[h!]
\centering
  \def\epsi{3.25} %
  \def\epsii{4.5} %
  \pgfmathsetmacro{\norm}{\epsi/sqrt(2*pi)}      %
  \pgfmathsetmacro{\normii}{\epsii/sqrt(2*pi)}      %
  \pgfmathsetmacro{\excoef}{(\epsi*\epsi)/2}      %
  \pgfmathsetmacro{\excoefii}{(\epsii*\epsii)/2}      %
  \pgfmathsetmacro{\muA}{1 - exp(-1)}           %

\begin{minipage}{0.42\textwidth}
\centering
\begin{tikzpicture}[xscale=2.5, yscale=1.75] %
    \draw[->] (-0.75,0) -- (1.75,0) node[right] {$x$};
    \draw[->] (-0.75,0) -- (-0.75,1.85) node[above] {$ $};

    \draw (0,0.02) -- (0,-0.02) node[below] {$0$};

    \draw (0.15,0.02) -- (0.15,-0.02) node[below] {$\frac{t}{n}$};

    \draw (\muA,0.02) -- (\muA,-0.02) node[below] {\small $p_{n,t}$};

    \draw (1,0.02) -- (1,-0.02) node[below] {\small $1$};

    \draw[blue, line width=2pt] (0.15,0) -- (0.15,0.666);
    
    \draw[blue, line width=2pt] (1,0) -- (1,1.134);

    \draw[red, dashed, very thick, domain=-0.7:1.7, samples=320, smooth, variable=\x]
      plot ({\x}, {(\normii)*exp(-\excoefii*(\x-\muA)^2)});

    \node[fill=white, inner sep=2pt] at (1.25,1.5) {\small \color{red} $\overline{b_j}$ for regular%
    };
    \node[fill=white, inner sep=2pt] at (1.6,0.8) {\small \color{blue} $\overline{b_j}$ for starpartite};

  \end{tikzpicture}
\end{minipage}
\hfill
\raisebox{2.5ex}{\large $\xrightarrow{\text{add $\cN(0, \savg^2)$}}$} %
\hfill
\begin{minipage}{0.42\textwidth}
\centering
\begin{tikzpicture}[xscale=2.5, yscale=1.75] %
    \draw[->] (-0.75,0) -- (1.75,0) node[right] {$x$};
    \draw[->] (-0.75,0) -- (-0.75,1.85) node[above] {$ $};

    \draw (0,0.02) -- (0,-0.02) node[below] {$0$};

    \draw (0.15,0.02) -- (0.15,-0.02) node[below] {$\frac{t}{n}$};

    \draw (\muA,0.02) -- (\muA,-0.02) node[below] {\small $p_{n,t}$};

    \draw (1,0.02) -- (1,-0.02) node[below] {\small $1$};

    \draw[blue, very thick, domain=-0.7:1.7, samples=320, smooth, variable=\x]
      plot ({\x}, {exp(-1)*(\norm)*exp(-\excoef*(\x-0.15)^2) + (1-exp(-1))*(\norm)*exp(-\excoef*(\x-1)^2)});

    \draw[red, dashed, very thick, domain=-0.7:1.7, samples=320, smooth, variable=\x]
      plot ({\x}, {(\norm)*exp(-\excoef*(\x-\muA)^2)});

    \node[fill=white, inner sep=2pt] at (1.2,1.3) {\small \color{red} $\overline{a_j}$ for regular};
    \node[fill=white, inner sep=2pt] at (1.5,1) {\small \color{blue} $\overline{a_j}$ for starpartite};

  \end{tikzpicture}
\end{minipage}
\caption{Depiction of the distributions of the non-private averages $\overline{b_j}$ (left) and private averages $\overline{a_j}$ (right) for starpartite graphs (blue) and regular graphs (red). 
The probability of the event $\overline{a_j} \in [0, 1]$ differs by $\Omega\bparen{\frac{1}{\savg^3}}$ between regular and starpartite graphs, so $s = O(\savg^6)$ samples suffice to distinguish these distributions.}
\label{fig:reg-star-compare}
\end{figure*}

In the arguments that follow, we show that the distribution of noisy averages $\overline{a_j}$ is more concentrated in the interval $[0, 1]$ for regular graphs than for starpartite graphs.
We let the threshold $\tau$ be the midpoint $\frac{1}{2} \left(\preg + \pstar\right)$ and then use a Chebyshev bound to show that $\sum_{j=1}^s Y_j$ lies on the correct side of $\tau$, depending on whether $G$ is regular or starpartite, with high probability, which proves \Cref{thm:reg-star-dist-unrestr}.

We now state the lemmas used in the accuracy proof. \cref{lem:prob-star-in-sj}, used for the analysis of the distributions of the indicators $Y_j$, shows that the probability that a star node is in the multiset $S_j$ when $G$ is starpartite is equal to the probability that a particular node $i$ has a neighbor in $S_j$ when $G$ is regular.

\begin{lemma}
\label{lem:prob-star-in-sj}
    Let $G_{star}$ and $G_{reg}$ be $t$-starpartite and $t$-regular graphs on node set $[n]$, respectively, and fix $i' \in [n]$ to be any node in $G_{reg}$. Suppose $S$ is a random multiset of $\frac n t$ nodes chosen uniformly with replacement (as is each of the $S_j$'s in \Cref{alg:star-vs-reg}). Then the probability that $S$ contains
    \cw{a star center}
    when the input is  $G_{star}$ equals the probability that $S$ contains a neighbor of $i$ when the input is $G_{reg}$. More precisely,
    \begin{flalign*}
        \Pr_S\Bbracks{S \text{ has a star center in }G_{star}}
        &= 
        \Pr_S\bracks{\neigh{i'}{G_{reg}} \cap S \neq \varnothing}
        \cond{\\&}=
        p_{n, t}, \cond{\\}
        \quad \text{where}\quad  p_{n,t} \cond{&}:=  1- \bparen{1 - \tfrac{t}{n}}^{n/t}\, .
    \end{flalign*}
\end{lemma}
\begin{proof}
Let $G_{star}$ be a $t$-starpartite graph and %
$S$ be a random multiset of $\frac n t$ nodes chosen uniformly with replacement. Since $G_{star}$ has $t$ star centers whereas
$S$ contains $\frac n t$ nodes (possibly with repetition), the probability that no star center is in $S$ is $\left(1-\frac{t}{n}\right)^{n/t}$. Thus, $\Pr[S \text{ contains a star center in }G_{star}] = 1- \left(1 - \frac{t}{n}\right)^{n/t}$.

Next, let $G_{reg}$ be a $t$-regular graph and $i' \in [n]$ be a node in $G_{reg}$. Since $i'$ has $t$ neighbors, the same argument suffices --- 
the probability that no neighbor of $i'$ is in $S$ is $\left(1-\frac{t}{n}\right)^{n/t}$, giving the desired equality.
\end{proof}

Let $p_{n, t}$ be as in \Cref{lem:prob-star-in-sj}, $\savg = \sigma_{priv}/\sqrt{n}$, $\gamma = \frac{1}{200 \savg^2}$, and $r = \sqrt{\frac{3p}{n} \ln(1/\gamma)}$, where $\spriv$ is defined in \Cref{alg:star-vs-reg}. Define
\begin{equation} \label{eqn:preg-pstar}
\begin{split}
    \preg &:= (1-\gamma)\Pr_{Z \sim \cN(p_{n,t}, \savg^2)}[Z \in [r, 1-r]], \\ %
    \pstar &:= (1-p_{n,t}) \Pr_{Z \sim \cN(t/n, \savg^2)}[Z \in [0, 1]] 
    \cond{\\&\qquad}
    + p_{n,t} \Pr_{Z \sim \cN(1, \savg^2)}[Z \in [0, 1]].
\end{split}
\end{equation}
In the next lemma, we analyze the distributions of $\overline{a_j}$ to show that $\Pr[Y_i = 1 \;|\; \text{$G$ is $t$-starpartite}] = \pstar$, and that $\Pr[Y_i = 1 \;|\; \text{$G$ is $t$-regular}] \geq \preg$ (which is a weaker yet sufficient condition).

\begin{lemma}
\label{lem:bern-pstar-preg}
Let $\pstar$ and $\preg$ be defined as in $\Cref{eqn:preg-pstar}$, and let $Y_j$ be as in \Cref{alg:star-vs-reg}. For every $j \in [s]$, if $G \sim \distgstar{t}$ is $t$-starpartite, then 
\begin{conddisplay}
    \underset{\substack{S_j \subseteq [n] \\ G \sim \distgstar{t}}}{\Pr}[Y_j = 1] = \pstar,
\end{conddisplay}
and if $G \sim \distgreg{t}$ is $t$-regular, then 
\begin{conddisplay}
    \underset{\substack{S_j \subseteq [n] \\ G \sim \distgreg{t}}}{\Pr}[Y_j = 1] \geq \preg.
\end{conddisplay}
\end{lemma}

\begin{proof}%

We separately analyze the distributions of the non-private averages $\overline{b_j}$ and the private averages $\overline{a_j}$ in the regular and starpartite cases.

\textit{Case of $G \sim \distgstar{t}$:} 
First, assume that $G \sim \distgstar{t}$ is $t$-starpartite. Fix $j \in [s]$, and consider the multiset $S_j$ as constructed in \Cref{alg:star-vs-reg}. 
If a star node is in $S_j$, then $b_{i,j} = 1$ for all $i \in [n]$, giving $\overline{b_j} = 1$. If there is no star node in $S_j$, then the only nodes $i$ satisfying $b_{i,j} = 1$ are the $t$ star nodes outside of $S_j$, giving $\overline{b_j} = \frac{t}{n}$. The probability that at least one  star is in $S_j$ is precisely $p_{n,t}$ as in \Cref{lem:prob-star-in-sj}, giving that $\Pr[\overline{b_j} = t/n] = 1-p_{n,t}$ and $\Pr[\overline{b_j} = 1] = p_{n,t}$. 

We obtain the private averages to be $\overline{a_j} = \overline{b_j} + Z_j$, where $Z_j = \frac{1}{n} \sum_{i=1}^n Z_{i,j} \sim \cN(0, \spriv^2/n)$. Combining this with the bimodal distribution of $\overline{b_j}$ described above, we get that the distribution of each $\overline{a_j}$ is a mixture of two Gaussians centered at $\frac{t}{n}$ and $1$ respectively, namely
$$
\overline{a_j} \sim (1-p_{n, t}) \cdot \cN(t/n, \savg^2) + p_{n,t} \cdot \cN(1, \savg^2),
$$
where $\savg = \spriv/\sqrt{n}$. This gives us that 
\begin{conddisplay}
    Y_j = \indic[\overline{a_j} \in [0, 1]] \sim \Bern(\pstar),
\end{conddisplay}
where the probability $\pstar$ is
\begin{align*}
    \pstar = (1-p_{n,t}) \Pr_{Z \sim \cN(t/n, \savg^2)}[Z \in [0, 1]]
    \cond{&\\}
    + p_{n,t} \Pr_{Z \sim \cN(1, \savg^2)}[Z \in [0, 1]].\cond{&}
\end{align*}
\textit{Case of $G \sim \distgreg{t}$:} 
For the case when $G$ is $t$-regular, we have that $\Pr[b_{i, j} = 1] = p_{n, t}$ by \Cref{lem:prob-star-in-sj}, as $b_{i, j} = 1$ if and only if $\neigh{G}{i} \cap S_j \neq \varnothing$. This implies that $\E[\overline{b_j} \;|\; G \sim \distgreg{t}] = p_{n,t}$.

Define $r = \sqrt{\frac{3p_{n,t}}{n} \ln \left(\frac{1}{\gamma} \right)}$ and $\gamma = \frac{1}{200 \savg^2}$.
For the noisy average $\overline{a_j} = \overline{b_j} + \cN(0, \spriv^2/n)$ to land within $[0, 1]$, we condition on the event that $\overline{b_j}$ does not deviate more than $r$ away from its mean $p_{n, t}$, and then find the probability that the Gaussian noise $\cN(0, \spriv^2/n)$ lands in the interval $[-p_{n, t}, 1-p_{n, t}]$, which implies that $\overline{a_j}$ itself lands in $[0, 1]$.

The $b_{i, j}$'s are not independent for a fixed $j \in [s]$: they are, by \Cref{lem:neg-correl}, \emph{negatively correlated} --- conditioning on a node $i$ being connected to $S_j$ reduces the probability that another node $i'$ is connected to $S_j$ (since each node in a $t$-regular graph $G$ has fixed degree $t$).
Thus, the $\overline{b_j}$'s are distributed more tightly than $\Bin(n, p_{n, t})$, and we can
apply a Chernoff bound (see  \cite[Theorem 1.16]{Doerr11}) 
to $\overline{b_j} = \frac{1}{n} \sum_{i=1}^n b_{i,j}$ to get
$$
\underset{\substack{S_j \subseteq [n] \\ G \sim \distgreg{t}}}{\Pr} \left[|\overline{b_j} - p_{n,t}| \geq r \right] \leq \gamma,
$$
where $r$ and $\gamma$ are defined as above. Then
\begin{align*}
    \cond{&} \underset{\substack{S_j \subseteq [n] \\ G \sim \distgreg{t}}}{\Pr}[Y_j = 1] 
    \cond{\\}
    &\cond{\qquad} 
    = \underset{\substack{S_j \subseteq [n], G \sim \distgreg{t} \\ Z \sim \cN(0, \savg^2)}}{\Pr}\left[ Z + \overline{b_j} \in [0, 1]\right] \\
    &\cond{\qquad} \geq \underset{\substack{S_j \subseteq [n], G \sim \distgreg{t} \\ Z \sim \cN(0, \savg^2)}}{\Pr}\left[ Z + \overline{b_j} \in [0, 1] \;\Big|\; |\overline{b_j} - p_{n,t}| \leq r \right] 
    \cond{\\ &\qquad \qquad \quad }
    \cdot \underset{\substack{S_j \subseteq [n] \\ G \sim \distgreg{t}}}{\Pr}[|\overline{b_j} - p_{n,t}| \leq r ] \\
    &\cond{\qquad}\geq (1-\gamma) 
    \underset{Z \sim \cN(p_{n,t}, \savg^2)}{\Pr}\left[ Z \in [r, 1-r]\right] \\
    &\cond{\qquad}= \preg.\qedhere
\end{align*}
\end{proof}

The next lemma states that $\preg$ and $\pstar$, as defined in \Cref{eqn:preg-pstar}, differ by a gap of size $\Theta \left(\frac{1}{\savg^3} \right)$. The proof is highly technical and is deferred to \Cref{lem:nst-new-satisfy,lem:integral-diff} in \Cref{sec:star-vs-reg-deferred-proofs}.

\newcommand{\yuck}{c_0}

\begin{lemma}
\label{lem:integral-diff-easy} 
Let $\eps \in (0, \frac{1}{2})$ and $\delta \in (0, \frac{1}{10})$. Define $\cedp = \frac{\sqrt{2 \ln (2.5/\delta)}}{\eps}$, and set $n, s, t \in \mathbb{N}$ to be
\begin{gather*}
    n \geq \frac{3}{4} \yuck \ln^5(2/\delta) \cedp^{10}, 
    \condelse{\\}{\qquad} 
    t = 30 \yuck \ln^2(2/\delta) \cedp^6, 
    \condelse{\\}{\qquad} 
    s = 3 t \ln(2/\delta), 
\end{gather*}
where $\yuck = 2^7 \,3^2 \,5^4 \,\pi$. %
Let $\savg = \spriv/\sqrt{n}$ where $\spriv$ is as in \Cref{alg:star-vs-reg}. Define $\preg$ and $\pstar$ as in \Cref{eqn:preg-pstar}.
Then
$$\preg - \pstar \geq \frac{1}{100 \sqrt{2 \pi} \cdot \savg^3} = \Theta \paren{\frac{1}{\savg^3}}.$$
\end{lemma}

We now prove the accuracy of \Cref{alg:star-vs-reg}. The proof leverages the gap given by \Cref{lem:integral-diff-easy} to show that, 
with probability at least $\frac{2}{3}$, that $\sum_{i=1}^s Y_i \geq \tau$ when $G\sim\distgreg{t}$, and that $\sum_{i =1}^s Y_i < \tau$ when $G\sim\distgstar{t}$.
This implies that the algorithm gives the correct answer with high probability. Note that $\tau$ is chosen to be the midpoint $\frac{1}{2} (\preg+\pstar)$ between $\preg$ and $\pstar$.

\begin{proof}[Proof of \Cref{thm:reg-star-dist-unrestr}]
\Cref{alg:star-vs-reg} is \edLNDP by \Cref{lem:star-vs-reg-dp}. 
To analyze accuracy, set $n$, $t$ and $s$ as in the statement of \Cref{lem:integral-diff-easy}, and note that the setting of $s$ agrees with that in \Cref{alg:star-vs-reg}.

First, assume that $G\sim\distgreg{t}$ is regular. By \Cref{lem:bern-pstar-preg}, each random variable $Y_j = \indic[\overline{a_j} \in [0, 1]]$ is distributed as $Y_j \sim \Bern(\mu)$, where $\preg \leq \mu \leq 1$. Define $g = \mu-\pstar$, so $g \geq \preg - \pstar \geq \frac{1}{100 \sqrt{2 \pi} \savg^3}$ by \Cref{lem:integral-diff-easy}. Using $\tau = \frac{1}{2} (\pstar + \preg) \leq \mu - \frac{g}{2}$ and $s \geq \frac{\yuck}{3} \savg^6$ (see condition (a) of \Cref{lem:nst-new-satisfy}), a Chebyshev bound gives that the probability $\cA(G)$ incorrectly outputs ``starpartite'' is
\begin{align*}
    \cond{&}\Pr \left[\frac{1}{s} \sum_{i=1}^s Y_i < \tau \right] 
    \cond{\\&\qquad}\leq \Pr \left[\left|\frac{1}{s} \sum_{i=1}^s Y_i - \mu \right| \geq \frac{g}{2}\right] 
    \cond{\\&\qquad}\leq \frac{s\mu(1-\mu)}{(sg/2)^2} \leq \frac{4}{sg^2} \leq \frac{8 \cdot 10^4 \cdot \pi \cdot \savg^6}{24 \cdot 10^4 \cdot \pi \cdot \savg^6} \leq \frac{1}{3}. 
\end{align*}
Now, suppose that $G \sim \distgstar{t}$ is starpartite. Then $Y_j \sim \Bern(\pstar)$ by \Cref{lem:bern-pstar-preg}. Using $t = \frac{1}{2} (\pstar + \preg) \geq \pstar + \frac{g}{2}$, a similar Chebyshev bound gives that $\cA(G)$ incorrectly outputs ``regular'' with probability
\begin{align*}
    \cond{&}\Pr \left[\frac{1}{s} \sum_{i=1}^s Y_i \geq \tau \right] 
    \cond{\\&\qquad \qquad}
    = \Pr \left[\left|\frac{1}{s} \sum_{i=1}^s Y_i - \pstar \right| \geq \frac{g}{2}\right] 
    \cond{\\&\qquad \qquad}
    \leq \frac{s \pstar(1-\pstar)}{(sg/2)^2} \leq \frac{4}{sg^2} \leq \frac{1}{3},
\end{align*}
completing the proof of accuracy, and consequently the proof of \Cref{thm:reg-star-dist-unrestr}.
\end{proof}

\subsection{Degrees-Only \texorpdfstring{\LNDP}{Noninteractive LNDP} Cannot Distinguish Starpartite and Regular Graphs}
\label{sec:indist-star-reg}

We now prove \cref{thm:od-indist}, which shows that if a %
\emph{degrees-only} \edLNDP algorithm distinguishes between uniformly random $t$-starpartite and $t$-regular graphs on $n$ nodes with high probability, then $t = \Omega\bparen{\frac{\sqrt{n}}{\eps}}$.
This is in contrast with our result in \cref{sec:unrestricted-star-vs-reg}, which shows an %
\emph{unrestricted} \edLNDP algorithm for distinguishing uniformly random $t$-starpartite and $t$-regular graphs on $n$ nodes for some $t$ independent of $n$.

We prove \cref{thm:od-indist} by reducing from a basic statistical problem in the standard notion of local differential privacy. The reduction is nontrivial, since the standard LDP model allows for arbitrarily-selected inputs, while the inputs to an \LNDP algorithm are correlated by the fact that they must represent an actual graph. In our case, we want to ensure that the reduction generates degree lists that are either that of a $t$-regular graph (all inputs are $t$) or that of $t$-starpartite graph ($t$ inputs are $n-1$, and the rest are $t$). To do so, we reduce from the (standard-model LDP) problem of distinguishing an input of all 0's from an input where exactly $t$ randomizers receive input 1. Although similar problems have been considered before in the local model~\cite{BeimelNO08,DuchiJW13,JosephMNR19}, existing lower bound frameworks apply to product distributions on the inputs. The highly correlated input distributions we consider require a new, direct lower bound proof.

\begin{lemma}[\LNDP implies LDP] 
\label{lem:lndp-is-lnp}
    Let $\eps>0, \del\in[0,1]$, and $n\in\N$.
    If $\cA$ is an \edLNDP algorithm on $n$ nodes, then $\cA$ is also a public-coin noninteractive $(\eps,\del)$-LDP algorithm on $n$ parties.
\end{lemma}

\begin{proof}[Proof of \cref{lem:lndp-is-lnp}]
    For algorithm $\cA$ with distribution $\distpubrand$ over public randomness, let $\cP$ denote its postprocessing algorithm and
    $\cR_{1,\pubrand},\ldots,\cR_{n,\pubrand}$ with $\pubrand\sim \distpubrand$
    denote its local randomizers (by definition, every \LNDP algorithm can be written like this).
    Additionally, recall that an algorithm is $(\eps,\del)$-LDP
    if it can be written as the composition of some postprocessing algorithm and a set of local randomizers, where each randomizer's distribution of outputs is $(\eps,\del)$-indistinguishable for all pairs of inputs to that randomizer.

    By Item (a) of \cref{clm:priv-loss-facts}, for all $i\in[n]$, \cw{fixed public randomness $\pubrand$,} and all pairs $X_i, X'_i$ of edge lists for node $i$,
    \cw{
    \[
        \cR_{i,\pubrand}(X_i) \approx_{\eps,\del} \cR_{i,\pubrand}(X'_i).
    \]
    }
    Therefore, where we use $\cP$ as the postprocessing algorithm, \cw{distribution $\distpubrand$ for drawing public randomness $\pubrand\sim\distpubrand$,} and randomizers
    \cw{$(\cR_{1,\pubrand},\ldots,\cR_{n,\pubrand})$,}
    we see that $\cA$ satisfies all the criteria of a \az{public-coin} noninteractive $(\eps,\del)$-LDP algorithm. %
\end{proof}

\begin{lemma}
\label{lem:fixed-weight-indist}
    Let $n\in\N$, $\eps\in(0,1]$, and $\delta \in [0, 1)$.
    Let $P$ and $Q$ be the uniform distributions over bit strings in $\zo^n$ with exactly zero ``1''s and exactly $t$
    ``1''s, respectively.
    Let $\cA$ be a \cw{public-coin} noninteractive $(\eps,\del)$-LDP algorithm, where each user receives one bit from the string specified by either $P$ or $Q$. Then
    \[
        \dtv\bparen{\cA(P), \cA(Q)} 
        \leq
        \frac{4t\eps}{\sqrt{n}} + \delta n.
    \]
\end{lemma}

\cref{lem:fixed-weight-indist} differs from the standard statement about bit summation under LDP: standard bit-summation lower bounds (e.g., \cite{BeimelNO08,DuchiJW13,BassilyS15,JosephMNR19}) show that LDP algorithms cannot distinguish between the two settings where each input is set to
``1'' with probability $p$ or probability $q$, independently of other inputs, 
where $\abs{p-q}\approx \frac{1}{\eps \sqrt{n}}$.
Existing lower bounds in this setting take advantage of 
independence to reduce to analyzing an $\eps'$-LDP algorithm that aims to distinguish all ``0''s from all ``1''s, where $\eps' = O\bparen{\eps(p-q)} = O(1/\sqrt{n})$.
However, we need a lower bound showing that LDP algorithms cannot distinguish between two distributions with a fixed number of ``0''s and ``1''s, where the number of ``1''s differs by $\Theta\bparen{\frac{\sqrt{n}}{\eps}}$. Because the inputs here are correlated (e.g., if input $i$ is ``1'', then input $j\neq i$ is less likely to be ``1''), 
we rely on a different technique to prove \cref{lem:fixed-weight-indist}.

Our proof of \cref{lem:fixed-weight-indist} uses the ``simulation lemma'' from \cite{KairouzOV15}. We use the version presented in \cite[Lemma 3.2]{MurtaghV18}, which relies on the following definition.

\begin{definition}[Leaky randomized response; from \cite{KairouzOV15}, as presented in \cite{MurtaghV18}]
\label{def:leaky-rr}
    Define the \emph{leaky randomized response} function $\lrr:\zo\to \set{0,1,2,3}$ as follows, setting $\alpha = 1-\del$:
    \begin{align*}
      \Pr[\lrr(0) = 0] &= \alpha\cdot\tfrac{e^\eps}{1+e^\eps}   &   \Pr[\lrr(1) = 0] &= \alpha\cdot\tfrac{1}{1+e^\eps} \\
      \Pr[\lrr(0) = 1] &= \alpha\cdot\tfrac{1}{1+e^\eps} &   \Pr[\lrr(1) = 1] &= \alpha\cdot \tfrac{e^\eps}{1+e^\eps} \\
      \Pr[\lrr(0) = 2] &= \del   &   \Pr[\lrr(1) = 2] &= 0 \\
      \Pr[\lrr(0) = 3] &= 0   &   \Pr[\lrr(1) = 3] &= \del.
    \end{align*}
\end{definition}

Note that $\lrr$ is $(\eps,\del)$-DP. The simulation lemma of \cite{KairouzOV15} shows that $\lrr$ can be used to simulate %
an arbitrary $(\eps,\del)$-DP algorithm on neighboring inputs.

\begin{lemma}[Simulation lemma of \cite{KairouzOV15}, as presented in \cite{MurtaghV18}]
\label{lem:sim-via-lrr}
    Let $\eps > 0$ and $\del\in[0,1]$. 
    For every $(\eps,\del)$-DP algorithm $\cM$ and pair of neighboring datasets $D_0,D_1$, there exists a randomized algorithm $T$ such that $T\bparen{\lrr(b)}$ and $\cM(D_b)$ are identically distributed for $b\in\zo$.
\end{lemma}

We now prove \cref{lem:fixed-weight-indist}.

\begin{proof}[Proof of \cref{lem:fixed-weight-indist}]
    Let $\cR^\eps$ denote the standard \emph{$\eps$-LDP randomized response} algorithm (\cref{lem:rr-alg}) that, on input $b\in\zo$, returns $b$ with probability $\frac{e^\eps}{e^\eps + 1}$ and returns $1-b$ with probability $\frac{1}{e^\eps + 1}$. Let $\overrightarrow{\cR}^\eps \colon \zo^n\to\zo^n$ denote the algorithm that applies $\cR^\eps$ to every element of a length-$n$ bit string.

    We first show that we can assume, at little loss of generality, that every local randomizer is a copy of $\cR^\eps$.
    By \cref{lem:sim-via-lrr}, we can write
    every \cw{public-coin} noninteractive $(\eps,\del)$-LDP algorithm as $\cP\bparen{T'_1(\lrr(\cdot)),\ldots,T'_n(\lrr(\cdot))}$. 
    Consider the event $E$ that none of the instances of leaky randomized response return ``2'' or ``3'': by a union bound over all $n$ randomizers, we have $\Pr[E] \geq 1 - \del n$.
    Conditioned on $E$, the output has the same distribution as $\cP\bparen{T'_1(\cR^\eps(\cdot)),\ldots,T'_n(\cR^\eps(\cdot))} = T\bparen{\overrightarrow{\cR}^\eps(\cdot)}$. %
    By the data processing inequality for total variation distance, $\dtv\bparen{T(\rrvec(P)), T(\rrvec(Q))} \leq \dtv\bparen{\rrvec(P), \rrvec(Q)}$.
    
    For all $x\in\zo^n$, let $\cS(x) = \sum_{i\in[n]}x_i$ denote the summation function. Since the random variables in both $\rrvec(P)$ and $\rrvec(Q)$ are exchangeable, the sums 
    $\cS(\rrvec(P))$ and $\cS(\rrvec(Q))$ are sufficient statistics for distinguishing $\rrvec(P)$ from $ \rrvec(Q)$---that is,
    $\dtv\bparen{\rrvec(P), \rrvec(Q)} =
    \dtv\bparen{\cS(\rrvec(P)), \cS(\rrvec(Q))}$.
    Summarizing the argument so far:
    \begin{equation}
        \label{eq:suff-bd-on-tv}
        \dtv\bparen{\cA(P),\cA(Q)} \leq \dtv\bparen{\cS(\rrvec(P)), \cS(\rrvec(Q))} + \delta n.
    \end{equation}

    We now analyze the distributions of $\cS(\rrvec(P))$ and $\cS(\rrvec(Q))$.
    Define $A\sim\Bin\bparen{n,p}$, and $B \sim \Bin(n, q)$, where $p = \frac{1}{e^\eps + 1}$ and $q = \dfrac{1+\frac{t}{n} (e^\eps - 1)}{1+e^\eps}$. We show the following three conditions hold:
    \begin{enumerate}[label=(\alph*)]
        \item $\cS(\rrvec(P)) \sim \Bin\bparen{n,p}$, i.e., $\cS(\rrvec(P))$ and $A$ are identically distributed,
        \item $\dtv\bparen{ \cS(\rrvec(Q)), B } < \frac{t(e^\eps - 1)^2}{ne^\eps}$, and
        \item $\dtv(A, B) \leq \frac{t(e^\eps-1)}{e^{\eps/2} \sqrt{2n}}$.
    \end{enumerate}
    We then show that these conditions give us a bound on $\dtv\bparen{\cS(\rrvec(P)), \cS(\rrvec(Q))}$, and consequently our desired bound on $\dtv\bparen{\cA(P),\cA(Q)}$.

    \paragraph*{Proof of condition (a)}
    Note that $\rr(0)\sim\Bern\paren{\frac{1}{e^\eps + 1}}$ and $\rr(1)\sim\Bern\paren{\frac{e^\eps}{e^\eps + 1}}$.
    Since $P$ generates the all-``0'' string with probability one, we have
    $\cS(\rrvec(P))\sim\Bin\bparen{n,p=\frac{1}{e^\eps + 1}}$.%
    \paragraph*{Proof of condition (b)}
    The distribution of $\cS(\rrvec(Q))$ is 
    the sum of $n-t$ independent $\Bern\paren{\frac{1}{e^\eps + 1}}$ random variables and $t$ independent $\Bern\paren{\frac{e^\eps}{e^\eps + 1}}$ random variables. Such a distribution is a (special case of a) \textit{Poisson binomial}. We can apply a known upper bound of \cite{Ehm91} on the TV distance between a Poisson binomial and the single binomial with the same $n$ and  matching mean. 
    \cref{lem:apply-pois-bin} encapsulates Ehm's bound for our setting of parameters, implying that
    \[
        \dtv\bparen{ \cS(\rrvec(Q)), B } < \frac{ t \cdot (e^\eps-1)^2}{(n+1)\cdot e^\eps} < \frac{t(e^\eps - 1)^2}{ne^\eps}.%
    \]
    \paragraph*{Proof of condition (c)}
    For all $i\in[n]$, let $X_i\sim \Bern(p)$ and $Y_i\sim\Bern(q)$, so that $A = \sum_{i\in[n]} X_i$ and $B = \sum_{i\in[n]} Y_i$. We have
    \begin{align*}
        \dtv\paren{A, B}
        &\leq \dtv\bparen{(X_1,\ldots,X_n),(Y_1,\ldots,Y_n)} \tag{data processing inequality} \\
        &\leq \sqrt{\frac{1}{2}\dkl\bparen{(X_1,\ldots,X_n) \|(Y_1,\ldots,Y_n)}} \tag{Pinsker's inequality} \\
        &= 
        \Bparen{\frac{1}{2}
        \sum_{i\in[n]} \dkl(X_i \| Y_i)}^{1/2}.
        \tag{tensorization of KL divergence}
    \end{align*}
    Since $X_i$ and $Y_i$ are i.i.d.\ for all $i\in[n]$, it suffices to bound $\dkl(X_i \| Y_i)$.
    By \cref{lem:bern-kl}, $\dkl(X_i \| Y_i) \leq \frac{(p-q)^2}{q(1-q)}$.
    Note that $q\in\bracks{\frac{1}{e^\eps + 1}, \frac{e^\eps}{e^\eps + 1}}$, so $q(1-q) \geq \frac{e^\eps}{(e^\eps + 1)^2}$.
    For all $i\in[n]$, substituting for $p$ and $q$ gives us 
    \begin{align*}
        \dkl(X_i \| Y_i)
        \cond{&}
        \leq \paren{\frac{t}{n}}^2\cdot \paren{\frac{e^\eps-1}{e^\eps+1}}^2 \cdot \frac{(e^\eps  +1)^2}{e^\eps}
        \cond{\\&}
        = \paren{\frac{t}{n}}^2 \cdot \frac{(e^\eps-1)^2}{e^\eps}.
    \end{align*}
    Substituting this into our above bound on $\dtv(A,B)$ gives
    \(
        \dtv(A,B)
        \leq \frac{t(e^\eps -1)}{e^{\eps/2}\sqrt{2n} }.
    \)

    \paragraph*{Combining conditions (a), (b), (c)}
    We have
    \begin{align*}
        \dtv&\bparen{\cS(\rrvec(P)), \cS(\rrvec(Q))}\\
        &= \dtv(A, \cS(\rrvec(Q))) \tag{by condition (a)} \\
        &\leq \dtv(A, B) + \dtv\bparen{ \cS(\rrvec(Q)), B } \tag{\text{since $\dtv$ is a metric}} \\
        &\leq \frac{t(e^\eps -1)}{e^{\eps/2}\sqrt{2n} } + \frac{t(e^\eps - 1)^2}{ne^\eps} \tag{using conditions (b) and (c)} \\
        &\leq \frac{4t\eps^2}{n} + \frac{2t \eps}{\sqrt{2n}} \leq \frac{4t \eps}{\sqrt{n}}. \tag{\text{using $\eps \in (0, 1)$ and $e^x \leq 1+2x$ for $x \in[0,1]$}}
    \end{align*}
    Finally, \cref{eq:suff-bd-on-tv} implies that 
    \begin{align*}
        \dtv \bparen{\cA(P), \cA(Q)}
        \cond{&}
        \leq \dtv\bparen{\cS(\rrvec(P)), \cS(\rrvec(Q))} + \delta n
        \cond{\\&}
        \leq  \frac{4t \eps}{\sqrt{n}} + \delta n. \qedhere
    \end{align*}
\end{proof}

We now prove \cref{thm:od-indist}.

\begin{proof}[Proof of \cref{thm:od-indist}]
    Let $t = 2\floor{\frac{\sqrt{n}}{200\eps}}$.
    Let $P$ and $Q$ be the uniform distributions over bit strings in $\zo^n$ with exactly zero ``1''s and exactly $t$ ``1''s, respectively.
    Our proof has the following structure. Assume for contradiction that there is a degrees-only \edLNDP algorithm $\cA$ such that
    \begin{gather*}
        \Pr_{G\sim \distgstar{t}}[\cA(G) = \text{``starpartite''}] \geq 2/3
        \condelse{\\ \text{and} \\}{\quad \text{and} \quad}
        \Pr_{G\sim \distgreg{t}}[\cA(G) = \text{``regular''}] \geq 2/3.
    \end{gather*}
    We show that then there is a public-coin noninteractive $(\eps,\del)$-LDP algorithm $\cB$ such that
    \(
        \dtv\bparen{\cB(P), \cB(Q)} > \frac{1}{10},
    \)
    which contradicts \cref{lem:fixed-weight-indist}.

    Let 
    $\cR_{1,\pubrand},\ldots,\cR_{n,\pubrand}$
    denote the randomizers of $\cA$ with public randomness $\pubrand\sim \distpubrand$, and let $\cP$ be the postprocessing algorithm. Algorithm $\cB$ is as follows. Generating public randomness as in $\cA$ (i.e., by drawing %
    $\pubrand\sim\distpubrand$), for all $i\in[n]$, if individual $i$ holds $0$ then run $\cR_{i,\pubrand}(t)$;
    otherwise, 
    run
    $\cR_{i,\pubrand}(n-1)$;
    send the output to the central server. The central server runs $\cP$ on the outputs. If the result is ``regular'', return ``$P$''; if the result is ``starpartite'', return ``$Q$''.
    
    We next analyze the privacy and accuracy of $\cB$. By \cref{lem:lndp-is-lnp}, $\cA$ is a noninteractive $(\eps,\del)$-LDP algorithm. Because the transformation in $\cB$ from a bit to either $t$ or $n-1$ can be performed locally, the resulting algorithm $\cB$ is also a noninteractive $(\eps,\del)$-LDP algorithm. Additionally, note that if the input is from $P$, then the list of values provided as input is the same as the degree list of a uniformly random $t$-regular graph on $n$ nodes (i.e., every party holds input $t$, and $t$ is even, so a $t$-regular graph on $n$ nodes must exist). Likewise, if the input is from $Q$, then the list of values provided as input is the same as the degree list of a uniformly random $t$-starpartite graph on $n$ nodes (i.e., a uniformly random subset of $t$ parties holds $n-1$, and all other parties hold $t$). Therefore, by the accuracy of $\cA$, algorithm $\cB$ has the property
    \[
        \Pr_{X\sim P}[\cB(X) = \text{``$P$''}] \geq 2/3
        \quad \text{and} \quad
        \Pr_{X\sim Q}[\cB(X) = \text{``$Q$''}] \geq 2/3.
    \]
    Let $E$ denote the event that $\cB$ outputs ``$P$''.
    By the second inequality above,
    $\Pr_{X\sim Q}[\cB(X) = \text{``$P$''}] \leq \frac{1}{3}$. Therefore, the difference in probability witnessed by event $E$ means that
    $\dtv\bparen{\cB(P), \cB(Q)} \geq \frac{1}{3}$.
    This contradicts the fact from \cref{lem:fixed-weight-indist} that
    \(
        \dtv\bparen{\cB(P), \cB(Q)} \leq \frac{1}{10}.
    \)
    Thus, the assumption about the accuracy of $\cA$ must be false.
\end{proof}

\AtNextBibliography{\small}
\printbibliography[heading=bibintoc]

\ifnum\focsversion=1
    \appendix
\else
    \newpage
    \appendix
    \section*{Appendix}
    \addcontentsline{toc}{section}{Appendix}
\fi

\section{Background on Differential Privacy}
\label{sec:dp-background}

This section collects useful background on differential privacy: common mechanisms, standard properties of the definition, and the usual notion of noninteractive local DP for tabular data.

A common way to release a function's value under DP is to add noise scaled to its sensitivity.

\begin{definition}[Sensitivity] Let $k\in\mathbb{N}$ and $f:\cX \rightarrow \mathbb{R}^k$
be a function. Let $p\in\{1,2\}$. 
The \emph{$\ell_p$-sensitivity of $f$}, denoted by $\Delta_p$, is defined as
    $\Delta_p=\max_{x\sim y}\|f(x)-f(y)\|_p.$
\end{definition}

\begin{lemma}[Laplace mechanism~\cite{DworkMNS16}] \label{lem:Laplacemech} Let $k\in\N$ and $\eps>0$ and $f:\cX^n \rightarrow \mathbb{R}^k$ 
be a function
with $\ell_1$-sensitivity $\Delta_1$. 
The Laplace mechanism is defined as
$\cA(x)=f(x)+(Z_1,\dots,Z_k)$, where the $Z_i \sim 
\Lap(\Delta_1/\eps)$ are independent for all $i\in[k]$. The Laplace mechanism is $(\eps,0)$-DP.
\end{lemma}

\begin{lemma}[Gaussian mechanism~\cite{BlumDMN05,BunS16}]\label{lem:gaussianmech}
 Let $k\in\mathbb{N}$ and $f:\cU^{*}\rightarrow \mathbb{R}^k$ 
be a function  with $\ell_2$-sensitivity $\Delta_2$.
Let $\eps\in(0,1)$, $\delta\in(0,1)$, 
$c > \sqrt{2 \ln(1.25/\delta)}$, %
and $\sigma\geq c\Delta_2/\eps$. 
The Gaussian mechanism is defined as $\mathcal{A}(x)=f(x)+(Z_1,\dots,Z_k)$, where the $Z_i\sim N(0,\sigma^2)$ are independent for all $i\in[k]$, and
is $(\eps,\delta)$-DP.
\end{lemma}

Differential privacy is robust to postprocessing (i.e., the application of an arbitrary randomized algorithm to the output of a differentially private algorithm), and its parameters degrade gracefully under composition.

\begin{lemma}[DP is robust to postprocessing \cite{DworkMNS16}]
\label{lem:post-proc}
Let $\cM:\cU^*\to \cY$ be a randomized algorithm that is $(\eps,\del)$-DP. Let $f:\cY\to \cZ$ be a randomized function. Then $f\circ \cM:\cU^*\to \cZ$ is $(\eps,\del)$-DP.
\end{lemma}

\begin{lemma}[Basic composition~\cite{DworkL09,DworkRV10,DworkKMMN06}]
\label{lem:composition_theorem}
Let $\eps_1,\eps_2> 0$ and $\delta_1,\delta_2\in[0,1)$.  
Let $\cA_1 \colon \cU^* \to \cY$ be an $(\eps_1, \delta_1)$-DP algorithm and let $\cA_2 \colon \cU^* \times \cY \to \cZ$ be an $(\eps_2, \delta_2)$-DP algorithm.
Then for all $x\in\cU^*$, algorithm
$\cA_2(x,\cA_1(x))$
is $(\eps_1+\eps_2,\delta_1+\delta_2)$-DP. 
\end{lemma}

\begin{lemma}[Advanced composition~\cite{DworkRV10}]
\label{lem:adv-comp}
For all $\eps > 0$, $\delta \geq 0$ and $\delta' > 0$, the adaptive composition of $k$ algorithms which are $(\eps, \delta)$-DP is $(\tilde{\eps}, \tilde{\delta})$-DP, where
$\tilde{\eps} = \eps \sqrt{2k \ln(1/\delta')} + k \eps \frac{e^{\eps} - 1}{e^{\eps} + 1}$, and $\tilde \delta = k \delta + \delta'$.
\end{lemma}

Differential privacy also extends to groups: an algorithm that is DP for individuals also provides (weaker) guarantees for groups, with adjusted parameters. We state the formulation of this property from \cite{Vadhan17}.

\begin{lemma}[DP offers group privacy]
\label{lem:group-priv}
Let $\cA:\cU^*\to \cY$ be an $(\eps,\del)$-DP algorithm. If $x,x'\in \cU^*$ differ in at most $k$ entries, then $\cA(x)$ and $\cA(x')$ are $(k\cdot \eps, k\cdot e^{k\eps} \cdot \del)$-indistinguishable, that is,
$
    \cA(x) \approx_{k\cdot \eps, k\cdot e^{k\eps} \cdot \del} \cA(x').
$
\end{lemma}

We recall the standard definition of the noninteractive local model for tabular datasets.

\begin{definition}[Noninteractive local differential privacy (LDP)] 
\label{def:ldp}
    Let $\eps > 0$, $\del\in[0,1]$, and $n\in \N$.
    An algorithm $\cA:\cX^n\to \cY$ is \emph{noninteractive and local}
    if it can be written in the form
    \[
        \cA(X) = \cP\bparen{\cR_{1,\pubrand}(X_1),\ldots,\cR_{n,\pubrand}(X_n)}
    \]
    for some set of \emph{local randomizers} $\cR_{i,\pubrand}:\cX\to \cZ$, public randomness distribution $\distpubrand$, and a \emph{postprocessing algorithm} $\cP:\cZ^n\to \cY$. If $\cR_{i,\pubrand} :\cX\to\Z$ has the property $\cR_{i,\pubrand}(X_i)\approx_{\eps,\del}\cR_{i,\pubrand}(X'_i)$ for all $\pubrand\sim\distpubrand$, $X_i,X'_i\in\cX$, and $i\in[n]$, we say that $\cA$ satisfies \emph{public-coin noninteractive $(\eps,\del)$-local differential privacy (LDP)}.
\end{definition}

Many LDP algorithms are built from the \emph{randomized response} primitive of \cite{Warner65}, which was shown by \cite{DworkMNS16} to satisfy $(\eps,0)$-LDP.

\begin{lemma}[Randomized response \cite{Warner65,DworkMNS16}]
\label{lem:rr-alg}
    Let $\eps > 0$. Define $\cR^\eps:\zo \to \zo$ as the standard \emph{$(\eps,0)$-LDP randomized response algorithm} that, on input $b\in\zo$, returns $b$ with probability $\frac{e^\eps}{e^\eps + 1}$ and returns $1-b$ with probability $\frac{1}{e^\eps + 1}$. Then $\cR^\eps$ satisfies $(\eps, 0)$-LDP.
\end{lemma}

\section{Useful Probability Results}
\label{sec:appendix}

\subsection{Properties of Gaussians}

\begin{lemma}[Gaussian concentration bounds]
\label{lem:gauss-concen}
    For $a>0$ and Gaussian random variable $Z\sim \cN(0, \sigma^2)$,
    \[
        \Pr[ Z\geq a ] \leq \exp\left( \frac{-a^2}{2\sigma^2} \right).
    \]
\end{lemma}

\begin{proof}[Proof of \cref{lem:gauss-concen}]
    Let $Z\sim \cN(0,\sigma^2)$. A Gaussian random variable $X\sim \cN(0,\sigma^2)$ has moment-generating function $M_X(t) = \exp\left( \frac{\sigma^2 t^2}{2} \right)$. By the Chernoff bound,
    \[
        \Pr[Z\geq a] \leq \inf_{t > 0} \exp\paren{ \frac{\sigma^2 t^2}{2} - ta }. 
    \]
    This expression is minimized at $t = \frac a{\sigma^2}$. Therefore,
    $
        \Pr[ Z\geq a ] \leq \exp\paren{ \frac{-a^2}{2\sigma^2} }.
    $
\end{proof}

\begin{lemma}[Maximum magnitude of Gaussians]
\label{lem:max-gauss}
    Let $\beta \in (0, 1)$ %
    and $m\in \N$.
    If $Z_i\sim \cN\paren{0, \sigma^2}$
    for all $i\in [m]$, then
    \[
        \Pr\left[\max_{i\in[m]} |Z_i| \geq \sigma \sqrt{2\ln (2m/\beta)} \right] \leq \beta.%
    \]
\end{lemma}

\begin{proof}[Proof of \cref{lem:max-gauss}]
    For each $i\in[m]$, since the normal distribution is symmetric, \Cref{lem:gauss-concen} implies
    \[
        \Pr[ |Z_i| \geq a ] \leq 2 \exp\left( \frac{-a^2}{2\sigma^2} \right).
    \]
    Setting $a = \sigma \sqrt{2 \ln (2 m/\beta)}$ and taking a union bound over all $i \in [m]$ gives
    \[
        \Pr\left[ \max_{i\in[m]} |Z_i| \geq \sigma \sqrt{2 \ln (2m/\beta)} \right] \leq m \cdot \frac{\beta}{m} = \beta. %
        \qedhere
    \]
\end{proof}

\subsection{Useful Tail Bounds}

\begin{lemma}[Tails for binomial distributions]
\label{lem:bin-tails}
    Let $n\in \N$, $p\in [0,1]$, and $\beta\in [0,1]$.
    Then the tails of $X\sim \Bin(n,p)$ have the following behavior:
    \begin{enumerate}
        \item
        \textit{(lower tail)}
        
         $\Pr\bracks{X\leq np - \sqrt{2np\ln(1/\beta)}} \leq \beta$, and
        \item
        \textit{(upper tail)}
        
        $\Pr\bracks{X\geq np + \max\set{3\ln(1/\beta),\sqrt{3np \ln(1/\beta)}}} \leq \beta$.
    \end{enumerate}
\end{lemma}
\begin{proof}[Proof of \cref{lem:bin-tails}]
    The statement follows by a standard Chernoff-Hoeffding bound. \qedhere

\end{proof}

\begin{lemma}[Degree bounds for $\ergraph$ graphs]
\label{thm:er-tails}
    Let $n\in \N$ and $p,\beta\in [0,1]$.
    Let $G\sim G(n,p)$ be an $\ergraph$ graph, and let $D^\mathit{min}_G$ and $D^\mathit{max}_G$ denote its minimum and maximum degrees, respectively. 
    \begin{enumerate}
        \item \textit{(lower tail)}
        
        \(
            \Pr\bracks{D^\mathit{min}_G
            \leq
            (n-1)p - \sqrt{2(n-1)p\ln(n/\beta)}}
            \leq \beta
        \), and
        \item \textit{(upper tail)}
        
        \(\cond{\!\!\!\!\!\!\!\!}
            \Pr\cond{\!}\bracks{ \cond{\!}D^\mathit{max}_G \cond{\!}
            \geq
             \cond{\!}(n\cond{\!}-\cond{\!}1)p \cond{\!} + \cond{\!} \max \cond{\!}\set{ \scalebox{\condelse{0.85}{1}}{$\cond{\!}3 \cond{\!}\ln(n/\beta) , \cond{\!}\sqrt{3(n\cond{\!}-\cond{\!}1)p\ln(n/\beta)}$}}}
            \cond{\!}\leq \cond{\!}\beta.
        \)
    \end{enumerate}
\end{lemma}

\begin{proof}[Proof of \cref{thm:er-tails}]
    Each node's degree in an $\ergraph$ graph has distribution $\Bin(n-1,p)$.
    The statement follows by setting ``$\beta$'' in \cref{lem:bin-tails} to $\frac{\beta}{n}$ and taking a union bound over all $n$ nodes. 
\end{proof}

\subsection{Approximating a Poisson Binomial}

In \cref{sec:indist-star-reg} we use \cref{lem:pois-bin-approx} \cite{Ehm91}, which shows that a Poisson binomial distribution can be well approximated by a binomial distribution.

\begin{lemma}[\cite{Ehm91}]
\label{lem:pois-bin-approx}
    Let $X_1,\ldots,X_n$ be independent Bernoullis, with $X_i\sim \Bern(p_i)$.
    Let $A = \sum_{i\in[n]} X_i$,\footnote{That is, $A$ is a \emph{Poisson binomial}.} and let $\mu = \frac{1}{n}\sum_{i\in[n]} p_i$ and $\nu = 1-\mu$.
    If $\mu \in(0,1)$ and $B\sim \Bin(n,\mu)$, then
    \[
        \dtv(A, B)
        \leq
        \paren{1 - \mu^{n+1} - \nu^{n+1}} \cdot \frac{\sum_{i\in[n]} (p_i - \mu)^2 }{(n+1)\mu\nu}.
    \]
\end{lemma}

\begin{lemma}
\label{lem:apply-pois-bin}
    Let $\eps > 0$, $n\in \N$, and $k\in[n]$.
    Let $X_1,\ldots,X_n$ be independent Bernoullis, where 
    \[
        X_i\sim
        \begin{cases}
            \Bern\bparen{\frac{1}{e^\eps + 1}} & \text{if $i\in [n] \setminus [k]$, and} \\
            \Bern\bparen{\frac{e^\eps}{e^\eps + 1}} & \text{if $i\in [k]$},
        \end{cases}
    \]
    and let $A = \sum_{i\in[n]} X_i$.
    Let $p = \dfrac{1+\frac{k}{n}\cdot (e^\eps - 1)}{1+e^\eps}$.
    If $B\sim \Bin(n,p)$, then
    $\dtv\paren{A,B} < \dfrac{ k \cdot (e^\eps-1)^2}{(n+1)\cdot e^\eps}$.
\end{lemma}

\begin{proof}[Proof of \cref{lem:apply-pois-bin}]
    We apply \cref{lem:pois-bin-approx}. We first show that, where $\mu$ is defined as in \cref{lem:pois-bin-approx}, we have $p = \mu$. We then show that our bound on $\dtv(A,B)$ is bounded above by the term in \cref{lem:apply-pois-bin}.

    By the definition of $\mu$ in \cref{lem:pois-bin-approx}, we have
    \begin{align*}
        \mu
        &= \frac{1}{n} \paren{ \sum_{i\in[n]\setminus[k]} \frac{1}{e^\eps + 1} + \sum_{i\in[k]} \frac{e^\eps}{e^\eps + 1} } 
        \cond{\\&}
        = \frac{1}{n} \paren{ \frac{n-k}{1+e^\eps} + \frac{k\cdot e^\eps}{1+e^\eps} } 
        = \frac{1+\frac{k}{n}\cdot (e^\eps-1)}{1+e^\eps} =: p.
    \end{align*}
    Let $p_i = \Pr[X_i = 1]$, $q = 1-p$, and note that $p,q\in\bracks{\frac{1}{e^\eps + 1}, \frac{e^\eps}{e^\eps + 1}}$. If $B\sim\Bin(n,p)$, then, by \cref{lem:pois-bin-approx},
    \begin{equation}
    \label{eq:pois-bin-tv-bd}
        \dtv(A,B) < \frac{1}{n+1} \cdot \frac{(e^\eps + 1)^2}{e^\eps} \cdot \sum_{i\in[n]}(p_i-p)^2.
    \end{equation}
    We next upper bound $\sum_{i\in[n]}(p_i-p)^2$ as follows:
    \begin{align*}
        \cond{&}\sum_{i\in[n]} (p_i-p)^2
        \cond{\\}&\cond{\quad}= \frac{\bracks{ k \paren{  \paren{e^\eps - 1} -  \frac{k}{n} \paren{e^\eps - 1} }^2  + \paren{n-k} \paren{\frac{k}{n} (e^\eps - 1)}^2 }}{(e^\eps + 1)^2}  \\
        &\cond{\quad}= \frac{\bracks{ k \paren{e^\eps - 1}^2 \paren{1-\frac{k}{n}}^2 + \paren{n-k} \paren{e^\eps - 1}^2 \paren{\frac{k}{n}}^2 }}{\paren{e^\eps + 1}^2}   \\
        &\cond{\quad}= \frac{n}{\paren{e^\eps + 1}^2} \bracks{ \frac{k}{n} \paren{e^\eps - 1}^2 \paren{1-\frac{k}{n}}\paren{1 - \frac{k}{n} + \frac{k}{n}} } \\
        &\cond{\quad}= \frac{ k \paren{e^\eps - 1}^2\paren{1-\frac{k}{n}} }{(e^\eps + 1)^2} 
        \cond{\\&\quad}< \frac{k (e^\eps - 1)^2}{(e^\eps + 1)^2}.
    \end{align*}
    Substituting this bound into \cref{eq:pois-bin-tv-bd} gives us
    \(
        \dtv(A,B)
        <
        \dfrac{k\cdot (e^\eps - 1)^2}{(n+1) \cdot e^\eps},
    \)
    which is what we wanted to show.
\end{proof}

\subsection{KL Divergence Between Bernoullis}

We use the following statement in \cref{sec:indist-star-reg}.

\begin{lemma}
\label{lem:bern-kl}
    Let $p,q\in(0,1)$ such that $p\leq q$.
    If $P= \Bern(p)$ and $Q= \Bern(q)$, then
    \[
        \dkl(P \| Q ) \leq \frac{(p-q)^2}{q(1-q)}.
    \]
\end{lemma}

\begin{proof}[Proof of \cref{lem:bern-kl}]
    For $x>0$ we have $\ln x \leq x-1$.
    By the definition of KL divergence,
    \begin{align*}
        \dkl(P \| Q)
        &=
        p\ln\paren{\frac{p}{q}} + (1-p)\ln\paren{\frac{1-p}{1-q}}
        \cond{\\&}\leq
        p \paren{\frac{p}{q} - 1} + (1-p) \paren{\frac{1-p}{1-q} - 1} 
        \cond{\\&}= \frac{(p-1)^2}{q (1-q)}. \qedhere
    \end{align*}
\end{proof}

\section{Deferred Proofs from \texorpdfstring{\Cref{sec:unrestricted-star-vs-reg}}{Starpartite vs. Regular Distinguisher}} \label{sec:star-vs-reg-deferred-proofs}

In this section, we prove several technical lemmas used under the hood in \Cref{sec:unrestricted-star-vs-reg}.
\Cref{lem:neg-correl} shows that the variables $b_{i,j}$ from \Cref{alg:star-vs-reg} are negatively correlated, allowing us to obtain concentration bounds for them in \Cref{lem:bern-pstar-preg}.
\begin{lemma}
\label{lem:neg-correl}
Let $\distgreg{t}$ be the uniform distribution of $t$-regular graphs, and let $b_{i, j}$ be as in \cref{alg:star-vs-reg}. When $G \sim \distgreg{t}$, the random variables $b_{1, j}, \ldots, b_{n, j}$ are negatively correlated for all $j \in [s]$; that is, for all $j \in [s]$ and all $A \subseteq [n]$, we have
$$
\underset{S_j, G}{\E} \left[ \prod_{i \in A} b_{i, j} \right] \leq \prod_{i \in A} \underset{S_j, G}{\E} \left[ b_{i,j} \right].
$$
\end{lemma}
\begin{proof}

For Bernoulli random variables, it suffices to check the following: For all $j \in [n]$, $A \subseteq [n]$ and $i^* \in A$, 
\begin{equation}
\label{eqn:neg-correl-equiv}
    \Pr_{S_j, G}[b_{i^*,j} = 1 \;|\; b_{i,j} = 1 \;\forall i \in A \setminus \{i^*\}] \leq \Pr_{S_j, G}[b_{i^*,j} = 1],
\end{equation}
where randomness is taken over multisets $S_j$ of $[n]$ with $|S_j| = \frac{n}{t}$ and $G \sim \distgreg{t}$.
By \Cref{lem:prob-star-in-sj}, the unconditional probability on the right hand side is equal to
$$
\Pr_{S_j, G}[b_{i^*,j} = 1] = 1 - \left(1-\frac{t}{n}\right)^{n/t}.
$$

Fix $A \subseteq [n]$ and $i^* \in A$.
For $|A| = 1$, the conditional event is vacuous, so \cref{eqn:neg-correl-equiv} holds with equality. Assume that $|A| \geq 2$.
As each $\ell \in S_j$ is chosen independently at random with replacement, the conditional probability of the complement is %
\begin{align*}
    \cond{&}\Pr_{S_j, G}[b_{i^*,j} = 0 \;|\; b_{i,j} = 1 \;\forall i \in A \setminus \{i^*\}] 
    \cond{\\&\quad}
    = \E_{S_j} \left[ \Pr_{G}[i^* \notin \neigh{G}{\ell} \; \forall \ell \in S_j \mid b_{i,j} = 1 \; \forall i\in A \setminus \{i^*\}] \right].
\end{align*}

Fix $S_j$ and a random $G \sim \distgreg{t}$, and let $\ell \in S_j$. Conditioning on every $i\in A\setminus \{i^*\}$ having at least one neighbor in $S_j$, there must be some $\ell\in S_j$ and $k \in A \setminus \{i^*\}$ such that $k \in \neigh{G}{\ell}$; for that particular $\ell$ the probability that $i^*$ is not also a neighbor is at least $1-\tfrac{t-1}{n}$. For the remaining $\ell'$ in $S_j$, we still have the trivial lower bound $\Pr[i^* \notin \neigh{G}{\ell'}] \geq 1 - \tfrac{t}{n}$. So, for every fixed $S_j$, we have
\begin{align*}
    \cond{&}\Pr_{G}[i^* \notin \neigh{G}{\ell} \; \, \forall \ell \in S_j \mid b_{i,j} = 1 \; \forall i\in A \setminus \{i^*\}] 
    \cond{\\&\qquad} \geq \left(1 - \frac{t-1}{n}\right) \cdot \left(1 - \frac{t}{n}\right)^{(n/t)-1}.
\end{align*}
Taking an expectation over $S_j$ preserves this lower bound, giving us
\begin{align*}
    \cond{&}\Pr_{S_j, G}[b_{i^*,j} = 0 \;|\; b_{i,j} = 1 \;\forall i \in A \setminus \{i^*\}] 
    \cond{\\&\qquad} \geq \left(1 - \frac{t-1}{n}\right) \cdot \left(1 - \frac{t}{n}\right)^{(n/t)-1}
    \cond{\\&\qquad} \geq \left(1 - \frac{t}{n}\right)^{n/t }
    \cond{\\&\qquad} = \Pr_{S_j, G}[b_{i^*,j} = 1],
\end{align*}
and taking complements implies \Cref{eqn:neg-correl-equiv}.
\end{proof}
\begin{lemma}
\label{lem:nst-new-satisfy}
Let $\eps \in (0, \frac{1}{2})$, and $\delta \in (0, \frac{1}{10})$. Define $n, s, t \in \mathbb{R}$ as
\begin{gather*}
    n \geq \frac{3}{4} K \ln^5(2/\delta) \cedp^{10}, 
    \condelse{\\}{\quad} t = 30K \ln^2(2/\delta) \cedp^6, 
    \condelse{\\}{\quad} s = 3 t \ln(2/\delta),  
\end{gather*}
where $K = 72 \cdot 10^4 \cdot \pi$ and $\cedp= \frac{\sqrt{2 \ln (2.5/\delta)}}{\eps}$.
Define 
\begin{conddisplay}
    \savg = \cedp \sqrt{\frac{s}{n} + \frac{s}{t} + \sqrt{\frac{3s}{t} \ln \left(\frac{2}{\delta} \right)}},
\end{conddisplay}
as well as $\gamma = \frac{1}{200 \savg^2}$ and $p = 1 - (1 - \frac{t}{n})^{n/t}$. Then, the following conditions hold:

\begin{center}
\begin{enumerate*}[label=(\alph*)]
    \item $s \geq \frac{K}{3} \savg^6$, $\quad $
    \item $\savg \geq 1$, $\quad $
    \item $r := \sqrt{\frac{3p}{n} \ln(1/\gamma)} \leq \gamma$,  $\quad $
    \item $n\geq 3t$. $\quad $
\end{enumerate*}
\end{center}

\end{lemma}

\begin{proof}
For any $x > 0$, define $L_x = \ln(x/\delta)$. The stated restrictions on $n, s, t \in \mathbb{N}$ are:
\[
n \geq \frac{3K}{4} L_2^5 \cedp^{10}, \quad t = 30K L_2^2 \cedp^6, \quad s = 3t L_2.
\]
Next, define
\begin{align*}
    A &:= \frac{s}{n} + \frac{s}{t} + \sqrt{\frac{3s}{t} \ln \left(\frac{2}{\delta}\right)} 
    \cond{\\&}
    = \frac{s}{n} + \frac{s}{t} + \sqrt{\frac{3sL_2}{t}} 
    = \frac{s}{n} + 3L_2 + \sqrt{9 L_2^2} 
    = \frac{s}{n} + 6L_2,
\end{align*}
where we use $\frac{s}{t} = 3L_2$. We write $\savg = \cedp \sqrt{A}$.
First, we give upper and lower bounds on $A$. Since $s, n > 0$, we have $A = \frac{s}{n} + 6L_2 \geq 6L_2$.
For the upper bound, we use the definition of $s$ and $n$ as well as $\eps \in (0, \frac{1}{2})$ and $\del \in (0, \frac{1}{10})$ to get
\begin{align*}
A &= \frac{s}{n} + 6 L_2 
\leq \frac{120}{L_2^2 \cedp^4} + 6 L_2 
\cond{\\&}
\leq \frac{120}{\ln^2(20)} \cdot \frac{(1/2)^4}{(2 \ln(25))^2} + 6 L_2 \leq 6L_2^2 + \frac{1}{40},
\end{align*}
giving the bounds
\begin{equation}
\label{eqn:bound-on-a}
    6L_2 \leq A \leq 6L_2 + \frac{1}{40}.
\end{equation}

We now separately prove each inequality.

\paragraph*{Proof of condition (a)}
Plugging in the setting for $s$, we observe
\begin{align*}
s \geq \frac{K}{3} \savg^6 &\iff 90KL_2^3 \cedp^6 \geq \frac{K}{3} \cedp^6 A^3 
\cond{\\&}
\iff A \leq (270)^{1/3} L_2.
\end{align*}
The last inequality follows by \Cref{eqn:bound-on-a}, since $A \leq 6L_2 + \frac{1}{40} \leq \frac{19}{3} L_2 \leq (270)^{1/3} L_2$, using $L_2 \geq \ln(20)$ for $\del \in (0, \frac{1}{10})$.

\paragraph*{Proof of condition (b)}
Using $\savg = \cedp \sqrt{A}$,the lower bound from \Cref{eqn:bound-on-a}, and $\eps \in (0, \frac{1}{2})$ and $\del \in (0, \frac{1}{10})$, we obtain
\[
\savg^2 = \cedp^2 A \geq 6 \cedp^2 L_2 \geq 6 \cdot \frac{2 \ln(25)}{(1/2)^2} \cdot \ln(20) \geq 1.
\]

\paragraph*{Proof of conditions (c), (d)}
To analyze $p = 1 - (1 - \frac{t}{n})^{n/t}$, we lower bound $\frac{n}{t}$ using $\eps \in (0, \frac{1}{2})$, $\del \in (0, \frac{1}{10})$:
\begin{align*}
    \frac{n}{t} &\geq \frac{\frac{3}{4} K L_2^5 \cedp^{10}}{30 KL_2^2 \cedp^6} = \frac{90}{4} L_2^3 \cedp^4 
    \cond{\\&}
    \geq \frac{90}{4} \cdot (\ln(20))^3 \cdot \frac{(2 \ln(25))^2}{(1/2)^4} \geq 100.
\end{align*}
This proves condition (d), and also gives $p = 1 - (1-\frac{t}{n})^{n/t}  \in [\frac{6}{10}, \frac{7}{10}]$.
We then have
\begin{align*}
    \sqrt{\frac{3p}{n} \ln(1/\gamma)} \leq \gamma 
    \iff n \cond{&}\geq \frac{3p \ln(1/\gamma)}{\gamma^2} 
    \cond{\\&}= 3 \cdot 200^2 p \savg^4 \ln(200 \savg^2) 
    \cond{\\&}
    = 3 \cdot 200^2 p \cedp^4 A^2 \ln(200 \cedp^2 A).
\end{align*}
Since $p \leq \frac{7}{10}$ and $n \geq \frac{3K}{4} L_2^5 \cedp^{10}$, it suffices to show that
\begin{equation}
\label{eqn:cond-on-k}
\begin{split}
\cond{&}\frac{3K}{4} L_2^5 \cedp^{10} \geq \frac{21}{10} \cdot 200^2 \cedp^4 A^2 \ln(200 \cedp^2 A) 
\cond{\\}\iff\cond{&} K \geq \frac{M A^2 \ln(200 \cedp^2 A)}{L_2^5 \cedp^6},
\end{split}
\end{equation}
where $M = \frac{4}{3} \cdot \frac{21}{10} \cdot 200^2$. Bounding the right-hand side, we use $A \leq 6L_2 + \frac{1}{40} \leq 7L_2$ to get
\begin{align*}
    \frac{M A^2 \ln(200\cedp^2 A)}{L_2^5 \cedp^6} 
    \cond{&}\leq \frac{49 M L^2 \ln(200 \cedp^2 A)}{L_2^5 \cedp^6}
    \cond{\\&}\leq \frac{49 M \ln(1400 \cedp^2L_2)}{(\cedp^2L)^3}
    \leq M,
\end{align*}
where in the last inequality we use that the function $f(x) = \frac{\ln(1400x)}{x^3}$ satisfies $f(x) \leq \frac{1}{50}$ for all $x \geq 8$. The inequality $\cedp^2 L_2 \geq 8$ follows from $L_2 \geq \ln(20)$ and $\cedp \geq \frac{2\ln(25)}{1/8}$. Setting $K = 72 \cdot 10^4 \cdot \pi \geq M \geq \frac{M A^2 \ln(200c^2 A)}{L_2^5 \cedp^6}$ shows \Cref{eqn:cond-on-k}, and consequently shows condition (c), completing the proof.
\end{proof}

\begin{lemma}
\label{lem:integral-diff}
\az{
Let $n, s, t \in \N$. Define $K = 72 \cdot 10^4 \cdot \pi$, $\cedp= \frac{\sqrt{2 \ln (2.5/\delta)}}{\eps}$,
$\savg = \cedp \sqrt{\frac{s}{n} + \frac{s}{t} + \sqrt{\frac{3s}{t} \ln \left(\frac{2}{\delta} \right)}}$, $\gamma = \frac{1}{200 \savg^2}$ and $p = 1 - (1 - \frac{t}{n})^{n/t}$. Suppose that the following conditions hold:
\begin{center}
\begin{enumerate*}[label=(\alph*)]
    \item $n \geq 3t$, $\quad $
    \item $\savg \geq 1$, $\quad $
    \item $r := \sqrt{\frac{3p}{n} \ln(1/\gamma)} \leq \gamma$.
\end{enumerate*}
\end{center}
}
Define $\preg$ and $\pstar$ as in \Cref{eqn:preg-pstar}.
Then
$$\preg - \pstar \geq \frac{1}{100 \sqrt{2 \pi} \cdot \savg^3}.$$
\end{lemma}

\begin{proof}[Proof of \Cref{lem:integral-diff}]
The Taylor series expansion of the Gaussian PDF $\varphi(x, \mu, \sigma^2) = \frac{1}{\sqrt{2 \pi \sigma^2}} \exp \left(-\frac{(x-\mu)^2}{2 \sigma^2}\right)$ is
$$
\varphi(x, \mu, \sigma^2) = \frac{1}{\sqrt{2 \pi \sigma^2}} \sum_{k=0}^\infty \frac{(-1)^k (x-\mu)^{2k}}{k! \cdot 2^k \cdot \sigma^{2k}},
$$
with its first three nonzero Taylor polynomials (with degree $0, 2, 4$) being
\begin{gather*}
    P_0(x, \mu, \sigma^2) := \frac{1}{\sqrt{2 \pi \sigma^2}},
    \condelse{\\}{\qquad}
    P_2(x, \mu, \sigma^2) := \frac{1 - \frac{(x-\mu)^2}{2\sigma^2}}{\sqrt{2 \pi \sigma^2}} , 
    \condelse{\\}{\qquad}
    P_4(x, \mu, \sigma^2) := \frac{1 - \frac{(x-\mu)^2}{2\sigma^2} + \frac{(x-\mu)^4}{8 \sigma^4}}{\sqrt{2 \pi \sigma^2}} .
\end{gather*}
The polynomials $P_0$ and $P_4$ bound the Gaussian PDF above, and $P_2$ bounds it below, i.e.
\begin{gather*}
    P_0(x, \mu, \sigma^2) \geq \varphi(x, \mu, \sigma^2), 
    \condelse{\\}{\qquad} 
    \varphi(x, \mu, \sigma^2) \geq P_2(x, \mu, \sigma^2), 
    \condelse{\\}{\qquad} 
    P_4(x, \mu, \sigma^2) \geq \varphi(x, \mu, \sigma^2).
\end{gather*}
We express
\begin{align*}
\cond{&}
\preg - \pstar 
\cond{\\&\qquad}
= (1-\gamma) \int_r^{1-r} \varphi(x, p, \savg^2) \;\mathrm dx 
\cond{\\&\qquad \quad}
- \int_{0}^{1} \left((1-p)\varphi(x, t/n, \savg^2) + p \cdot \varphi(x, 1, \savg^2) \right) \;\mathrm dx
\cond{\\&\qquad}
= A - B,
\end{align*}
where we define
\begin{align*}
    A &=  (1-\gamma) \int_r^{1-r} \varphi(x, p, \savg^2) \mathrm dx 
    \cond{\\&\quad}
    - \int_r^{1-r} \cond{\!\!\!} \left((1-p) \cdot \varphi(x, t/n, \savg^2) + p \cdot \varphi(x, 1, \savg^2) \right) \mathrm dx, \\
    B &= \int_{0}^r \left((1-p) \cdot \varphi(x, t/n, \savg^2) + p \cdot \varphi(x, 1, \savg^2) \right) \mathrm dx 
    \cond{\\&\quad}
    + \int_{1-r}^{1} \left((1-p) \cdot \varphi(x, t/n, \savg^2) + p \cdot \varphi(x, 1, \savg^2) \right) \mathrm dx.
\end{align*}
Lower bounding $A$ by integrating the respective Taylor polynomials of the Gaussian PDFs, we get
\begin{align*}
A 
&= (1-\gamma) \int_r^{1-r} \varphi(x, p, \savg^2) \mathrm dx 
    \cond{\\&\qquad}
    - (1-p) \int_r^{1-r} \varphi(x, t/n, \savg^2) \mathrm dx 
    \cond{\\&\qquad}
    - p \int_r^{1-r} \varphi(x, 1, \savg^2) \mathrm dx \\
&\geq (1- \gamma) \int_r^{1-r} P_2(x, p, \savg^2) \mathrm dx 
    \cond{\\&\qquad}
    - (1-p) \int_r^{1-r} P_4(x, t/n, \savg^2)  \mathrm dx 
    \cond{\\&\qquad}
    - p \int_r^{1-r} P_4(x, 1, \savg^2)  \mathrm dx \\
&= \frac{1}{\sqrt{2 \pi \savg^2}} \left[ (1\cond{\!}-\cond{\!}\gamma)\cond{\!}\left(\cond{\!}(1-2r) - \frac{(1\cond{\!}-\cond{\!}r\cond{\!}-\cond{\!}p)^3+(p-r)^3}{6 \savg^2}\right) \right.\\
    &\qquad\qquad\qquad
    - (1\cond{\!}-\cond{\!}p) \cond{\!} \left(\cond{\!} (1\cond{\!}-\cond{\!}2r) - \frac{(1\cond{\!}-\cond{\!}r\cond{\!}-\cond{\!}\frac{t}{n})^3 + (\frac{t}{n} - r)^3}{6 \savg^2}\right. 
    \cond{\\&\qquad\qquad\qquad\qquad\qquad\qquad}
    + \left.\frac{\left(1-r-\frac{t}{n}\right)^5 + \left(\frac{t}{n} - r\right)^5}{40 \savg^4}\right)\\
    &\left.\qquad\condelse{\quad\;\;}{\qquad\qquad} 
    - p \left((1-2r) - \frac{(1-r)^3-r^3}{6\savg^2} + \frac{(1-r)^5-r^5}{40 \savg^4} \right)\right] \\
&\geq \frac{1}{\sqrt{2 \pi \savg^2}} \cond{\!} \left[-\gamma \cond{\!} +\frac{p((1\cond{\!}-\cond{\!}\gamma)^3\cond{\!}-\cond{\!}\gamma^3)- (1\cond{\!}-\cond{\!}\gamma)((1\cond{\!}-\cond{\!}p)^3\cond{\!}+\cond{\!} p^3)}{6 \savg^2} \right.
    \cond{\\&\qquad\qquad\qquad}
    \left.
    - \frac{(1-p)((1-\frac{t}{n})^5+(\frac{t}{n})^5)+p}{40\savg^4}\right] \tag*{\text{using $0 \leq r \leq \gamma$ (condition (c))}} \\
&\geq \frac{1}{\sqrt{2 \pi \savg^2}} \left[ \- \gamma\cond{\!} +\cond{\!} \frac{p((1\cond{\!}-\cond{\!}\gamma)^3\cond{\!}-\cond{\!}\gamma^3) - (1\cond{\!}-\cond{\!}\gamma)((1-p)^3+p^3)}{6 \savg^2} \right.
    \cond{\\&\qquad\qquad\qquad}
    \left.
    - \frac{1-p}{40 \savg^4} - \frac{p}{40 \savg^4} \right] \tag*{\text{$n \geq 3t$ gives $0 \leq (1-\frac{t}{n})^5 + (\frac{t}{n})^5 \leq 1$}} \\
&= \frac{1}{\sqrt{2 \pi \savg^2}} \left[-\frac{1}{200 \savg^2} + \frac{p((1-\frac{1}{200\savg^2})^3-(\frac{1}{200\savg^2})^3)}{6 \savg^2}
    \right.
    \cond{\\&\qquad\qquad\qquad}
    \left.-\frac{(1-\frac{1}{200\savg^2})((1-p)^3+p^3)}{6 \savg^2} - \frac{1}{40\savg^4}\right] \tag*{\text{setting $\gamma = \frac{1}{200 \savg^2}$}}\\
&\geq \frac{1}{\sqrt{2\pi \savg^2}} \left[\frac{1}{50 \savg^2}\right] = \frac{1}{50 \sqrt{2 \pi} \savg^3},
\end{align*}
where the last inequality holds for all $\savg \geq 1$ (which holds by condition (b)) for $p = 1 - (1-\frac{t}{n})^{n/t} \in [\frac{6}{10}, \frac{7}{10}]$, which occurs when $n \geq 3t$.
Lastly, we upper bound $B$ by using the first Taylor polynomial of the Gaussian PDF, giving
\begin{align*}
B &= (1-p)\int_{0}^r \varphi(x, t/n, \savg^2) \; \mathrm dx 
    + p \int_{0}^r \varphi(x, 1, \savg^2)\;\mathrm dx 
    \cond{\\&\quad}
    + (1-p) \int_{1-r}^{1} \cond{\!\!}\cond{\!\!} \varphi(x, t/n, \savg^2) \; \mathrm dx + p \int_{1-r}^{1} \cond{\!\!}\cond{\!\!} \varphi(x, 1, \savg^2) \;\mathrm dx \\
&\leq (1-p)\int_{0}^r \cond{\!\!} P_0(x, t/n, \savg^2) \; \mathrm dx + p \int_{0}^r P_0(x, 1, \savg^2)\;\mathrm dx 
    \cond{\\&\quad}
    + (1\cond{\!}-\cond{\!}p) \cond{\!} \int_{1-r}^{1} \cond{\!\!}\cond{\!\!} P_0(x, t/n, \savg^2) \; \mathrm dx + p \cond{\!} \int_{1-r}^{1} \cond{\!\!}\cond{\!\!\!} P_0(x, 1, \savg^2) \;\mathrm dx \\
&= \frac{r}{\sqrt{2 \pi \savg^2}}((1-p)+p+(1-p)+p) 
\cond{\\&}
= \frac{2r}{\sqrt{2 \pi} \savg} \leq \frac{1}{100 \sqrt{2 \pi} \savg^3},
\end{align*}
where the last inequality comes from $r \leq \gamma = \frac{1}{200 \savg^2}$.
We conclude that
\begin{align*}
        \preg - \pstar
    &= A - B \geq \frac{1}{50 \sqrt{2 \pi} \savg^3} - \frac{1}{100 \sqrt{2 \pi} \savg^3}
    \cond{\\&}
    = \frac{1}{100 \sqrt{2 \pi} \savg^3}. \qedhere
\end{align*}
\end{proof}

\section{
Counting Edges under \texorpdfstring{\LNDP}{Noninteractive LNDP}
}\label{sec:soft-threshold}
Counting edges in graphs can be formulated as a linear query about the degree distribution\cw{, and thus follows from \cref{thm:fact-mech-blurry}}. For this important special case of our general algorithmic framework from \Cref{sec:blur-new}, our algorithm can be simplified significantly.\footnote{An edge counting algorithm \cw{with the error given in \cref{thm:edge-ct-alg}} can be obtained by using \Cref{alg:matrix-local}, 
\cw{setting $s=\max\set{\sqrt{n},D}$, and asking the linear query $(0, ns/2, \ldots, ns/2)$.}

}
Here we present a self-contained version of the algorithm that does not rely on \Cref{sec:blur-new}.
Our \LNDP algorithm for counting edges (\cref{alg:edge-ct-new}) %
achieves optimal error for sparse graphs (i.e., graphs with maximum degree $D \leq \sqrt{n}$), with privacy guarantees holding unconditionally for all graphs.

The algorithm proceeds as follows.
First, each node reports %
a clipped version of its degree with added Gaussian noise. The central server then sums these noisy (clipped) degrees and divides them by two, yielding an unbiased estimate for the edge count %
when $G$ is $D$-bounded (i.e., all nodes have degree at most $D$). %

\begin{algorithm}[ht!]
    \caption{$\algedges$ for privately counting edges.}
    \label{alg:edge-ct-new}
    
    \begin{algorithmic}[1]
        \Statex \textbf{Parameters:} Privacy parameters $\eps > 0, \del\in(0,1]$, maximum degree $D\in \N$, number of nodes $n\in\N$.
        \Statex \textbf{Input:}
        Graph $G$ on node set $[n]$.
        \Statex \textbf{Output:} Edge count estimate $\widehat{m} \in \R$.
        
        \For{{\bf all} nodes $i\in[n]$} \Comment{\commentstyle{Define $d_i$ as the degree of node $i$.}}
            \State Node $i$ releases $y_i \gets \min \{d_i, D\} + Z_i$ to the central server, where  $Z_i\sim \cN \left(0, \left(D^2 + n \right) \cdot \frac{2\ln(1.25/\del)}{\eps^2} \right).$
            \label{line:edge-ct-local-rand-ans}
        \EndFor
        \State The central server returns $\widehat{m} = \frac{1}{2} \sum_{i \in [n]} y_i$. \label{line:edge-ct-ans}
    \end{algorithmic}
\end{algorithm}

\begin{theorem}
\label{thm:edge-ct-alg}
    Let $\algedges$ be \cref{alg:edge-ct-new} with parameters
    $\eps > 0, \del\in(0,1]$, 
    $n\in\N$, and maximum degree parameter $D\in\N$.
    Then $\algedges$ is \edLNDP. Moreover, if the input graph $G$ is $D$-bounded, %
    then 
    \begin{conddisplay}
        \E \bracks{|\algedges(G) - m|} = O\Bparen{  \paren{D\sqrt{n} + n}
    \cdot \frac{\sqrt{\log(1/\del)}}{\eps} },
    \end{conddisplay}
    where $m$ is the number of edges in $G$.
\end{theorem}

\begin{proof}[Proof of \cref{thm:edge-ct-alg}]

    \textit{(Accuracy.)}
    Let $G$ be a $D$-bounded graph on node set $[n]$, and let $d_i$ denote the degree of node $i \in [n]$. Note that $\min \{d_i, D\} = d_i$ for all $i \in [n]$, since $G$ is $D$-bounded. We write the output of the central server as
    \begin{align*}
        \hat{m} &= \frac{1}{2} \sum_{i =1}^n y_i = \frac{1}{2}\sum_{i=1}^n (\min \{d_i, D\} + Z_i) 
        \cond{\\&}
        = \frac{1}{2} \sum_{i=1}^n (d_i + Z_i) = m + \frac{1}{2}Z,
    \end{align*}
    where $m = \frac{1}{2} \sum_{i=1}^n d_i$ and 
    \begin{conddisplay}
        Z = \sum_{i=1}^n Z_i \sim \cN \left(0, (D^2 + n) \cdot \frac{2\ln(1.25/\del)}{\eps^2}\right).
    \end{conddisplay}
    This implies that
    \begin{align*}
        \E [|\algedges(G) - m|] &= \frac{\E\left[|Z|\right]}{2} 
        \cond{\\&}
        = O \left(\sqrt{(D^2n + n^2) \cdot \frac{\log(1/\del)}{\eps^2}}\right) 
        \cond{\\&}
        = O \left((D \sqrt{n} + n) \cdot \frac{\sqrt{\log(1/\del)}}{\eps}\right).
    \end{align*}

    \textit{(Privacy.)}  By \Cref{def:lndp}, to show that the algorithm is $(\eps, \del)$-\LNDP, 
    it suffices to show that releasing the vector of randomizer outputs $\vec y = (y_1, \ldots, y_n)$ is $(\eps, \del)$-DP. %
    Writing $y_i = x_i + Z_i$ where $x_i = \min \{d_i, D\}$ and $Z_i$ is as on Line~\ref{line:edge-ct-local-rand-ans}, by the privacy of the Gaussian mechanism (\Cref{lem:gaussianmech}) it suffices to show that the $\ell_2$-sensitivity $\Delta_2$ of $\hat x = (x_1, \ldots, x_n)$ %
    is at most $\sqrt{D^2 + n}$. 
    
    Let $G$ and $G'$ be two node-neighboring graphs differing in edges incident to node $i^*$. Let $\vec{x}'$ and $d_i'$ be defined analogously for $G'$ as above. %
   Then
    \begin{align*}
        \|\vec x - \vec{x}'\|_2^2 &= \sum_{i=1}^n |\min\{d_i, D\} - \min \{d_i', D\}|^2 \\
    &= |\min\{d_{i^*}, D\} - \min \{d'_{i^*}, D\}|^2 \\
    & \ \ \ + \sum_{i \neq i^*} |\min\{d_i, D\} - \min \{d_i', D\}|^2 \\
    &\leq D^2 + (n-1) \cdot 1^2 \leq D^2 + n.
    \end{align*}
    Since this bound holds for all node-neighboring graphs $G \sim G'$, we have $\Delta_2 = \max_{G \sim G'} \|\vec x - \vec{x}'\|_2^2 \leq \sqrt{D^2 + n}$, completing the proof.
\end{proof}

\section{Baseline \texorpdfstring{\LNDP}{Noninteractive LNDP} Algorithms}
\label{sec:basic-algs}

In this section we describe two natural \LNDP algorithms for counting edges: the first adds Laplace noise to each node's degree (\cref{sec:laplace-mech-algo}), and the second releases each bit in the adjacency matrix using randomized response (\cref{sec:rr-edge-count}).
Both algorithms count edges with expected additive error $O_\delta\bparen{\frac{n\sqrt{n}}{\eps}}$.

\subsection{Laplace Mechanism-Based Edge Counting}
\label{sec:laplace-mech-algo}

We prove the privacy and accuracy of a natural approach for counting edges with \ezLNDP, where each node adds Laplace noise with scale $\Theta\bparen{\frac{n}{\eps}}$ to its degree. The resulting algorithm counts edges with expected additive error $O\bparen{\frac{n\sqrt{n}}{\eps}}$.

\begin{algorithm}[ht!]
    \caption{$\alglaplace$ for privately counting edges.}
    \label{alg:laplace-count}
    
    \begin{algorithmic}[1]
        \Statex \textbf{Parameters:}
        Privacy parameter $\eps > 0$;
        number of nodes $n\in\N$.

        \Statex \textbf{Input:}
        Graph $G$ on node set $[n]$.
        \Statex \textbf{Output:} Edge count estimate $\widehat{m} \in \R$.
        
        \For{{\bf all} nodes $i\in[n]$}
            \State Node $i$ draws $\displaystyle Z_i \sim \Lap \paren{ \frac{2n}{\eps} }$.
            \State Node $i$ sends $y_i \gets d_i + Z_i$ to the central server.
            \cond{\Statex}
            \Comment{\commentstyle{$d_i$ denotes the degree of node $i$.}}
        \EndFor

        \State The central server returns $\widehat{m} = \frac{1}{2} \sum_{i\in[n]} y_i$. %
    \end{algorithmic}
\end{algorithm}

\begin{theorem}[Laplace edge counting]
\label{thm:laplace-edge-count}
Let $\eps > 0$ and $n \in \mathbb{N}$. Let $\alglaplace$ be \Cref{alg:laplace-count} with privacy parameter $\eps$. Then $\alglaplace$ is \ezLNDP, and, moreover, 
for every input graph where $m$ denotes its edge count, we have 
$\Ex\bracks{\abs{\hat m - m }} = O\Bparen{\frac{n\sqrt{n}}{\eps}}.$
\end{theorem}
\begin{proof}
    We separately analyze privacy and accuracy of \Cref{alg:laplace-count}.

    \textit{(Privacy.)} %
    Each
    randomizer $\cR_i$ satisfies $\cR_i(d_i) = d_i + Z_i$ where $Z_i \sim \Lap(\frac{2n}{\eps})$.
    Each node's message to the server can be parallelized, meaning the algorithm is noninteractive.
    Additionally, the vector of degrees $(d_1, \ldots, d_n)$ has $\ell_1$ sensitivity at most $2n$, so by the privacy of the Laplace mechanism (\cref{lem:Laplacemech}) the algorithm is \ezLNDP.

    \textit{(Accuracy.)} Let $G$ be an undirected graph on $n$ nodes, and consider running the algorithm $\alglaplace$ on $G$ with privacy parameters as in the theorem statement. The estimate $\hat m$ returned by the central server can be written as
    \begin{align*}
        \hat{m} = \frac{1}{2} \sum_{i \in [n]} y_i = \frac{1}{2} \sum_{i \in [n]} (d_i + Z_i) &= \frac{1}{2} \sum_{i \in [n]} d_i + \frac{1}{2} \sum_{i \in [n]} Z_i 
        \cond{\\ &}
        = m + \frac{1}{2} Z,
    \end{align*}
    where $m = \frac{1}{2} \sum_{i \in [n]} d_i$ is the number of edges in $G$, and $Z = \sum_{i \in [n]} Z_i$.
    
    By a standard concentration bound for sums of Laplace random variables (e.g., \cite[Lemma 2.8]{ChanSS11}),
    \[
        \Pr \left[|Z|
        \geq
        \frac{2n \sqrt{n \log(1/\beta)}}{\eps} \right] \leq 2 \beta,
    \]
    so setting $\beta = \frac{1}{6}$ and
    $\alpha = \frac{n \sqrt{n \log(6)}}{\eps} = O \bparen{\frac{n \sqrt{n}}{\eps} }$
    gives $\Pr \left[|\hat m - m| \geq \alpha\right] \leq \frac{1}{3}$.
\end{proof}

\subsection{Randomized Response-Based Edge Counting}
\label{sec:rr-edge-count}

In this section, we state and analyze another natural algorithm for edge counting based on randomized response \cite{Warner65}, where each node applies randomized response to the presence of each of their edges. The outputs are then debiased and aggregated by the central server to produce an edge count with additive error $O \Bparen{n\sqrt{n}\cdot {\smallepsdel}}$ in the case of \edLNDP.
We use the standard definition of randomized response provided in \cref{lem:rr-alg}.

\begin{algorithm}[ht!]
    \caption{$\algrr$ for privately counting edges.}
    \label{alg:rand-response-count}
    
    \begin{algorithmic}[1]
    \label{algorithmic-rand-response-count}
        \Statex \textbf{Parameters:}
        Privacy parameter $\eps > 0$;
        number of nodes $n\in\N$.

        \Statex \textbf{Input:}
        Graph $G$ on node set $[n]$ represented by an $n \times n$ adjacency matrix $A = [a_{i,j}]$.
        \Statex \textbf{Output:} Edge count estimate $\widehat{m} \in \R$.
        
        \For{{\bf all} nodes $i\in[n]$}
            \State Node $i$ sends
                \cond{\Statex \qquad}
                $(b_{i,i+1}, \ldots, b_{i,n}) \gets (\rr(a_{i,i+1}), \ldots, \rr(a_{i,n}))$. \label{step:rr-edgecount-vector}
        \EndFor
        \State For all $i < j \in [n]$, set $\hat{a}_{i,j} \gets \frac{b_{i,j}(e^\eps+1)-1}{e^\eps - 1}$.
        \State \textbf{return} $\widehat{m} = \sum_{i< j \in [n]} \hat{a}_{i,j}$.
    \end{algorithmic}
\end{algorithm}

\begin{theorem}[Randomized response-based edge counting]
\label{thm:rr-edge-count}
Let $\eps \in (0, 1)$, $\delta \in\left(0,\frac{1}{2}\right]$, and $n \in \mathbb{N}$. Let $\algrr$ be \Cref{alg:rand-response-count} with privacy parameter $\eps' = \frac{\eps}{\sqrt{8 n \log(1/\delta)}}$. Then, $\algrr$ is \edLNDP, and, moreover, 
    for every input graph where $m$ denotes its edge count, we have $\Ex\bracks{\abs{\hat m - m }} = O\Bparen{\frac{n\sqrt{n\ln(1/\delta)}}{\eps}}.$

\end{theorem}
\begin{proof}
\textit{(Privacy.)}
Note that each node $i$ outputs once independently of other nodes' outputs. By \cref{lem:rr-alg}, each call to $\rr(\cdot)$ is $\eps'$-DP. Let $G$ be a graph on node set $[n]$, and let $G'$ be obtained by changing (some) edges incident to node $i^*$.
Consider the corresponding adjacency matrices $A_G = [a_{i,j}]$ and $A_{G'} = [a'_{i,j}]$. For all $i, j \in [n]$ such that $i^* \notin \{i, j\}$, we have $a_{i,j} = a'_{i,j}$, so $\rr(a_{i,j}) = \rr(a'_{i,j})$. In the upper triangular matrices given by $[a_{i,j}]_{i<j \in [n]}$ and $[a'_{i,j}]_{i<j \in [n]}$, we have at most $n$ entries that are different.
Thus, by advanced composition (\Cref{lem:adv-comp}), $\algrr$ is $(\tilde{\eps},\delta)$-\LNDP, where  
\(
\tilde{\eps}
= \eps'\sqrt{2n\ln(1/\delta)} + n\eps'\frac{e^{\eps'}-1}{e^{\eps'}+1}.
\)
Substituting $\eps' = \frac{\eps}{\sqrt{8n\ln(1/\delta)}}$ and using 
$\frac{e^x - 1}{e^x + 1} \le x/2$ for $x \ge 0$ gives  
\(
\tilde{\eps} \le \frac{\eps}{2} + \frac{\eps^2}{16\ln(1/\delta)} \le \eps,
\)
for all $\delta \in \left(0, \frac{1}{2}\right]$. 
Hence, $\algrr$ satisfies \edLNDP.

\textit{(Accuracy.)}
By the standard accuracy guarantees for randomized response, each value $\hat a_{i,j}$ is unbiased and has variance $\frac{e^{\eps'}}{(e^{\eps'}-1)^2}$. Summing over $\binom{n}{2}$ independent pairs gives
\(
\Var[\hat m]=\binom{n}{2}\frac{e^{\eps'}}{(e^{\eps'}-1)^2},
\)
and with $\eps'=\frac{\eps}{\sqrt{8n\ln(1/\delta)}}$ this equals $\Theta\!\big(\frac{n^3\ln(1/\delta)}{\eps^2}\big)$. Therefore, by Cauchy--Schwarz, where we let $m$ denote the true edge count of the input graph, we have
\(
\Ex\bracks{\abs{\hat m - m }} = O\bparen{\frac{n\sqrt{n\ln(1/\delta)}}{\eps}}.
\)
\end{proof}

\end{document}